\newcommand{\vlong }[1]{#1}
\newcommand{\vshort}[1]{}
\documentclass[a4paper,UKenglish]{lipics-v2019}

\usepackage[utf8]{inputenc}
\usepackage{amsmath}
\usepackage{amsthm}
\usepackage{amssymb}
\usepackage{mathrsfs}

\usepackage{bm}
\usepackage{stmaryrd}
\usepackage{textcomp}
\usepackage{proof}
\usepackage{color}
\usepackage{xcolor}
\usepackage{xparse}
\usepackage[]{hyperref}

\hypersetup{allbordercolors=black!2,urlbordercolor=cyan,citebordercolor=green,pdfborder=0 0 0.5}

\usepackage{tikz}
\usepackage{cleveref}
\usepackage{xparse}
\usepackage{dsfont}
\usepackage{cmll}
\usepackage{tcolorbox}
\newtcbox{\resbox}{nobeforeafter,tcbox raise base,boxrule=0.4pt,top=0mm,bottom=0mm,
  right=0mm,left=0mm,arc=1pt,boxsep=2pt,before upper={\vphantom{dlg}},
  colframe=green!50!black,coltext=green!25!black,colback=green!8!white,
  overlay={}
    }
\newtcbox{\blabox}{nobeforeafter,tcbox raise base,boxrule=0.4pt,top=0mm,bottom=0mm,
  right=0mm,left=0mm,arc=2pt,boxsep=2pt,before upper={\vphantom{dlg}},
  colframe=black!50,coltext=black,colback=white,
  overlay={}
    }    

\robustify{\resbox}
\robustify{\blabox}

\usepackage{listings}
\usepackage{lstcoq}
\usepackage[countsame]{coqtheorem}
\usepackage{perfectcut}
\usepackage{multicol}
\usetikzlibrary{shapes.arrows,shapes.geometric,chains,decorations.pathreplacing,,decorations.text,positioning,arrows,automata,fadings,shadows}
\usepackage[colorinlistoftodos,prependcaption,textsize=scriptsize,textwidth=1.2cm]{todonotes}
\newcommand{\etienne}[1]{\todo[color=yellow]{#1}}
\renewcommand{\etienne}[1]{}

\usepackage{newunicodechar}
\newunicodechar{⊗}{\tensor}
\newunicodechar{¬}{\neg}
\newunicodechar{∈}{\in}
\newunicodechar{ʆ}{\sep}

\newcommand \seq[2]{\shortstack{$#1$ \\ \mbox{}\\
                    \mbox{}\hrulefill\mbox{}\\ \mbox{}\\ $#2$}}
                    
\renewcommand{\seq}[2]{\infer{#2}{#1}}

\newcommand{\tmu}{{\tilde\mu}}
\newcommand{\cut}[2]{\perfectcut{#1}{#2}}
\newcommand{\imp}{\rightarrow}
\newcommand{\defeq}{\triangleq}
\newcommand{\automath}[1]{\relax\ifmmode{#1}\else{$#1$}\fi}

\newcommand{\id}{\textrm{id}}

\makeatletter
\newcommand{\oset}[3][0ex]{%
  \mathrel{\mathop{#3}\limits^{
    \vbox to#1{\kern-2\ex@
    \hbox{#2}\vss}}}}
\makeatother
\newcommand{\red}{\rightarrow}

\newcommand{\negt}{{\bot\!\!\!\!\!\bot}}

\newcommand{\C}{\mathcal{C}}
\newcommand{\V}{\mathcal{V}}
\renewcommand{\S}{\mathcal{S}}
\newcommand{\pole}{{\bot\!\!\!\bot}}
\newcommand{\orth}{{\pole}}
\newcommand{\prfcasetext}{Case}
\newcommand{\prfcase}[1]{\paragraph*{\normalfont \textbf{\prfcasetext} #1.}}
\newcommand{\prfcases}[1]{\paragraph*{\normalfont \textbf{{\prfcasetext}s} #1.}}

\newcommand{\tvn}[2]{|#2|_{#1}}

\newcommand{\tvv}[1]{\tvn{v}{#1}}

\newcommand{\real}{\Vdash}

\newcommand{\tr}[1]{\llbracket #1\rrbracket}

\newcommand{\union}[2]{#1\sqcup #2}

\newcommand{\lcut}[1]{{\mbox{$\langle$}} {#1} {|}}
\newcommand{\rcut}[1]{{|} {#1} {\mbox{$\rangle$}}}

\newcounter{rep}

\makeatletter

\newtheorem{repthm}[rep]{Theorem}{\bfseries}{\rmfamily}
 {\end{repthm}}

\newtheorem{replm}[rep]{Lemma}{\bfseries}{\rmfamily}
 {\end{replm}}

\newtheorem{repprop}[rep]{Proposition}{\bfseries}{\rmfamily}
 {\end{repprop}}
 
\newtheorem{reprmk}[rep]{Remark}{\bfseries}{\rmfamily}
 {\end{reprmk}} 
 
\newtheorem{repex}[rep]{Example}{\bfseries}{\rmfamily}
\newenvironment{repexample}[1]{%
\setcounter{rep}{#1}
\addtocounter{rep}{-1}
 \begin{repex}}%
 {\end{repex}}  

\newtheorem{repdef}[rep]{Definition}{\bfseries}{\rmfamily}
 {\end{repdef}}

\newtheorem{basetheorem    }[theorem]{Theorem}{\bfseries}{\rmfamily}
\newtheorem{baselemma      }[theorem]{Lemma}{\bfseries}{\rmfamily}
{\bfseries}{\rmfamily}
\newtheorem{baseremark     }[theorem]{Remark}{\bfseries}{\rmfamily}
\newtheorem{baseexample    }[theorem]{Remark}{\bfseries}{\rmfamily}
\newtheorem{basedefinition }[theorem]{Definition}{\bfseries}{\rmfamily}

\makeatother

\newenvironment{mycenter}{%
  \setlength\topsep{1pt}
  \setlength\parskip{1pt}
  \begin{center}
}{%
  \end{center}
}

\newcommand{\lmmt}{\automath{\lambda\mu\tmu}}

\newcommand{\myhref}[2]{\href{#1}{#2}} 
\newcommand{\trth}{\texttt{tr}}

\newcommand{\minf}{\bigcurlywedge}
\newcommand{\msup}{\bigcurlyvee}

\renewcommand{\sup}{\curlyvee}
\newcommand{\myleq}{\preccurlyeq}
\newcommand{\mygeq}{\succcurlyeq}

\let\oldleq\leq

\renewcommand{\leq}{\myleq}
\renewcommand{\geq}{\mygeq}

\newcommand{\revord}{\triangleleft}

\newcommand{\tensor}{\otimes}

\newcommand{\ar}{\rightarrow}
\newcommand{\sep}{\mathcal{S}}
\newcommand{\psep}{\mathcal{S}^\parr}
\newcommand{\tsep}{\mathcal{S}^\tensor}

\makeatletter
\newcommand{\uset}[3][0ex]{%
  \mathrel{\mathop{#3}\limits^{
    \vbox to #1{\kern-2\ex@
    \hbox{\tiny $#2$}\vss}}}}
\makeatother

\newcommand{\art}{\uset{\text{\tiny \!\!$\tensor$}}{\ar}}
\newcommand{\arp}{\uset{\parr}{\ar}}

\newcommand{\bs}[1]{\bm{s}_{#1}}

\renewcommand{\k}{{\bm k}}

\newcommand{\cc}{\mathsf{\bm{cc}}}

\newcommand{\ps}[1]{\bs{#1}^{\scriptscriptstyle \! \parr}}
\newcommand{\ts}[1]{\bs{#1}^{\scriptscriptstyle  \tensor}}

\renewcommand{\L}{\text{L}}
\newcommand{\T}{\mathcal{T}}

\newcommand{\fval}[1]{\|#1\|}
\newcommand{\fv}[1]{\fval{#1}}
\newcommand{\tval}[1]{|#1|}
\newcommand{\tv}[1]{\tval{#1}}
\newcommand{\fvval}[1]{\|#1\|_V}
\newcommand{\tvval}[1]{|#1|_V}

\NewDocumentEnvironment{xdefinition}{m m o} 
{
 \spnewtheorem{mydef}[theorem]{#1{#2}{Definition}}{\bf}{}
 \IfNoValueTF {#3}
  {\begin{mydef}}
  {\begin{mydef}[#3]}
}
{\end{mydef}}

\NewDocumentEnvironment{xexample}{m m o}
{

\spnewtheorem{myex}[theorem]{#1{#2}{Example}}{}{}
 \IfNoValueTF {#3}
  {\begin{myex}}
  {\begin{myex}[#3]}
  }
  {\end{myex}
  }

\NewDocumentEnvironment{xlemma}{m m o}
{
\theoremstyle{plain}
\spnewtheorem{mylem}[theorem]{#1{#2}{Lemma}}{}{}
 \IfNoValueTF {#3}
  {\begin{mylem}}
  {\begin{mylem}[#3]}
  }
  {\end{mylem}}

\NewDocumentEnvironment{xproposition}{m m o}
{
\let\c@myprop\relax
\relax\c@myprop
\@ifundefined{c@myprop}{\spnewtheorem{myprop}[theorem]{#1{#2}{Proposition}}{}{}}
 \IfNoValueTF {#3}
  {\begin{myprop}}
  {\begin{myprop}[#3]}
  }
  {\end{myprop}
  }    

\NewDocumentEnvironment{xcorollary}{m m o}
{
\spnewtheorem{mycor}[theorem]{#1{#2}{Corollary}}{}{}
 \IfNoValueTF {#3}
  {\begin{mycor}}
  {\begin{mycor}[#3]}}
  {\end  {mycor}}

\NewDocumentEnvironment{xtheorem}{m m o}
{
\spnewtheorem{mythm}[theorem]{#1{#2}{Theorem}}{}{}
 \IfNoValueTF {#3}
  {\begin{mythm}}
  {\begin{mythm}[#3]}
  }
  {\end{mythm}
  } 

\NewDocumentEnvironment{xremark}{m m o}
{

\spnewtheorem{myrmk}[theorem]{#1{#2}{Remark}}{}{}
 \IfNoValueTF {#3}
  {\begin{myrmk}}
  {\begin{myrmk}[#3]}
  }
  {\end{myrmk}
  }

\newcommand{\trurl}[2]{  \myhref{https://emiquey.gitlab.io/ImplicativeAlgebras/ImpAlg.Duality.html\##1}{#2}}

\newcommand{\iaurl}[2]{ \myhref{https://emiquey.gitlab.io/ImplicativeAlgebras/ImpAlg.ImplicativeAlgebras.html\##1}{#2}}

\newcommand{\purl}[2]{  \myhref{https://emiquey.gitlab.io/ImplicativeAlgebras/ImpAlg.ParAlgebras.html\##1}{#2}}
\newcommand{\lpurl}[2]{ \myhref{https://emiquey.gitlab.io/ImplicativeAlgebras/ImpAlg.LPar.html\##1}{#2}}

\newcommand{\turl}[2]{  \myhref{https://emiquey.gitlab.io/ImplicativeAlgebras/ImpAlg.TensorAlgebras.html\##1}{#2}}

\newcommand{\lturl}[2]{ \myhref{https://emiquey.gitlab.io/ImplicativeAlgebras/ImpAlg.LTensor.html\##1}{#2}}
\newcommand{\comb}[1]{\textbf{\textsc{#1}}}

\newcommand{\HA}{\text{\bf HA}}

\newcommand{\Set}{\text{\bf Set}}

\newcommand{\op}{^{op}}
\renewcommand{\H}{\mathcal{H}}
\newcommand{\Cop}{\C\op}

\newcommand{\E}{\mathcal{E}}

\renewcommand{\P}{\mathcal{P}}

\newcommand{\noem}{\vspace{-1em}}
\newcommand{\A}{\mathcal{A}}
\newcommand{\B}{\mathcal{B}}
\newcommand{\sleq}{\oldleq_{\H}}
\newcommand{\uleq}{\oldleq_{\usep}}
\renewcommand{\seq}{\cong_{\sep}}

\newcommand{\bigfa }{\mathop{\raisebox{-0.1em}{\mbox{\Large $\mathsurround0pt{\boldmath{\forall}}$}}}} 
\newcommand{\hugefa}{\mathop{\raisebox{-0.3em}{\mbox{\huge $\mathsurround0pt{\boldmath{\forall}}$}}}}
\newcommand{\bigex }{\mathop{\raisebox{-0.1em}{\mbox{\Large $\mathsurround0pt\exists$}}} }
\newcommand{\hugeex}{\mathop{\raisebox{-0.3em}{\mbox{\huge $\mathsurround0pt{\boldmath{\exists}}$}}} }

\newcommand{\Lat}{\mathcal{L}}
\newcommand{\I}{\mathcal{I}}
\renewcommand{\H}{{\mathcal{H}}}

\newcommand{\basearrow}{\longrightarrow}

\newcommand{\bmark}{\beta     }

\newcommand{\bred} {\basearrow_{\bmark}}

\makeatletter
\newcommand{\shorteq}{%
  \settowidth{\@tempdima}{-}
  \resizebox{\@tempdima}{\height}{=}%
}
\makeatother

\newcommand{\myfig}[1]{\framebox{\vbox{#1}}}
\renewcommand{\myfig}[1]{{#1}\hrule~\\[-1.2em]}

\newcommand{\FV}{{FV}}

\newcommand{\nomidem}{\vspace{-0.5em}}

\newcommand{\Prop}{\mathop{\texttt{Prop}}}

\newcommand{\N}{{\mathbb{N}}}
\newcommand{\Fork}{{\pitchfork}}

\let\oldleq\leq

\renewcommand{\leq}{\myleq}
\renewcommand{\geq}{\mygeq}

\renewcommand{\k}{{\bm k}}

\newcommand{\usep}{{\sep[I]}}

\renewcommand{\L}{\text{L}}
\newcommand  {\Lpar}{\automath{\L^{\!\parr}}}
\newcommand{\Ltens}{\automath{\L^{\!\tensor}}}

\newcommand{\autorule}[1]{\relax\ifmmode{\scriptstyle(#1)}\else$(#1)$\fi}

\newcommand{\cutrule}{\autorule{\textsc{Cut} }}

\newcommand{\axrrule}{\autorule{\textsc{Ax}_r}}
\newcommand{\axlrule}{\autorule{\textsc{Ax}_l}}

\newcommand{\falrule}{\autorule{\forall_l         }}
\newcommand{\farrule}{\autorule{\forall_r         }}

\newcommand{\negrrule  }{\autorule{\neg_r   }}
\newcommand{\neglrule  }{\autorule{\neg_l   }}

\newcommand{\cord}{\unlhd}

\makeatletter
\newcommand\footnoteref[1]{\protected@xdef\@thefnmark{\ref{#1}}\@footnotemark}
\makeatother

\renewcommand{\usep}{{\sep[I]}}

\setBaseUrl{{https://emiquey.gitlab.io/ImplicativeAlgebras/}}
\setCoqFilename{ImpAlg.ImplicativeStructures}

\title{Revisiting the Duality of Computation: an Algebraic Analysis of Classical Realizability Models}
\titlerunning{Revisiting the Duality of Computation}
\author{Étienne Miquey}{Équipe Gallinette, INRIA\\ LS2N, Université de Nantes}{etienne.miquey@inria.fr}{}{}
\authorrunning{É. Miquey}
\Copyright{Étienne Miquey}
\vshort{\relatedversion{An extended version of this paper including proofs and further details is available at: \url{https://hal.archives-ouvertes.fr/hal-02305560}.}}
\vlong{\relatedversion{This is an extended version of the following CSL paper: \href{https://doi.org/10.4230/LIPIcs.CSL.2020.30}{10.4230/LIPIcs.CSL.2020.30}.}}

\vshort{
\EventEditors{Maribel Fern\'{a}ndez and Anca Muscholl}
\EventNoEds{2}
\EventLongTitle{28th EACSL Annual Conference on Computer Science Logic (CSL 2020)}
\EventShortTitle{CSL 2020}
\EventAcronym{CSL}
\EventYear{2020}
\EventDate{January 13--16, 2020}
\EventLocation{Barcelona, Spain}
\EventLogo{}
\SeriesVolume{152}
\ArticleNo{30}
}


\usepackage{verbatim}
\newenvironment{myproof}{\comment}{\endcomment}

\newlength{\myl}
\newlength{\mysep}
\newcommand{\twoelt}[2]{
  \settowidth{\myl}{$#1$}
  \setlength{\mysep}{0.5\textwidth}
  \addtolength{\mysep}{-\myl}
  \addtolength{\mysep}{-15pt}
\[
  #1
  \hspace{\mysep}
  #2
\]
}

\vshort{\funding{This research was partially funded by the ANII research project
FCE\_1\_2014\_1\_104800. }
\acknowledgements{The author would like to thank Alexandre Miquel 
to which several ideas in this paper, especially the 
definition of conjunctive separators, should be credited.}}

\nolinenumbers
\begin{document}
\vlong{\hideLIPIcs}

\maketitle

\begin{abstract}
In an impressive series of papers, Krivine showed at the edge of the last decade
how classical realizability provides a surprising technique to build models for classical theories.
In particular, he proved that classical realizability subsumes Cohen's forcing, and even more, gives
rise to unexpected models of set theories. 
Pursuing the algebraic analysis of these models that was first undertaken by Streicher,
Miquel recently proposed to lay the algebraic foundation of classical realizability and forcing 
within new structures which he called \emph{implicative algebras}.
These structures are a generalization of Boolean algebras based on an internal law representing the implication. Notably, implicative algebras
allow for the adequate interpretation of both programs (i.e. proofs) and their types (i.e. formulas) in the same
structure. 

The very definition of implicative algebras takes position on a presentation of logic 
through universal quantification and the implication and, computationally, relies on the call-by-name $\lambda$-calculus.
In this paper, we investigate the relevance of this choice, 
by introducing two similar structures.
On the one hand, we define \emph{disjunctive algebras}, which rely on internal
laws for the negation and the disjunction and which we show to be particular cases of implicative algebras.
On the other hand, we introduce \emph{conjunctive algebras}, which rather put the focus on conjunctions
and on the call-by-value evaluation strategy. We finally show how disjunctive and conjunctive algebras 
algebraically reflect the well-known duality of computation between call-by-name and call-by-value.

\end{abstract}



\section{Introduction}

It is well-known since Griffin's seminal work~\cite{Griffin90} that
a classical Curry-Howard correspondence can be obtained by adding control operators to the $\lambda$-calculus.
Several calculi were born from this idea, amongst which Krivine $\lambda_c$-calculus~\cite{Krivine09},
defined as the $\lambda$-calculus extended with Scheme's \texttt{call/cc} operator (for \emph{call-with-current-continuation}).
Elaborating on this calculus, Krivine's developed in the late 90s the theory of \emph{classical realizability}~\cite{Krivine09},
which is a complete reformulation of its intuitionistic twin.
Originally introduced to analyze the computational content of classical programs,
it turned out that classical realizability also provides interesting semantics for classical theories.
While it was first tailored to Peano second-order arithmetic (\emph{i.e.} second-order type systems),
classical realizability actually scales to more complex classical theories like ZF~\cite{Krivine11},
and gives rise to surprisingly new models.
In particular, its generalizes Cohen's forcing~\cite{Krivine11,Miquel11} and
allows for the direct definition of a model in which neither the continuum hypothesis nor the axiom of choice holds~\cite{Krivine14}.

\subparagraph*{Algebraization of classical realizability}
During the last decade, the algebraic structure of the models that classical realizability induces
has been actively studied.
This line of work was first initiated by Streicher, who proposed the concept of \emph{abstract Krivine structure}~\cite{Streicher13},
followed among others by Ferrer, Frey, Guillermo, Malherbe and Miquel who introduced other structures peculiar to
classical realizability~\cite{FerrerEtAl15,FerrerEtAl17,FerMal17,Frey15,Frey16,VanOosZou16}. 
In addition to the algebraic study of classical realizability models, these works had the interest of building
the bridge with the algebraic structures arising from intuitionistic realizability.
In particular, Streicher showed in \cite{Streicher13} how classical realizability could be analyzed in terms of \emph{triposes}~\cite{Pitts02},
the categorical framework emerging from intuitionistic realizability models, 
while the later work of Ferrer \emph{et al.}~\cite{FerrerEtAl15,FerrerEtAl17}
connected it to Hofstra and Van Oosten's notion of \emph{ordered combinatory algebras}~\cite{HofOos03}.
More recently, Alexandre Miquel introduced the concept of \emph{implicative algebra}~\cite{Miquel17},
which appear to encompass the previous approaches and which we present in this paper.

\subparagraph*{Implicative algebras}
In addition to providing an algebraic framework conducive to the analysis of classical realizability,
an important feature of implicative structures is that they allow us to identify \emph{realizers} (\emph{i.e.} $\lambda$-terms)
and \emph{truth values} (\emph{i.e.} formulas).
Concretely, implicative structures are complete lattices equipped with a binary operation $a\to b$ satisfying properties coming from the logical implication.
As we will see, they indeed allow us to interpret both the formulas and the terms in the same structure.
For instance, the ordering relation $a \oldleq b$ will encompass different intuitions 
depending on whether we regard $a$ and $b$ as formulas or as terms. Namely, $a\oldleq b$ will be given the following meanings:
\begin{itemize}
 \item the formula $a$ is a \emph{subtype} of the formula $b$;
 \item the term $a$ is a \emph{realizer} of the formula $b$;
 \item the realizer $a$ is \emph{more defined} than the realizer $b$.
\end{itemize}
In terms of the Curry-Howard correspondence, this means that we not only identify types with formulas and proofs with programs, 
but \emph{we also identify types and programs}.

\subparagraph*{{Side effects}} 
Following Griffin's discovery on control operators and classical logic,
several works have renewed the observation that 
within the proofs-as-programs correspondence, with side effects
come new reasoning principles~\cite{Krivine03,JabSozTab12,Miquel11b,Herbelin10,JaberEtAl16}.
More generally, it is now clear 
that computational features of a calculus may have consequences on the models \etienne{to check} it induces.
For instance, computational proofs of the axiom of dependent choice
can be obtained by adding a \texttt{quote} instruction~\cite{Krivine03},
using memoisation~\cite{Herbelin12,Miquey18a}
or with a bar recursor~\cite{Krivine16}.
Yet, such choices may also have an impact on the structures of the corresponding
realizability models: 
the non-deterministic operator $\Fork$ is known to make
the model collapse on a forcing situation~\cite{Krivine12}, while
the bar recursor requires some continuity properties~\cite{Krivine16}.

If we start to have a deep understanding of the algebraic structure of classical realizability models,
the algebraic counterpart of side effects on these structures is still unclear.
As a first step towards this problem, it is natural to wonder:
does the choice of an evaluation strategy have algebraic consequences on realizability models?
This paper aims at bringing new tools for addressing this question.

\subparagraph*{Outline of the paper}
We start by recalling the definition of Miquel's implicative algebras and their
main properties in \Cref{s:ia}.
We then introduce the notion of \emph{disjunctive algebras} in \Cref{s:da},
which naturally arises from the negative decomposition of the implication $A\to B = \neg A \parr B$.
We explain how this decomposition induces realizability models based on a call-by-name
fragment of Munch-Maccagnoni's system L~\cite{Munch09}, and we show  that disjunctive algebras are in fact 
particular cases of implicative algebras.
In \Cref{s:ta}, we explore the positive dual decomposition $A\to B = \neg(A \tensor \neg B)$,
which naturally corresponds to a call-by-value fragment of system L. 
We show the corresponding realizability models naturally induce a notion of \emph{conjunctive algebras}.
Finally, in \Cref{s:du} we revisit the well-known duality of computation through this algebraic structures.
In particular, we show how to pass from conjunctive to disjunctive algebras and vice-versa,
while inducing isomorphic triposes.
  ~\\[-0.5em]
\par\emph{\small 
\vshort{Most of the proofs have been formalized in the Coq proof assistant, 
in which case their statements include hyperlinks to
their formalizations\footnote{Available at \url{https://gitlab.com/emiquey/ImplicativeAlgebras/}}.}
\vlong{Proofs, as well as further details on the different constructions presented in this paper, are given in appendices. Most of them have been formalized in the Coq proof assistant, 
in which case their statements include hyperlinks to
their formalizations\footnote{Available at \url{https://gitlab.com/emiquey/ImplicativeAlgebras/}}.}
}

\section{Implicative algebras}
\label{s:ia}

\subsection{Krivine classical realizability in a glimpse}
\label{s:realizability}
We give here an overview of the main characteristics of Krivine realizability and
of the models it induces\footnote{For a detailed introduction on this topic,
we refer the reader to~\cite{Krivine09} or~\cite{these}.}.
Krivine realizability models are usually built above the $\lambda_c$-calculus,
a language of abstract machines including a set of terms $\Lambda$ 
and a set of stacks $\Pi$ (\emph{i.e.} evaluation contexts). Processes $t\star\pi$ in the 
abstract machine are given as pairs of a term $t$ and a stack $\pi$.

%


Krivine realizability interprets a formula~$A$ as a set of
closed terms $|A|\subseteq\Lambda$, called the \emph{truth value}
of~$A$, and whose elements are  called the \emph{realizers} of~$A$.
Unlike in intuitionistic realizability models, 
this set is actually defined by orthogonality to a \emph{falsity value} $\|A\|$
made of stacks, which intuitively represents a set of opponents to the formula $A$.
Realizability models are parameterized
by a {pole} $\pole$, a set of processes in the underlying abstract machine
which somehow plays the role of a referee betweens terms and stacks.
The pole  allows us to define the orthogonal set $X^\orth$ of any falsity value $X\subseteq \Pi$ by:
$ X^\orth ~~\defeq~~ \{t\in\Lambda : \forall \pi\in X, t\star \pi \in\pole\} $.
Valid formulas $A$ are then defined as the ones admitting a proof-like \emph{realizer}\footnote{One specificity
of Krivine classical realizability is that the set of terms
 contains the control operator $\cc$ and continuation constants $\k_\pi$.
 Therefore, to preserve the consistency of the induced models, one has to consider only proof-like
 terms, \textit{i.e.} terms that do not contain any continuations constants see~\cite{Krivine09,these}.} $t\in|A|$.
 
Before defining implicative algebras, 
we would like to draw the reader's attention on an important observation about realizability:
there is an omnipresent lattice structure, which is reminiscent of the concept of subtyping~\cite{Cardelli91}.
Given a realizability model it is indeed always possible to define a semantic notion of subtyping:
$A \oldleq B ~\defeq~  \fv{B}\! \subseteq\! \fv{A}$.
This informally reads as \emph{``$A$ is more precise than $B$''}, in that $A$ admits more opponents than $B$.\linebreak
In this case, the relation $\oldleq$ being induced from (reversed) set inclusions
comes with a richer structure of complete lattice, where the meet $\wedge$ is defined
as a union and the join $\vee $ as an intersection.
In particular, the interpretation of a universal quantifier $\fv{\forall x.A}$
is given by an union $\bigcup_{n\in \N}\fv{A[n/x]}=\minf_{n\in\N}\fv{A[n/x]}$,
while the logical connective $\land$ is interpreted as the type of pairs $\times$
\emph{i.e.} with a computation content.
As such, \emph{realizability} corresponds to the following picture:
\blabox{
$~ \forall = \minf \qquad \land = \times~$	
}.
This is to compare with \emph{forcing}, that can be expressed in terms of Boolean algebras 
where both the universal quantifier and the conjunction are  interpreted by meets without any computational content: 
\blabox{$~\forall = \land = \minf~$}~\cite{BellBVM}.

%
%

\subsection{Implicative algebras}
\label{s:imp_struct}
\emph{Implicative structures} are tailored to represent both the formulas of second-order logic and realizers
arising from Krivine's $\lambda_c$-calculus. 
For their logical facet, they are defined as meet-complete lattices (for the universal quantification)
with an internal binary operation satisfying the properties of the implication:
\setCoqFilename{ImpAlg.ImplicativeStructures}
\begin{definition}[][ImplicativeStructure]
An \emph{implicative structure} is a complete lattice $(\A, \leq)$
equipped with an operation $(a, b) \mapsto (a \to b)$,
such that for all $a, a_0 , b, b_0 \in\A$ and any subset $B \subseteq \A$:
\begin{enumerate}
\item $\text{If  } a_0 \leq a \text{  and  } b  \leq b_0\text{  then  }
(a \to b) \leq (a_0 \to b_0 ).$\qquad~
\item $\minf_{b\in B} (a \to b) = a \to \minf_{b\in B} b $
\end{enumerate}
\end{definition}

It is then immediate to embed any closed formula 
of second-order logic within any implicative structure.
Obviously, any complete Heyting algebra or any complete
Boolean algebra defines an implicative structure
with the canonical arrow.
More interestingly, any ordered combinatory algebras, a structure arising
naturally from realizability~\cite{HofOos03,VanOosten08,Streicher13,FerrerEtAl13},
also induces an implicative structure~\cite{Miquey18b}.
Last but not least, any classical realizability model
induces as expected an implicative structure on
the lattice $(\P(\Pi),\supseteq)$ by considering the arrow defined by\footnote{This is 
actually nothing more than the definition of the falsity value $\fv{A\Rightarrow B}$.}: 
$a \to b \defeq a^\pole \cdot b = \{t \cdot \pi : t \in a^\pole , \pi \in b\}$~(\cite{Miquel17,Miquey18b}.

Interestingly, if any implicative structure $\A$ trivially provides us with  
an embedding of second-order formulas, we can also encode $\lambda$-terms 
with the following \coqlink[app]{definitions}:
\twoelt{
ab\defeq \minf \{c : a \leq b \to c\} 
}{
\lambda f \defeq \minf_{a\in \A} (a\to f(a))
}
In both cases, one can understand the meet as a conjunction of all the possible 
approximations of the desired term.
From now on, we will denote by $t^\A$ (resp. $A^\A$) the interpretation
of the closed $\lambda$-term $t$ (resp. formula $A$).
Notably, these embeddings are at the same time:
\begin{enumerate}
 \coqitem[betarule] \setCoqFilename{ImpAlg.Adequacy}
 Sound with respect to the $\beta$-reduction,
 in the sense that  $(\lambda f) a \leq f(a)$ (and more generally,
 one can show that if $t\red_\beta u$ \coqlink[imp_betared]{implies} $t^\A\leq u^\A$);
 \coqitem[adequacy] Adequate with respect to typing, in the sense that if $t$ is of type $A$, then we have $t^\A \leq A^\A$
(which can reads as \emph{``$t$ realizes $A$''}).
\end{enumerate}
\setCoqFilename{ImpAlg.ImplicativeStructures}
In the case of certain combinators, including Hilbert's combinator \comb{k} and \comb{s},
their interpretations as $\lambda$-term is even equal to the interpretation of their principal types,
that is to say that we have $\textbf{\textsc{k}}^\A =\!\!\! \minf_{a,b\in\A} (a\to b \to a)$ and 
$\textbf{\textsc{s}}^\A =\!\!\! \minf_{a,b,c\in\A} ((a\to b \to c)\to (a\to b) \to a \to c)$.
This justifies the definition $\comb{cc}^\A\defeq\minf_{a,b} (((a\to b ) \to a) \to a) $.

Implicative structure are thus suited to interpret both terms and their types.
To give an account for realizability models, 
one then has to define a notion of validity:
\setCoqFilename{ImpAlg.ImplicativeAlgebras}
\begin{definition}[Separator][separator]
 Let $(\A,\leq,\imp)$ be an implicative structure.
 We call a \emph{separator} over $\A$ any set $\sep\subseteq \A$ such that for all $a,b\in A$, 
 the following conditions hold:\noem
 \begin{multicols}{2}
\begin{enumerate}
  \item If $a\in \sep$ and $a\leq b$, then $b\in \sep$.\qquad~ 

  \item $\comb{k}^\A\in \sep$, and $\comb{s}^\A \in \sep$.\qquad~
  \item If $(a\to b)\in \sep$ and $a\in \sep$, then $b\in \sep$.
 \end{enumerate}
  \end{multicols}
 \noem \noindent A separator $\sep$ is said to be \emph{classical} if $\comb{cc}^\A\in \sep$ and \emph{consistent} if $\bot\notin \sep$.
 We call \emph{implicative algebra} any implicative structure  $(\A,\leq,\to,\sep)$ equipped with a separator $\sep$ over $\A$.
 \end{definition}

 \setCoqFilename{ImpAlg.ImplicativeStructures}



Intuitively, thinking of elements of an implicative structure
as truth values, a separator should be understood as the set
which distinguishes the valid formulas (think of a filter in a Boolean algebra).
Considering the elements as terms, it should rather be viewed as the set of valid realizers.
Indeed, conditions (2) and (3) ensure that all closed $\lambda$-terms are in any 
separator\footnote{The latter indeed implies the closure of separators under application.}.
Reading $a\leq b$ as ``\emph{the formula $a$ is a subtype of the formula $b$}'', condition (2) ensures the validity of semantic subtyping.
Thinking of the ordering as ``\emph{$a$ is a realizer of the formula $b$}'', condition (2) states that if a formula is realized, then it is in the separator.

%
 
 \setCoqFilename{ImpAlg.AKS}
 \begin{example}[][AKS_IA]
 Any \emph{Krivine realizability model} induces an implicative structure $(\A,\leq,\to)$ where $\A=\P(\Pi)$, $a\leq b \Leftrightarrow a\supseteq b$ and $a\to b = a^\pole\cdot b$.
  The set of realized formulas, namely $\sep = \{a\in\A: \exists t \in a^\pole,~ t \text{~proof-like} \}$, defines a valid separator~\cite{Miquel17}. 
 \end{example}

\subsection{Internal logic \& implicative tripos}
\setCoqFilename{ImplicativeAlgebras}
In order to study the internal logic of implicative algebras,
we define an\iaurl{entails}{\emph{entailment}} relation:
we say that $a$ {entails} $b$ and we write $a\vdash_{\sep} b$ if $a\to b \in \sep$.
This relation induces a preorder on $\A$.
Then, by defining products $a\times b$ and sums $a+b$
through their usual impredicative encodings in System F\footnote{That is to say
that we define $a\times b \defeq\!\! \minf_{c\in\A}\! ((a\to b \to c)\to c)$ 
and $a + b \defeq\!\! \minf_{c\in\A}\! ((a\to c) \to ( b \to c)\to c)$.
},
we recover a structure of pre-Heyting algebra with respect
to the entailment relation:
 $a \vdash_{\sep}  b \to c ~~ \text{\iaurl{ha_adjunction}{if and only if} }~~ a\times b \vdash_{\sep} c$.

In order to recover a Heyting algebra, it suffices to 
consider the quotient $\H=\A/_{\seq}$ by the equivalence relation $\seq$ induced by $\vdash_\sep$,
which is naturally equipped with an order relation:
$[a] \sleq [b] ~\defeq~ a \vdash_{\sep}  b$ (where we write $[a]$ for the equivalence class of $a\in\A$).
Likewise, we can extend the product, the sum and the arrow to equivalences classes 
to obtain a Heyting algebra  $(\H,\sleq,\land_\H,\lor_\H,\to_\H)$.

\renewcommand{\Prop}{\text{\normalfont Prop}}
\renewcommand{\trth}{\text{\normalfont tr}}

Given any implicative algebra, we can define construction of the implicative tripos is quite similar.
Recall that a (set-based) \emph{tripos} is a first-order hyperdoctrine $\T:\Set\op\to\HA$ 
 which admits a generic predicate\vlong{\footnote{\label{fn:tripos}See Appendix~\ref{a:imp_trip} for
more details on the tripos construction.}}.
To define a tripos, we roughly consider the functor of the
form $I\in\Set\op\mapsto \A^I$.
Again, to recover a Heyting algebra we quotient the product $\A^I$ (which defines an implicative structure)
by the \emph{uniform separator}  $\sep[I]$ defined by:
\[\sep[I] \defeq \{a\in\A^I: \exists s\in S.\forall i\in I. s\leq a_i\} \]

\begin{theorem}[Implicative tripos~\cite{Miquel17}]
 Let $(\A,\leq,\to,\sep)$ be an implicative algebra. The following functor 
 (where $f:J\to I$) defines a tripos:
 \twoelt{
 \T :  I\mapsto \A^I/\usep
 }{
 \T(f) :\left\{\begin{array}{lcl}
		 \A^I/\usep &\to& \A^J/\sep[J]\\ [0.5em]
		 \left[(a_i)_{i\in I}\right]& \mapsto &[(a_{f(j)})_{j\in J}]
		\end{array}
		\right.
}
\end{theorem}

Observe that we could also quotient the product $\A^I$
by the separator product $\S^I$.
Actually, the quotient $\A^I/{\S^I}$ is in bijection
with $(\A/\S)^\I$, and in the case where $\S$ is a classical separator,
$\A/{\S}$ is actually a Boolean algebra, so that the 
product $(\A/\S)^\I$ is nothing more than a Boolean-valued model (as in the case of forcing).
Since $\sep[I] \subseteq \sep^I$,
the realizability models that can not be obtained by forcing 
are exactly those for which 
$\sep[I] \neq \sep^I$ (see \cite{Miquel17}).

\section{Decomposing the arrow: disjunctive algebras}
\label{s:da}

\newcommand  {\murrule }{\autorule{\vdash\mu}     }
\newcommand  {\mulrule }{\autorule{\mu\,\vdash}     }
\renewcommand{\negrrule}{\autorule{\vdash\neg}    }
\renewcommand{\neglrule}{\autorule{\neg\,\vdash}    }
\renewcommand{\farrule }{\autorule{\vdash\forall} }
\renewcommand{\falrule }{\autorule{\forall\,\vdash} }
\newcommand  {\parrrule}{\autorule{\vdash\parr}   }
\newcommand  {\parlrule}{\autorule{\parr\,\vdash}   }
\renewcommand{\axrrule }{\autorule{\vdash ax}     }
\renewcommand{\axlrule }{\autorule{ax\vdash}      }


We shall now introduce the notion of disjunctive algebra, 
which is a structure primarily based on disjunctions, negations (for
the connectives) and meets (for the universal quantifier).
Our main purpose is to draw the comparison with implicative algebras,
as an attempt to justify eventually that the latter are more  general than the former, 
and to lay the bases for a dualizable definition.
In the seminal paper introducing linear logic~\cite{Girard87}, 
Girard refines the structure of the sequent calculus LK, introducing in particular
negative and positive connectives for disjunctions and 
conjunctions\footnote{We insist on the 
fact that even though we use linear notations afterwards, 
nothing will be linear here.}. 
With this finer set of connectives, the usual implication
can be retrieved using either the negative disjunction: $A\imp B\defeq \neg A \parr B$
or the positive conjunction: $A\imp B \defeq \neg (A \tensor \neg B)$.

In 2009, Munch-Maccagnoni gave a computational account of Girard's presentation 
for classical logic~\cite{Munch09}.
In his calculus, named L,  each connective corresponds to the type of a particular constructor (or destructor).
While L is in essence close to Curien and Herbelin's \lmmt-calculus~\cite{CurHer00} (in particular it is presented 
 with the same paradigm of duality between proofs and contexts), the syntax of terms
 does not include $\lambda$-abstraction (and neither does the syntax of formulas includes an implication).
 The two decompositions of the arrow evoked above are precisely reflected in decompositions of $\lambda$-abstractions 
 (and dually, of stacks) in terms of L constructors. 
 Notably, the choice of a decomposition corresponds to a particular choice 
 of an evaluation strategy\footnote{Phrased differently, this observation 
 can be traced back to different works, for instance by Blain-Levy~\cite[Fig. 5.10]{LevyCBPV}, Laurent~\cite{LaurentPhD} or Danos, Joinet and Schellinx~\cite{DanJoiSch95}.}
 for the encoded $\lambda$-calculus:
 picking the negative $\parr$ connective corresponds to call-by-name,
 while the decomposition using the $\tensor$ connective 
 reduces in a call-by-value fashion.
 
 We shall begin by considering the call-by-name case, 
 which is closer to the situation of implicative algebras. 
 The definition of disjunctive structures and algebras
 are guided by an analysis of the realizability model induced by $\Lpar$, 
 that is Munch-Maccagnoni's system L restricted to 
 the fragment corresponding to negative formulas:  
 $A,B := X \mid A \parr B \mid \neg A \mid \forall X.A$~\cite{Munch09}.\linebreak
 To leave room for more details on disjunctive algebras,
 we elude here the introduction of $\Lpar$ and its relation to the call-by-name
 $\lambda$-calculus, we refer the interested reader to the extended version.

\renewcommand{\T}{\mathcal{T}}  

\subsection{Disjunctive structures}
\setCoqFilename{ImpAlg.ParAlgebras}
We are now going to define the notion of \emph{disjunctive structure}.
Since we choose negative connectives and in particular a universal quantifier, 
we should define commutations with respect to arbitrary meets. 
The realizability interpretation for $\Lpar$  
provides us with a safeguard in this regard, since 
in the corresponding models, if $X\notin \FV(B)$ the following 
equalities\footnote{\label{fn:val}Technically, $\V_0$ is the set of
closed values which, in this setting, are evaluation contexts
(think of $\Pi$ in usual Krivine models), and $\fvval{A}\in\P(\V_0)$ is the (ground) falsity value
of a formula $A$\vlong{ (see App. \ref{s:CbNReal})}.} 
hold:\\[0.3em]
\begin{minipage}{0.49\textwidth}
\begin{enumerate}
 \item $\fvval{\forall X. (A \parr B)}= \fval{(\forall X. A) \parr B}$.
 \item $\fvval{\forall X. (B \parr A)}= \fval{B \parr (\forall X.A)}$.
  \end{enumerate}
\end{minipage}
\begin{minipage}{0.51\textwidth}
\begin{enumerate}
\setcounter{enumi}{2}
 \item $\fvval{\neg{(\forall X. A)}}= \bigcap_{S\in\P(\V_0)}\fvval{\neg{A\{X:=\dot S\}}}$\\[0.4em]
\end{enumerate}
\end{minipage}
\\[0.3em]
Algebraically, the previous proposition advocates for the following 
definition (remember that the order is defined as the reversed inclusion of primitive falsity values (whence $\cap$ is $\msup$)
and that the $\forall$ quantifier is interpreted by $\minf$):

\begin{definition}[Disjunctive structure][ParStructure]
\label{def:dis_struct}
A \emph{disjunctive structure} is a complete lattice $(\A,\myleq)$ equipped with a binary operation $(a,b)\mapsto a \parr b$,
together with a unary operation $a \mapsto \neg a$, 
such that for all $a,a',b,b'\in\A$
 and for any $B\subseteq \A$:\\[0.3em]
 \begin{minipage}{.52\textwidth}
  \begin{enumerate}
 \item $\text{if } a\leq a' \text{~then~} \neg a' \leq \neg a$
 \item $\text{if } a\leq a'  \text{~and~} b\leq b'\text{~then~} a \parr b \leq a' \parr b'$
\item $\minf_{b\in B} (b \parr a) =  (\minf_{b\in B}  b) \parr a$
\end{enumerate}
 \end{minipage}
 \begin{minipage}{.47\textwidth}
  \begin{enumerate}  \setcounter{enumi}{3}
\item $ \minf_{b\in B} (a \parr b) = a \parr (\minf_{b\in B}  b)$
\item $\neg \minf_{a\in A} a = \msup_{a\in A} \neg a$
\end{enumerate}
 \end{minipage}
\end{definition}

Observe that the commutation laws imply the value of the internal laws when applied to the maximal element $\top$: 
 \begin{minipage}{.8\textwidth}
 \nomidem
  \begin{multicols}{3}
 \begin{enumerate}
 \coqitem[par_top_l] $\top \parr a = \top \qquad\qquad$
 \coqitem[par_top_r] $a \parr \top = \top \qquad\qquad$
 \coqitem[neg_top]   $\neg \top = \bot$     
 \end{enumerate}  
  \end{multicols}

 \end{minipage}

We give here some examples of disjunctive structures.
\setCoqFilename{ImpAlg.Dummies}
\begin{example}[Dummy disjunctive structure][dummy_par]
 Given any complete lattice $(\Lat,\leq)$, 
 defining $a\parr b \defeq \top$ and  $\neg a\defeq \bot$
 gives rise to a dummy structure that fulfills the required properties.
 \end{example}
 
\setCoqFilename{ImpAlg.BooleanAlgebras}
\begin{example}[Complete Boolean algebras][cba_pa]
\label{p:ba_pa}
Let $\B$ be a complete Boolean algebra. It encompasses a disjunctive structure defined by:\noem
\begin{multicols}{4}
\begin{itemize}
 \item $\A        \defeq \B$ 
 \item $a \leq b  \defeq a \leq b$ 
 \item $a \parr b \defeq a \lor b$ 
\item $\neg a    \defeq \neg a  $
\end{itemize}
\end{multicols}
\noem
%
\end{example}

\begin{example}[\Lpar~realizability models]
\label{ex:lpar_ds}
Given a realizability interpretation of $\Lpar$, we define:\noem
\begin{multicols}{2}
\begin{itemize}
\item $\A \defeq \P(\V_0)$  
\item $a \leq b \defeq a \supseteq b$
\item $a \parr b \defeq \{(V_1,V_2): V_1\in a \land V_2 \in b\}$
\item $\neg a \defeq [a^\orth]=\{[t]: t\in a^\orth\}$
\end{itemize}
\end{multicols}\noem
\noindent where $\pole$ is the pole, $\V_0$ is the set of closed values\footnoteref{fn:val},
and $(\cdot,\cdot)$ and $[\cdot]$ are the maps corresponding 
to $\parr$ and $\neg$.
The resulting quadruple $(\A,\leq,\parr,\neg)$ is a disjunctive structure\vlong{ (see \Cref{{a:ds_real}})}.
\end{example}

Following the interpretation of the $\lambda$-terms in implicative structures,
we can embed {\Lpar} terms within disjunctive structures. 
We do not have the necessary space here to fully introduce here\vshort{\footnote{See the extended version for more details.}}\vlong{ (see Appendix~\ref{a:int_lpar})},
but it is worth mentioning that the orthogonality relation $t\orth e$ is
interpreted via the ordering $t^\A \leq e^\A$ 
(as suggested in~\cite[Theorem 5.13]{FerrerEtAl15} by the definition of 
an abstract Krivine structure and its pole from an ordered combinatory algebra).

%

\subsection{The induced implicative structure }
\setCoqFilename{ImpAlg.ParAlgebras}
As expected, any disjunctive structure directly induces an implicative structure through
the definition $a\arp b \defeq \neg a \parr b$:
\begin{proposition}[][PS_IS]
If $(\A,\myleq,\parr,\neg)$ is a disjunctive structure, then $(\A,\myleq,\arp)$ is an implicative structure.
\end{proposition}

Therefore, we can again define for all $a,b$ of $\A$ the application $ab$
as well as the abstraction $\lambda f$ for any function $f$ from $\A$ to $\A$;
and we get for free the properties of these encodings in implicative structures.

Up to this point, we have two ways of interpreting a $\lambda$-term 
into a disjunctive structure:
either through the implicative structure which is induced by the disjunctive one,
or by embedding into the $\Lpar$-calculus which is then 
interpreted within the disjunctive structure.
As a sanity check, we verify that both coincide:

\begin{proposition}[$\lambda$-calculus]
\label{p:sanity}
 Let $\A^\parr=(\A,\myleq,\parr,\neg)$ be a disjunctive structure, 
 and $\A^\imp=(\A,\myleq,\arp)$ the implicative structure it canonically defines, we write $\iota$ for the corresponding inclusion.
 Let $t$ be a closed $\lambda$-term (with parameter in $\A$), and $\tr{t}$ his embedding in $\Lpar$.
 Then we have $\iota(t^{\A^{\imp}}) =  \tr{t}^{\A^\parr}$.
\end{proposition}

\setCoqFilename{ImpAlg.ParAlgebras} 
\subsection{Disjunctive algebras} \renewcommand{\arp}{\to}
We shall now introduce the notion of disjunctive separator.
To this purpose, we adapt the definition of implicative separators,
using standard axioms\footnote{These axioms can be found for instance 
in Whitehead and Russell's presentation of logic~\cite{RusWhi25}.
In fact, the fifth axiom
is deducible from the first four as was later shown by Bernays~\cite{Bernays26}.
For simplicity reasons, 
we preferred to keep it as an axiom.}
for the disjunction and the negation instead of 
Hilbert's combinators $\comb{s}$ and $\comb{k}$.
We thus consider the following combinators:
$$
\begin{array}{c|c}
\begin{array}{rcl}
\ps1 &\defeq & \minf_{a\in\A}    \left[(a\parr a) \arp a                           \right]     \\
\ps2 &\defeq & \minf_{a,b\in\A}  \left[a \arp (a\parr b)                           \right]     \\
\ps3 &\defeq & \minf_{a,b\in\A}  \left[(a\parr b) \arp b\parr a                  \right]       \\
\end{array}
\quad&\quad
\begin{array}{rcl}
\ps4 &\defeq & \minf_{a,b,c\in\A}\left[(a\arp b) \arp (c\parr a) \arp (c\parr b)   \right]  \\
\ps5 &\defeq & \minf_{a,b,c\in\A}\left[(a\parr (b\parr c)) \arp ((a\parr b)\parr c)\right]  \\
\end{array}
\end{array}
$$
Separators for $\A$ are defined similarly to the separators for implicative structures, replacing the combinators $\comb{k},\comb{s}$ and $\comb{cc}$ 
by the previous ones.
\begin{definition}[Separator][ParAlgebra]
 We call \emph{separator} for the disjunctive structure $\A$ any subset $\sep\subseteq\A$ that fulfills the following conditions for all $a,b\in \A$:\noem
 \begin{multicols}{2}
  \begin{enumerate}
  \item If $a \in \sep$ and $a\leq b$ then $b\in\sep$.		    
  \item $\ps1,\ps2,\ps3,\ps4$ and $\ps5$ are in $\sep$.		    
  \item If $a\arp b \in \sep$ and $a\in\sep$ then $b\in\sep$.	
 \end{enumerate}
 \end{multicols}
 \nomidem A separator $\sep$ is said to be \emph{consistent} if $\bot\notin\sep$.
 We call \emph{disjunctive algebra} the given of a disjunctive structure
 together with a separator $\sep\subseteq \A$.
\end{definition}


\begin{remark}
The reader may notice that in this section, we do not distinguish between
classical and intuitionistic separators.
Indeed, $\Lpar$ and the corresponding fragment of the sequent calculus are intrinsically classical.
As we shall see thereafter, so are the disjunctive algebras: the negation is always 
involutive  modulo the equivalence $\seq$  (\Cref{p:dne}).
\end{remark}
\begin{remark}[Generalized modus ponens][mod_pon_inf]\label{lm:mod_pon_inf}
The modus ponens, that is the unique deduction rule we have, is actually compatible with meets.
Consider a set $I$ and two families $(a_i)_{i\in I},(b_i)_{i\in I} \in \A^I$,
we have:
\[ \infer{\vdash_I b}{a \vdash_I b & \vdash_I a}\]
where we write $a\vdash_{I} b$ for $(\minf_{i\in I}a_i\arp b_i) \in \sep$ and 
$\vdash_{I} a$ for $(\minf_{i\in I}a_i) \in \sep$.
As  our axioms are themselves expressed as meets, the results that we will obtain internally 
(that is by deduction from the separator's axioms) can all be generalized to meets.
\end{remark}

\setCoqFilename{ImpAlg.BooleanAlgebras}
\begin{example}[Complete Boolean algebras][CBA_PA]
 Once again, if $\B$ is a complete Boolean algebra, $\B$ induces a disjunctive structure in
 which it is easy to verify that the combinators $\ps1,\ps3,\ps3,\ps4$ and $\ps5$ are equal to the maximal element $\top$.
 Therefore, the singleton $\{\top\}$ is a valid separator for the induced disjunctive structure. 
 In fact, the filters for $\B$ are exactly its separators.
\end{example}

\begin{example}[\Lpar~realizability model]
Remember from \Cref{ex:lpar_ds} that any model of classical realizability based on the $\Lpar$-calculus induces a disjunctive structure.
As in the implicative case, the set of formulas realized by a closed term\footnote{Proof-like terms in $\Lpar$ 
simply correspond to closed terms.}
defines a valid separator\vlong{ (see \Cref{p:real_paralg} for further details)}. 
\end{example}

\setCoqFilename{ImpAlg.ParAlgebras}

\subsection{Internal logic}
As in the case of implicative algebras, 
we say that  $a$ \emph{entails} $b$ and write $a\vdash_{\sep} b$ if $a\arp b \in \sep$.
Through this relation, which is again a \purl{lm:pc6}{preorder} relation,
we can relate the primitive negation and disjunction to the negation
and sum type induced by the underlying implicative structure:

\begin{equation}
\hfill a+b \defeq \minf_{c\in\A} ((a\arp c) \arp (b\arp c) \arp c)\hfill 
\tag{$\forall a,b\in\A$}
\end{equation}

In particular, we show that from the point of view of the separator the principle of double negation elimination is valid and the disjunction and this sum type are equivalent:
\begin{proposition}[Implicative connectives]\label{p:dne}
 For all $a,b\in \A$, the following holds:\noem
 \begin{multicols}{3}
  \begin{enumerate}
   \coqitem[neg_imp_bot] $\neg a\vdash_\sep a \arp \bot$ 
   \coqitem[imp_bot_neg] $ a\arp\bot \vdash_\sep \neg a $
   \coqitem[dni_entails] $a\vdash_\sep \neg\neg a$ 
   \coqitem[dne_entails] $\neg \neg a \vdash_\sep a $
   \coqitem[par_or] $a\parr b \vdash_\sep a + b$ 
   \coqitem[or_par] $a +  b \vdash_\sep a \parr b$
  \end{enumerate}
 \end{multicols}\label{p:imp_neg} \noem
 \label{p:blibla}
\end{proposition}
\begin{myproof}
See Appendix \ref{a:p:dne} 
\end{myproof}

\subsection{Induced implicative algebras}
In order to show that any disjunctive algebra is a particular case of implicative algebra,
we first verify that Hilbert's combinators belong to any disjunctive separator:
\begin{proposition}[Combinators]
We have:
\begin{minipage}{0.5\textwidth}
 \nomidem
 \begin{multicols}{3}
\begin{enumerate}
 \coqitem[psep_K] $\comb{k}^\A\in\sep$
 \coqitem[psep_S] $\comb{s}^\A\in\sep$
 \coqitem[psep_cc] $\comb{cc}^\A\in\sep$
\end{enumerate}
\end{multicols}
\end{minipage}
\end{proposition}

\begin{myproof}
See Appendix \ref{a:s:da_ia}.
\end{myproof}

As a consequence, we get the expected theorem:
\begin{theorem}[][PA_IA]\label{thm:pa_ia}
 Any disjunctive algebra is a classical implicative algebra.
\end{theorem}
\begin{myproof}
 The conditions of upward closure and closure under modus ponens coincide for implicative and disjunctive separators, 
 and the previous propositions show that $\comb{k},\comb{s}$ and $\comb{cc}$ belong to the separator of any disjunctive algebra.
\end{myproof}
Since any disjunctive algebra is actually a particular case of implicative algebra, 
the construction leading to the implicative tripos can be rephrased entirely in this framework.
In particular, the same criteria allows us to determine whether the implicative tripos is isomorphic to a 
forcing tripos. 
Notably, a disjunctive algebra admitting an extra-commutation rule
the negation $\neg$ with arbitrary joins ($\neg \msup_{a\in A} a =  \minf_{a\in A}  \neg a$)
will induce an implicative algebra where the arrow commutes with arbitrary joins. 
In that case, the induced tripos would collapse to a forcing situation (see \cite{Miquel17}).

\section{A positive decomposition: conjunctive algebras}
\label{s:ta}
\newcommand{\tens}{\automath{\tensor}}
\newcommand{\xtens}{\automath{\tensor^\bot}}
\newcommand{\ktens}{\automath{\tensor^{\neg\!\neg}}}

\setCoqFilename{ImpAlg.TensorAlgebras}
\subsection{Call-by-value realizability models}
While there exists now several models build of classical theories 
constructed via Krivine realizability~\cite{Krivine12,Krivine15,Krivine16,Miquel11b},
they all have in common that they rely on a presentation of logic
based on negative connectives/quantifiers.
If this might not seem shocking from a mathematical perspective, 
it has the computational counterpart that these models
all build on a call-by-name calculus, namely the $\lambda_c$-calculus\footnote{Actually,
there is two occurrences of realizability interpretations for call-by-value calculus,
including Munch-Maccagnoni's system L, but both are focused on the analysis of
the computational behavior of programs rather than constructing models 
of a given logic~\cite{Munch09,Lepigre16}.}.
In light of the logical consequences that computational choices 
have on the induced theory, it is natural to wonder whether 
the choice of a call-by-name evaluation strategy is anecdotal or fundamental.

As a first step in this direction, we analyze here the algebraic structure
of realizability models based on the {\Ltens} calculus,
the positive fragment of Munch-Maccagnoni's system L corresponding to the formulas
defined by: $ A,B ::=  X \mid \neg A \mid A\tensor B \mid \exists X.A$.
Through the well-known duality between terms and evaluation contexts~\cite{CurHer00,Munch09},
this fragment is dual to the {\Lpar} calculus
and it  naturally allows to embed the $\lambda$-terms evaluated in a call-by-value fashion.
We shall now reproduce the approach we had for {\Lpar}:
guided by the analysis of the realizability models induced by 
the {\Ltens} calculus, we first define \emph{conjunctive structures}.
We then show how these structures can be equipped with a separator and how
the resulting \emph{conjunctive algebras} lead to the construction of a \emph{conjunctive tripos}.
We will finally show in the next section 
how conjunctive and disjunctive algebras are related 
by an algebraic duality.

\subsection{Conjunctive structures}
As in the previous section, we will not introduce here the $\Ltens$ calculus 
and the corresponding realizability models (see \vlong{\Cref{a:ta}}\vshort{the extended version} for details).
Their main characteristic is that, being build on top of a call-by-value calculus,
a formula $A$ is primitively interpreted by 
its \emph{ground truth value} $\tvv{A}\in\P(\V_O)$
which is a set of values. Its  falsity and truth values are then defined by 
orthogonality~\cite{Munch09,Lepigre16}.
Once again, we can observe the existing commutations in these realizability 
models. Insofar as we are in a structure centered on positive connectives,
we especially pay attention to the commutations with joins.
As a matter of fact, in any $\Ltens$ realizability model, 
we have that if $X\notin \FV(B)$:\\[0.5em]
\begin{minipage}{0.5\textwidth}
 \begin{enumerate}
 \item $\tvval{\exists X. (A \tensor B)}= \tvval{(\exists X. A) \tensor B}$.
 \item $\tvval{\exists X. (B \tensor A)}= \tvval{B \tensor (\exists X.A)}$.
  \end{enumerate}
 \end{minipage}
\begin{minipage}{0.5\textwidth}
 \begin{enumerate}
 \setcounter{enumi}{2}
  \item $\tvval{\neg{(\exists X. A)}}= \bigcap_{S\in\P(\V_0)}\tvval{\neg{A\{X:=\dot S\}}}$
  \item[]
  \end{enumerate}
 \end{minipage}\\[0.5em]
Since we are now interested in primitive truth values, which are logically ordered by inclusion (in particular,
the existential quantifier is interpreted by unions, thus joins),
the previous proposition advocates for the following definition:

\begin{definition}[Conjunctive structure][TensorStructure]
\label{def:conj_struct}
A \emph{conjunctive structure} is a complete join-semilattice $(\A,\myleq)$ equipped with a binary operation $(a,b)\mapsto a \tensor b$,
and a unary operation $a \mapsto \neg a$, 
such that  for all $a,a',b,b'\in\A$ and for all subset $B\subseteq \A$ we have: \\[0.5em]
\begin{minipage}{0.49\textwidth}
\begin{enumerate}
 \item $\text{if } a\leq a' \text{ then } \neg a' \leq \neg a$
  \item 
 $\text{if } a\leq a' \text{ and } b\leq b' \text{ then } a \tensor b \leq a' \tensor b'$
 
 \item $\msup_{b\in B} (a \tensor b) = a \tensor (\msup_{b\in B}  b)$
 \end{enumerate}
\end{minipage}
\begin{minipage}{0.49\textwidth}
\begin{enumerate}
\setcounter{enumi}{3}
 \item $\msup_{b\in B} (b \tensor a) =  (\msup_{b\in B}  b) \tensor a$
\item $\neg \msup_{a\in A} a = \minf_{a\in A} \neg a$
\item[]
\end{enumerate}
\end{minipage}

\end{definition}


As in the cases of implicative and disjunctive structures, the commutation 
rules imply that:
\begin{minipage}{0.8\textwidth}\nomidem
\begin{multicols}{3}
\begin{enumerate}  
 \coqitem[tensor_bot_l]$\bot \tensor a = \bot\qquad\qquad$
 \coqitem[tensor_bot_r]$a \tensor \bot = \bot\qquad\qquad$
 \coqitem[tensor_top  ]$\neg \bot = \top$     
\end{enumerate}
\end{multicols}
\end{minipage}

\setCoqFilename{ImpAlg.Dummies}
\begin{example}[Dummy conjunctive structure][dummy_tensor]
 Given a complete lattice $\L$, the following definitions give rise to a dummy conjunctive structure: \qquad $ a\tensor b \defeq \bot \qquad\qquad \neg a\defeq \top $.
\end{example}

\setCoqFilename{ImpAlg.BooleanAlgebras}
\begin{example}[Complete Boolean algebras][CBA_TS]
Let $\B$ be a complete Boolean algebra. It embodies a conjunctive structure, that is defined by:
\noem
\begin{multicols}{4}
\begin{itemize}
 \item $\A        \defeq \B$ 
 \item $a \leq b  \defeq a \leq b$ 
 \item $a \tensor b \defeq a \land b$ 
\item $\neg a    \defeq \neg a  $
\end{itemize}
\end{multicols}
\noem
\end{example}

\setCoqFilename{ImpAlg.TensorAlgebras}
\begin{example}[\Ltens~realizability models]
As for the disjunctive case, we can abstract the structure of the realizability interpretation of $\Ltens$
to define:\nomidem
\begin{multicols}{2}
\begin{itemize}
 \item $\A \defeq \P(\V_0)           $
 \item $a \tensor b \defeq \{(V_1,V_2): V_1\in a \land V_2 \in b\}$
 \item $a \leq b \defeq a \subseteq b                              $
 \item $\neg a \defeq [a^\orth]=\{[e]: e\in a^\orth\}               $
\end{itemize}
\end{multicols}\nomidem
\noindent where $\pole$ is the pole, $\V_0$ is the set of closed values and 
$(\cdot,\cdot)$ and $[\cdot]$ are the maps corresponding to $\tensor$ and $\neg$.
The resulting quadruple $(\A,\leq,\tensor,\neg)$ is a conjunctive structure\vlong{ (see \Cref{prop:ltens_ts})}.
\end{example}

It is worth noting that even though we can define an arrow by $a\art b \defeq \neg(a\tensor \neg b)$,
it does not induce an implicative structure:
indeed, the distributivity law is not true in general\footnote{For instance,
it is false in {\Ltens} realizability models.}.
In turns, we have another distributivity\!\! \turl{tarrow_join}{law}
which is usually wrong in implicative structure:
\twoelt{
(\msup_{a\in A} a) \art b = \minf_{a\in A}(a\art b)
}{
\minf_{b\in B}(a\art b)  \not\leq a \art (\minf_{b\in B} b)
}
Actually, implicative structures where both are true corresponds precisely to a degenerated forcing situation.

Here again, we can define an embedding of  $\Ltens$ into any conjunctive structure which is sound with respect to typing and reductions\textsuperscript{\ref{fn:ext}}.

%

\subsection{Conjunctive algebras}
The definition of conjunctive separators turns out to be more subtle than in the disjunctive case.
Among others things, conjunctive structures mainly axiomatize joins, 
while the combinators or usual mathematical axioms that we could wish to have in a separator
are more naturally expressed via universal quantifications, hence meets.
Yet, an analysis of the sequent calculus underlying {\Ltens} type system\footnote{\label{fn:ext}See \vlong{\Cref{a:ca:ca}}\vshort{the extended version} for more details.},
shows that
we could consider a tensorial calculus where deduction systematically involves a conclusion
of the shape $\neg A$.
This justifies to consider the following combinators\footnote{Observe that
are directly dual to the combinators 
for disjunctive separators and that they can be alternatively given the shape $\neg \msup_{\_\in\A} ...$.}:
$$
\begin{array}{c|c}
\begin{array}{r@{~}c@{~}l}
\ts1 &\defeq & \minf_{a\in\A}    \neg\left[\neg (a\tensor a) \tensor a                           \right]     \\
\ts2 &\defeq & \minf_{a,b\in\A}  \neg\left[\neg a \tensor (a\tensor b)                           \right]     \\
\ts3 &\defeq & \minf_{a,b\in\A}  \neg\left[\neg (a\tensor b) \tensor (b \tensor a)                  \right]       \\
\end{array}
&
\begin{array}{r@{~}c@{~}l}
\ts4 &\defeq & \minf_{a,b,c\in\A}\neg\left[\neg (\neg a \tensor b) \tensor (\neg (c \tensor a) \tensor (c \tensor b))   \right]  \\
\ts5 &\defeq & \minf_{a,b,c\in\A}\neg\left[\neg (a \tensor (b \tensor c) )\tensor ((a \tensor b) \tensor c)\right]  \\
\end{array}
\end{array}
$$
and to define conjunctive separators as follows:

\begin{definition}[Separator][TensorAlgebra]
 We call \emph{separator} for the disjunctive structure $\A$ any subset $\sep\subseteq\A$ 
 that fulfills the following conditions for all $a,b\in \A$:\\[0.5em]
 \begin{minipage}{0.5\textwidth}
 \begin{enumerate}
  \item If $a \in \sep$ and $a\leq b$ then $b\in\sep$.		
  \item $\ts1,\ts2,\ts3,\ts4$ and $\ts5$ are in $\sep$.		
  \end{enumerate}
 \end{minipage}
 \begin{minipage}{0.5\textwidth}
 \begin{enumerate}\setcounter{enumi}{2}
  \item If $\neg(a\tensor b) \in \sep$ and $a\in\sep$ then $\neg b\in\sep$.	
  \item If $a \in \sep$ and $b\in\sep$ then $a\tensor b\in\sep$.	
 \end{enumerate}
 \end{minipage}\\[0.5em]
  A separator $\sep$ is said to be \emph{classical} if besides $\neg\neg a\in\sep$ implies $a\in\sep$.
\end{definition}

\begin{remark}[Modus Ponens][MP_t]
 If the separator is classical, it is easy to see that the modus ponens is valid:
 if $a\art b \in \sep$ and $a\in\sep$, then $\neg\neg b\in\sep$ by (3) and thus $b\in \sep$.
\end{remark}

\begin{example}[Complete Boolean algebras][CBA_KTA]
 \label{p:ba_ta}
 Once again, if $\B$ is a complete Boolean algebra, $\B$ induces a conjunctive structure in
 which it is easy to verify that the combinators $\ps1,\ps3,\ps3,\ps4$ and $\ps5$ are equal to the maximal element $\top$.
 Therefore, the singleton $\{\top\}$ is a valid separator.
\end{example}

\begin{example}[\Ltens realizability model]
As expected, the set of realized formulas by a proof-like term:
defines a valid separator for the conjunctive structures induced 
by $\Ltens$ realizability models.
\end{example}

\begin{example}[Kleene realizability]
 We do not want to enter into too much details here,
 but it is worth mentioning that realizability interpretations \emph{à la}
 Kleene of intuitionistic calculi equipped with primitive pairs 
 (\emph{e.g.} (partial) combinatory algebras, the $\lambda$-calculus) 
 induce conjunctive algebras. 
 Insofar as many Kleene realizability models takes position against classical reasoning 
 (for $\forall X.X\lor \neg X$ is not realized and hence its negation is), 
 these algebras have the interesting properties of not being classical 
 (and are even incompatible with a classical completion).
\end{example}

\begin{remark}[Generalized axioms][MP_inf]
Once again, the axioms (3) and (4) generalize 
to meet of families $(a_i)_{i\in I},(b_i)_{i\in I}$:
\twoelt{
\infer{\vdash_I \neg b}{\vdash_I \neg (a\tensor b) & \vdash_I a}
}{
\infer{\vdash_I a\tensor b}{\vdash_I a & \vdash_I b}
}
where we write $\vdash_{I} a$ ~ for $(\minf_{i\in I}a_i) \in \sep$ and 
where the negation and conjunction of families are taken pointwise.
Once again, the axioms being themselves expressed as meets,
this means that any result obtained from the separator’s axioms 
(but the classical one) can be generalized to meets.
\end{remark}

\subsection{Internal logic}
As before, we consider the entailment relation defined by $a\vdash_\sep b \defeq (a\art b)\in\sep$.
Observe that if the separator is not classical,
we do not have that $a\vdash_\sep b$ and $a\in\sep$ entails\footnote{Actually we 
can consider a different \turl{nentails}{relation} $a\vdash^\neg b \defeq \neg(a\tensor b)$
for which $a\vdash^\neg b$ and $a\in\sep$ entails $\neg b$.
This one turns out to be
useful to ease proofs, but from a logical perspective, 
the significant entailment is the one given by $a\vdash_\sep b$.}  $b\in\sep$.
Nonetheless, this relation still defines a preorder in the sense that:
\begin{proposition}[Preorder]
 \label{a:po}
 For any $a,b,c\in\A$, we have:\vlong{\\[1em]}\vshort{\\}
 \begin{minipage}{0.5\textwidth}
  \begin{enumerate}
  \coqitem[id_t] $ a\vdash_\sep a$
\end{enumerate}
 \end{minipage}
  \begin{minipage}{0.5\textwidth}
  \begin{enumerate}\setcounter{enumi}{1}
  \coqitem[C6_t] If $a\vdash_{\sep} b$ and $b\vdash_{\sep} c$ then $a\vdash_{\sep} c$ 
\end{enumerate}
 \end{minipage}
\end{proposition}
Intuitively, this reflects the fact that despite we may not be able to extract 
the value of a computation,
we can always chain it with another computation expecting a value.

Here again, we can relate the negation $\neg a$ 
to the one induced by the arrow $a\art \bot$:
\begin{proposition}[Implicative negation]
 For all $a\in \A$, the \etienne{compact} following holds:\vshort{\noem}\vlong{\nomidem}
 \begin{multicols}{4}
  \begin{enumerate}
   \coqitem[tneg_imp_bot] $\neg a\vdash_\sep a \art \bot$ 
   \coqitem[imp_bot_tneg] $ a\art\bot \vdash_\sep \neg a $
   \coqitem[dni_t] $a\vdash_\sep \neg\neg a$ 
   \coqitem[dne_t] $\neg \neg a \vdash_\sep a $
  \end{enumerate}
 \end{multicols}\noem
\end{proposition}

\label{s:conj_lambda}

As in implicative structures, we can define 
the abstraction and application of the $\lambda$-calculus:
\twoelt{
\lambda f \defeq \minf_{a\in\A} (a \art f(a))
}{
ab \defeq \minf\{\neg\neg c: a \leq b \art c\}
}
Observe that here we need to add a double negation,
since intuitively $ab$ is a \emph{computation}
of type $\neg\neg c$ rather than a value of type $c$.
In other words, values are not stable by applications,
and extracting a value from a computation requires a form of classical control.
Nevertheless, for any separator we have:
\begin{proposition}[][app_closed]
 If $a\in\sep$ and $b\in\sep$ then $ab\in\sep$.
\end{proposition}

Similarly, the beta reduction rule now involves a double-negation 
on the reduced term:
\begin{proposition}[][beta_reduction]
 $(\lambda f)a \leq \neg\neg f(a)$
\end{proposition}
\begin{myproof}
 We first show that $t\leq a\art b$ implies $ta \leq \neg\neg b$,
 and we use that $\lambda f \leq a \art f(a)$ to conclude.
  \end{myproof}

We show that Hilbert's combinators $\comb{k}$ and $\comb{s}$
belong to any conjunctive separator:
\begin{proposition}[$\comb{k}$ and $\comb{s}$][tsep_K]
We have:\\[0.3em]
\begin{minipage}{0.43\textwidth}
 \begin{enumerate}
 \coqitem[tsep_K] $(\lambda x y.x)^\A \in\sep$
 \end{enumerate}
\end{minipage}
\begin{minipage}{0.55\textwidth}
 \begin{enumerate}\setcounter{enumi}{1}
 \coqitem[tsep_S] $(\lambda x y z.x\,z\,(y\,z))^\A \in\sep$
\end{enumerate}
\end{minipage}
\end{proposition}

By combinatorial completeness,  
for any closed $\lambda$-term $t$ we thus have the 
a combinatorial term $t_0$ (\emph{i.e.} a composition of $\comb{k}$ and $\comb{s}$)
such that $t_0 \to^* t$. 
Since $\sep$ is closed under application, 
$t_0^\A$ also belong to $\sep$. Besides,
since for each reduction step $t_n \to  t_{n+1}$, 
we have $t_n^\A \leq \neg\neg t_{n+1}^\A$, if 
the separator is classical\footnote{Actually,
since we always have that if $\neg\neg\neg\neg a\in\sep$ then $\neg\neg a\in\sep$,
the same proof shows that in the intuitionistic 
case we have at $\neg\neg t^\A\in\sep$.}, we can thus deduce
that it contains the interpretation of $t$ :
\begin{theorem}[$\lambda$-calculus] \label{thm:lambda}
 If $\sep$ is classical and $t$ is a closed $\lambda$-term,
 then $t^\A\in\sep$.
\end{theorem}

\vlong{\pagebreak}
Once more, 
the entailment relation induces a structure of (pre)-Heyting algebra,
whose conjunction and disjunction are naturally given by
$a\times b \defeq a\tensor b$ and $ a+b\defeq \neg(\neg a \tensor \neg b)$:
\begin{proposition}[Heyting Algebra]  For any $a,b,c\in\A$ \label{p:ta_heyting}
 For any $a,b,c\in\A$, we have:\nomidem
\begin{multicols}{3}
  \begin{enumerate}
    \coqitem[Heyting_and_l] $a\times b \vdash_\sep a$
    \coqitem[Heyting_and_r] $a\times b \vdash_\sep b$
    \coqitem[Heyting_or_l] $a \vdash_\sep a+b$ 
    \coqitem[Heyting_or_r] $b \vdash_\sep a+b$
    \coqitem[Heyting_adj] $a \vdash_{\sep}  b \art c ~~ \text{iff } ~~ a\times b \vdash_{\sep} c$
  \end{enumerate}
 \end{multicols}
\end{proposition}

We can thus quotient the algebra by the equivalence relation $\seq$ and extend the previous operation to
equivalence classes in order to obtain a Heyting algebra $\A/\seq$.
In particular, this allows us to obtain a tripos out of a conjunctive algebra
by reproducing the construction of the implicative tripos in our setting:
\begin{theorem}[Conjunctive tripos]
 Let $(\A,\leq,\to,\sep)$ be a classical\footnote{For technical reasons, we only give the proof in case where the separator is classical 
(recall that it allows to directly use $\lambda$-terms), but as explained, by adding double negation everywhere 
the same reasoning should work for the general case as well. Yet, this is enough to express our main result in the next section which only deals
with the classical case.}
conjunctive algebra. The following functor (where $f:J\to I$) defines a tripos:
\twoelt{
\T :  I\mapsto \A^I/\usep
}{
\T(f) :\left\{\begin{array}{lcl}
		 \A^I/\usep &\to& \A^J/\sep[J]\\ [0.5em]
		 \left[(a_i)_{i\in I}\right]& \mapsto &[(a_{f(j)})_{j\in J}]
		\end{array}
		\right.
}
\end{theorem}
\begin{myproof}
 The proof mimics the proof in the case of implicative algebras, 
 see Appendix \ref{a:ta:tripos}.
\end{myproof}

%
%
%
%
%
%
%
%
%

\section{The duality of computation, algebraically}
\label{s:du}

\setCoqFilename{ImpAlg.Duality}

In~\cite{CurHer00}, Curien and Herbelin introduce the $\lmmt$ 
in order to emphasize the so-called duality of computation
between terms and evaluation contexts.
They define a simple translation inverting the role of terms and stacks
within the calculus, which has the notable consequence of 
translating a call-by-value calculus into a call-by-name calculus and vice-versa.
The very same translation can be expressed within L,
in particular it corresponds to the trivial translation from mapping
every constructor on terms (resp. destructors) in {\Ltens} to the 
corresponding constructor on stacks (resp. destructors) in {\Lpar}.
We shall now see how this fundamental duality of computation can be retrieved
algebraically between disjunctive and conjunctive algebras.

We first show that we can simply pass from one structure to another
by reversing the order relation.
We know that reversing the order in a complete lattice 
yields a complete lattice in which meets and joins are exchanged.
Therefore, it only remains to verify that the axioms of disjunctive and conjunctive 
structures can be deduced through this duality one from each other,
which is the case.
\vshort{\pagebreak}

\begin{proposition}[][PS_TS]
Let $(\A,\leq,\parr,\neg)$ be a disjunctive structure.
Let us define:\noem
\begin{multicols}{4}
\begin{itemize}
\item $\A^\tensor \defeq \A^\parr   $
\item $a \revord b \defeq b \leq a  $
\item $a \tensor b \defeq a \parr b $
\item $\neg a \defeq \neg a         $
\end{itemize}
\end{multicols}\noem
\noindent then $(\A^\tensor,\revord,\tensor,\neg)$ is a conjunctive structure.
\end{proposition}

\begin{proposition}[][TS_PS]
Let $(\A,\leq,\tensor,\neg)$ be a conjunctive structure.
Let us define:\noem
\begin{multicols}{4}
\begin{itemize}
 \item $\A^\parr \defeq \A^\tensor    $
 \item $a \revord b \defeq b \leq a   $
 \item $a \parr b \defeq a \tensor b  $
 \item $\neg a \defeq \neg a$
\end{itemize}
\end{multicols}\noem
\noindent then $(\A^\tensor,\revord,\tensor,\neg)$ is a disjunctive structure.
\end{proposition}

Intuitively, by considering stacks as realizers, 
we somehow reverse the algebraic structure,
and we consider as valid formulas the ones whose orthogonals were valid.
In terms of separator, it means that when reversing a structure
we should consider the separator 
defined as the preimage through the negation of the original separator.

\begin{theorem}[][TA_PA]
\label{p:ta_pa}
 Let $(\A^\tensor,\tsep)$ be a conjunctive algebra, the set
 $\psep ~\defeq  \{a\in\A : \neg a\in \tsep\}$
 defines a valid separator for the dual disjunctive structure $\A^\parr$. 
 \end{theorem}

 \begin{theorem}[][PA_KTA]
\label{p:pa_kta}
 Let $(\A^\parr,\psep)$ be a disjunctive algebra.
 The set $\tsep ~\defeq \{a\in\A : \neg a\in \psep\}$
 defines a classical separator for the dual conjunctive structure $\A^\tens$. 
\end{theorem}

\begin{myproof}
See Appendix \ref{a:p:ta_pa}.
\end{myproof}

It is worth noting that reversing in both cases, the dual separator is classical. 
This is to connect with the fact that classical reasoning principles are true on negated formulas.
Moreover, starting from a non-classical conjunctive algebra, one can reverse it twice to get a classical algebra.
This corresponds to a classical completion of the original separator $\sep$:
\vlong{by definition, $\neg^{-2}(\sep)=\{a: \neg \neg a\in\sep\}$, and it is easy to see that $a\in\sep$ implies $\neg\neg a\in\sep$,
hence $\sep\subseteq \neg^{-2}(\sep)$. }
\vshort{it is easy to see that $a\in\sep$ implies $\neg\neg a\in\sep$,
hence $\sep\subseteq \{a: \neg \neg a\in\sep\}$. }

Actually, the duality between disjunctive and (classical) conjunctive algebras is even stronger,
in the sense that through the translation, the induced triposes are isomorphic.
Remember that an isomorphism $\varphi$ between two (\Set-based) triposes $\T,\T'$ 
is defined as a natural isomorphism $\T\Rightarrow\T'$ in the category \HA,
that is as a family of isomorphisms $\varphi_I:\T(I)\overset{\sim}{\to}\T'(I)$ (indexed by all $I\in\Set$)
that is natural in $\I$.
\begin{theorem}[Main result]\label{thm:iso}
 Let $(\A,\sep)$ be a disjunctive algebra and $(\bar\A,\bar\sep)$ its dual conjunctive
 algebra.
 The \vshort{following }family of maps \vshort{defines a tripos isomorphism}: 
\[
\varphi_I:\left\{ \begin{array}{ccc}
                    \bar\A/\bar\sep[I] &\to& \A/\sep[I]\\
                    {[a_i]} & \mapsto &  [\neg a_i]
                  \end{array}\right.
\] 
\vlong{defines a tripos isomorphism. }
\end{theorem}
\begin{myproof}
The naturality of $\varphi$ is clear by construction.
Recall that to prove that $\varphi$ is an isomorphism of Heyting algebras,
it is enough to show that it is an isomorphism of posets.
The whole proof rely once again on the key \trurl{arrow_tsep_psep}{fact} that
for any $a,b$:
$$ a \vdash_{\bar\sep} b \Leftrightarrow \neg a \vdash_{\sep} \neg b $$
which directly implies that $\varphi$ is order-preserving.
Let $I$ be a fixed set.
Let us show that $\varphi_I$ is injective.
Let $(a_i)_{i\in I},(b_i)_{i\in I}$ be two families
of elements of $\A$ such
that $\varphi_I([a_i])=\varphi_I([b_i])$, \emph{i.e.}
$\neg a_i \equiv_{\sep[I]}\neg b_i$. 
Using the lemma above, if  for any $i\in I$, $\neg a_i \seq \neg b_i$ 
then  for any $i\in I$, $a_i \equiv_{\bar\sep} b_i$, \emph{i.e.} $[a_i] = [b_i]$.
To show that $\varphi_I$ is surjective, it suffices to see that
for any family $(a_i)_{i\in I}$, we have 
$[a_i] = [\neg \neg a_i]=\varphi_I([\neg a_i])$ (in $\A/\S[I]$).
\end{myproof}

\section{Conclusion}
\subsection{An algebraic view on the duality of computation}
To sum up, in this paper we saw how 
the two decompositions of the arrow $a\to b$
as $\neg a \parr b$ and $\neg(a\tensor \neg b)$,
which respectively induce decompositions of 
a call-by-name and call--by-value $\lambda$-calculi within Munch-Maccagnoni's system L~\cite{Munch09},
yield two different algebraic structures reflecting the corresponding realizability models.
Namely, call-by-name models give rise to disjunctive algebras, which are
particular cases of Miquel's implicative algebras~\cite{Miquel17};
while conjunctive algebras correspond to call-by-value realizability models.

The well-known duality of computation between terms and contexts
is reflected here by simple translations from conjunctive to disjunctive algebras
and vice-versa, where the underlying lattices are simply reversed.
Besides, we showed that (classical) conjunctive algebras induce 
triposes that are isomorphic to disjunctive triposes.
The situation is summarized in \Cref{fig:final},
where $\ktens$ denotes classical conjunctive algebras.

\begin{figure}[t]
  \centering
  \begin{tikzpicture}[<-,>=stealth,shorten >=1pt,auto,node distance=2cm,
  thick,main node/.style={},scale=1]
  \node[main node] (PA)  				  {\textbf{$\parr$-algebras}};
  \node[main node] (IA)  [above left=0.6cm and 0.5cm of PA] {\textbf{$\to$-algebras}};
  \node[main node] (XTA) [right=2.25cm of PA]              {};
  \node[main node] (TA)  [above=0.6cm of XTA]               {\textbf{\tens-algebras~~~~}   };
  \node[main node] (KTA) [below=0.6cm of XTA]               {\textbf{\ktens-algebras~~~~} };
  \node[main node] (BA)  [below=1.5cm of PA]                {\textbf{Boolean algebras  }};

  \node[main node] (leg)  [above left=0.5cm and 2cm of BA]                {\small \emph{instance}};
  \node[main node] (leg2) [above right=0cm and -1.32cm of leg]                {\small \emph{translation}};
  \node[main node] (lega)  [left=0.65cm of leg]               {};
  \node[main node] (leg2a)  [left=0.65cm of leg2]               {};
  
  \path[every node/.style={font=\sffamily\scriptsize}]
    (IA) edge [] node [sloped, anchor=center,above] {Thm. \ref*{thm:pa_ia}} (PA)
    (PA) edge [right] node [sloped, anchor=center,above]{Ex. \ref*{p:ba_pa}} (BA)
    (TA) edge [] (KTA)
    (PA) edge [dashed] node [sloped, anchor=center,above]  {Thm. \ref*{p:ta_pa}} (TA)
    (KTA) edge [right] node [sloped, anchor=center,below]{Ex. \ref*{p:ba_ta}} node [above]{}  (BA)
	  edge [dashed] node [sloped, anchor=center,above]  {Thm. \ref*{p:pa_kta}} (PA)
    (leg) edge [] (lega)
    (leg2) edge [dashed] (leg2a);    
\end{tikzpicture}
\caption{Final picture}
\label{fig:final}
\end{figure}

\subsection{From Kleene to Krivine via negative translation}
We could now re-read within
our algebraic landscape
the result of Oliva and Streicher stating
that Krivine realizability models for PA2 
can be obtained as a composition
of Kleene realizability for HA2 and Friedman's negative translation~\cite{OliStr08,Miquel11}.
Interestingly, in this setting the fragment of formulas that is interpreted
in HA2 correspond exactly to the positive formulas of \Ltens,
so that it gives rise to an (intuitionistic) conjunctive algebra.
Friedman's translation is then used to encode the type of stacks within this fragment
via a negation.
In the end, realized formulas are precisely the ones that are realized through Friedman's translation:
the whole construction exactly matches the passage from a intuitionistic conjunctive structure
defined by Kleene realizability to a classical implicative algebras through the arrow 
from $\tens$-algebras to $\imp$-algebras via $\parr$-algebras.

\subsection{Future work}
While \Cref{thm:iso} implies that call-by-value and call-by-name models based on
the $\Ltens$ and $\Lpar$ calculi are equivalents, it does not 
provide us with a definitive answer to our original question.
Indeed, just as (by-name) implicative algebras are more general than disjunctive algebras,
it could be the case that there exists a notion of (by-value) implicative algebras that
is strictly more general than conjunctive algebras and which is not isomorphic to a by-name situation.

Also, if we managed to obtain various results about  conjunctive algebras,
there is still a lot to understand about them. Notably, the interpretation we have of the $\lambda$-calculus
is a bit disappointing in that it does not provide us with an adequacy result as nice as in implicative algebras.
In particular, the fact that each application implicitly gives rise to a double negation
breaks the compositionality. 
This is of course to connect with the definition of \emph{truth values} in by-value models
which requires three layers and a double orthogonal. We thus feel that many things remain
to understand about the underlying structure of by-value realizability models.

Finally, on a long-term perspective, the next step would be 
to understand the algebraic
impact of more sophisticated evaluation strategy (\emph{e.g.}, call-by-need) or side effects (\emph{e.g.}, a monotonic memory).
While both have been used in concrete cases to give a computational content to certain axioms (\emph{e.g.}, the 
axiom of dependent choice~\cite{Herbelin12}) or model constructions (\emph{e.g.}, forcing~\cite{Krivine11}),
for the time being we have no idea on how to interpret them in the realm of implicative algebras.

\vlong{\paragraph*{Acknowledgment}
The author would like to thank Alexandre Miquel 
to which several ideas in this paper, especially the 
definition of conjunctive separators, should be credited.
This research was partially funded by the ANII research project
FCE\_1\_2014\_1\_104800.
}

\newpage
\bibliographystyle{plainurl}
\bibliography{biblio}

\vlong{\newpage
\appendix

\section{Implicative tripos}
\label{a:imp_trip}

\begin{definition}[Hyperdoctrine]
Let $\C$ be a Cartesian closed category. A \emph{first-order hyperdoctrine} over $\C$ is a contravariant 
functor $\T:\Cop\to \HA$ with the following properties:
\begin{enumerate}
 \item For each diagonal morphism $\delta_X:X\to X\times X$ in $\C$, the left adjoint to $\T(\delta_X)$ at the top element
 $\top\in\T(X)$ exists. In other words, there exists an element $=_X\in\T(X\times X)$ such that for all $\varphi\in\T(X\times X)$:
 $$\top ~\oldleq~ \T(\delta_X)(\varphi)\qquad\Leftrightarrow\qquad =_X ~\oldleq~ \varphi\qquad$$
 \item For each projection $\pi^1_{\Gamma,X}:\Gamma\times X \to \Gamma$ in $\C$, the monotonic function 
 $\T(\pi^1_{\Gamma,X}):\T(\Gamma) \to \T(\Gamma\times X)$ has both a left adjoint $(\exists X)_\Gamma$ and a right adjoint $(\forall X)_\Gamma$:
 $$\begin{array}{c@{\qquad}c@{\qquad}c}
   \varphi \oldleq \T(\pi^1_{\Gamma,X})(\psi) & \Leftrightarrow & (\exists X)_\Gamma (\varphi) \oldleq \psi\\
   \T(\pi^1_{\Gamma,X})(\varphi)\oldleq \psi  & \Leftrightarrow & \varphi \oldleq (\forall X)_\Gamma (\psi)\\
  \end{array}$$
 \item These adjoints are natural in $\Gamma$, i.e. given $s:\Gamma \to \Gamma'$ in $\C$, the following diagrams commute:
 \begin{center}
\begin{tikzpicture}[->,>=stealth',shorten >=1pt,auto,node distance=2cm,
  thick,main node/.style={}]

  \node[main node] (1)              {$\T(\Gamma'\times X)$};
  \node[main node] (2) [right=2cm of 1] {$\T(\Gamma \times X)$};
  \node[main node] (3) [below of=1] {$\T(\Gamma')$};
  \node[main node] (4) [below of=2] {$\T(\Gamma )$};

  \path[every node/.style={font=\sffamily\scriptsize}]
    (1) edge [left] node [above] {$\T(s\times\id_X)    $} (2)
        edge [right] node [left] {$(\exists X)_{\Gamma'}$} (3)
    (2) edge node [right] 	 {$(\exists X)_\Gamma $} (4)
    (3) edge [right] node [below]{$\T(s)		      $} (4);  
\end{tikzpicture}\qquad
\begin{tikzpicture}[->,>=stealth',shorten >=1pt,auto,node distance=2cm,
  thick,main node/.style={}]

  \node[main node] (1)              {$\T(\Gamma'\times X)$};
  \node[main node] (2) [right=2cm of 1] {$\T(\Gamma \times X)$};
  \node[main node] (3) [below of=1] {$\T(\Gamma')$};
  \node[main node] (4) [below of=2] {$\T(\Gamma )$};

  \path[every node/.style={font=\sffamily\scriptsize}]
    (1) edge [left] node [above] {$\T(s\times\id_X)    $} (2)
        edge [right] node [left] {$(\forall X)_{\Gamma'}$} (3)
    (2) edge node [right] 	 {$(\forall X)_\Gamma $} (4)
    (3) edge [right] node [below]{$\T(s)		      $} (4);  
\end{tikzpicture}

\end{center}
 
 This condition is also called the \emph{Beck-Chevaley conditions}.
\end{enumerate}
The elements of $\T(X)$, as $X$ ranges over the objects of $\C$, are called the \emph{$\T$-predicates}.
\end{definition}

\begin{definition}[Tripos]
\label{def:tripos}
 A \emph{tripos} over a Cartesian closed category $\C$ is a first-order hyperdoctrine $\T:\C\op\to\HA$ 
 which has a \emph{generic predicate}, \emph{i.e.} there exists an object $\Prop\in\C$ and a predicate $\trth\in\T(\Prop)$ such
 that for any object $\Gamma\in\C$ and any predicate $\varphi\in\T(\Gamma)$, there exists a (not necessarily unique) morphism $\chi_\varphi\in\C(\Gamma,\Prop)$
 such that:
 $$\varphi = \T(\chi_\varphi)(\trth)$$
\end{definition}

\subsubsection*{Implicative tripos}

Let us fix an implicative algebra $(\A,\leq,\to,\sep)$
for the rest of this section.
In order to recover a Heyting algebra, it suffices to 
consider the quotient $\A/_{\seq}$ by the relation $\seq$
defined as $a\seq b \defeq (a\vdash_\sep b \land b \vdash_\sep a)$.
We equip this quotient with the canonical order relation:
$$[a] \sleq [b] ~~\defeq~~ a \vdash_{\sep}  b\eqno(\text{for all } a,b\in\A)$$
where we write $[a]$ for the equivalence class of $a\in\A$.
We define:
  \begin{center}
 \begin{tabular}{r@{~}c@{~}l}
    $[a]\imp_\H [b]$ & $\defeq$ & $[a\to b]   $\\
    $[a]\land_\H [b]$ & $\defeq$ & $[a\times b]$\\
    $[a]\lor_\H [b]$ & $\defeq$ & $[a + b]    $\\
   \end{tabular}
   \begin{tabular}{r@{~}c@{~}l}
    $\top_\H$      & $\defeq$ & $[\top] = \sep $\\
    $\bot_\H$      & $\defeq$ & $[\bot] = \{a\in\A:\neg a \in \sep\}$\\
   \end{tabular}
   \end{center}

The quintuple  $(\H,\sleq,\land_\H,\lor_\H,\to_\H)$ is a Heyting algebra.

We define:
  \begin{center}
 \begin{tabular}{r@{~}c@{~}l}
    $[a]\imp_\H [b]$ & $\defeq$ & $[a\to b]   $\\
    $[a]\land_\H [b]$ & $\defeq$ & $[a\times b]$\\
    $[a]\lor_\H [b]$ & $\defeq$ & $[a + b]    $\\
   \end{tabular}
   \begin{tabular}{r@{~}c@{~}l}
    $\top_\H$      & $\defeq$ & $[\top] = \sep $\\
    $\bot_\H$      & $\defeq$ & $[\bot] = \{a\in\A:\neg a \in \sep\}$\\
   \end{tabular}
   \end{center}

The quintuple  $(\H,\sleq,\land_\H,\lor_\H,\to_\H)$ is a Heyting algebra.

\begin{theorem}[Implicative tripos \cite{Miquel17}]
 Let $(\A,\leq,\to,\sep)$ be an implicative algebra. The following functor:
 $$
 \T : I\mapsto \A^I/\usep
\qquad\qquad
\qquad\qquad
\T(f) :\left\{\begin{array}{lcl}
		 \A^I/\usep &\to& \A^J/\sep[J]\\ [0.5em]
		 \left[(a_i)_{i\in I}\right]& \mapsto &[(a_{f(j)})_{j\in J}]
		\end{array}
		\right.
		\eqno(\forall f\in J\to I )$$
 defines\footnote{Note that the definition of the functor on functions $f:J\to I$ assumes implicitly the possibility
 of picking a representative in any equivalent class $[a]\in\A/\usep$, \emph{i.e.} the full axiom of choice.}
 a tripos.  
 \end{theorem}
 \begin{proof}We verify that $\T$ satisfies all the necessary conditions to be a tripos.
    \begin{itemize}
    \item The functoriality of $\T$ is clear.
    \item For each $I\in\Set$, the image of the corresponding diagonal morphism $\T(\delta_I)$ associates to any element $[(a_{ij})_{(i,j)\in I\times I}]\in\T(I\times I)$ 
    the element  $[(a_{ii})_{i\in I}]\in\T(I)$.
    We define\footnote{The reader familiar with classical realizability might recognize the usual interpretation of Leibniz's equality.}:
    $$(=_{I})~:~ i,j \mapsto \begin{cases}\minf_{a\in\A}(a\to a)  & \text{if } i=j\\ \bot \to \top & \text{if } i\neq j\\ \end{cases}$$

    and we need to prove that for all $[a]\in\T(I\times I)$:
    $$     [\top]_{I} \uleq \T(\delta_I)(a) \qquad\Leftrightarrow  \qquad[=_{I}] \oldleq_{S[I\times I]} [a]$$
    Let then $[(a_{ij})_{i,j\in I}]$ be an element of $\T(I\times I)$.
    From left to right, assume that $[\top]_{I} \uleq \T(\delta_I)(a)$, that is to say that 
    there exists $s\in\sep$ such that for any $i\in I$, $s\leq \top \to a_{ii}$.
    Then it is easy to check that for all $i,j\in I$, $\lambda z. z(s(\lambda x.x)) \leq i =_I j \to a_{ij}$.
    Indeed, using the adjunction and the $\beta$-reduction it suffices to show that for all $i,j\in I$,
     $(i =_I j)\leq  (s(\lambda x.x)) \to a_{ij}$. If $i=j$, this follows from the fact that $ (s(\lambda x.x)) \leq a_{ii}$.
     If $i\neq j$, this is clear by subtyping.
    
    From right to left, if there exists $s\in\sep$ such that for any $i,j\in I$, $s\leq  i =_I j \to a_{ij}$,
    then in particular for all $i\in I$ we have $s\leq  (\lambda x.x) \to a_{ii}$, and then $\lambda\_.s(\lambda x.x)\leq \top \to a_{ii}$ which concludes the case.

     \item For each projection $\pi^1_{I\times J}:I\times J \to I$ in $\C$, the monotone function 
 $\T(\pi^1_{I,J}):\T(I) \to \T(I\times J)$ has both a left adjoint $(\exists J)_I$ and a right adjoint $(\forall J)_I$ which are defined by:
 
 $$(\forall J)_I \big(\left[(a_{ij})_{i,j\in I\times J}\right]\big) \defeq \big[(\bigfa_{j\in J} a_{ij})_{i\in I} \big]\qquad\qquad     
  (\exists J)_I \big(\left[(a_{ij})_{i,j\in I\times J}\right]\big) \defeq \big[(\bigex_{j\in J} a_{ij})_{i\in I} \big]$$
  The proofs of the adjointness of this definition are again easy manipulation of $\lambda$-calculus.
  We only give the case of $\exists$, the case for $\forall$ is easier.
  We need to show that for any $[(a_{ij})_{(i,j)\in I\times J}]\in\T(I\times J)$ and for any  $[(b_{i})_{i\in I}]$, we have:
      $$ [(a_{ij})_{(i,j)\in I\times J}] \oldleq_{\sep[I\times J]} [(b_{i})_{(i,j)\in I}] \qquad \Leftrightarrow \qquad \big[(\bigex_{j\in J} a_{ij})_{i\in I} \big] \oldleq_{\sep[I]} [(b_i)_{i\in I}]$$
  Let us fix some $[a]$ and $[b]$ as above.
  From left to right, assume that there exists  $s\in\sep$ such that for all $i\in I$, $j\in J$, 
  $s\leq a_{ij} \to b_i$, and thus $s a_{ij}\leq b_i$. Using the semantic elimination rule of the existential quantifier, we deduce  that for all $i\in I$, if $t\leq \bigex_{j\in J} a_{ij}$, 
  then $t(\lambda x.sx) \leq b_i$. Therefore, for all $i\in I$ we have  $\lambda y.y(\lambda x.sx) \leq \bigex_{j\in J} a_{ij} \to b_i$.
  
  From right to left,  assume that there exists  $s\in\sep$ such that for all $i\in I$,  
  $s\leq \bigex_{j\in J} a_{ij} \to b_i$. For any $j\in J$, using the semantic introduction rule of the existential quantifier, 
  we deduce that for all $i\in I$, $\lambda x.x a_{ij} \leq \bigex_{j\in J} a_{ij}$.
  Therefore, for all $i\in I$ we have  $\lambda x.s(\lambda z.z x) \leq a_{ij} \to b_i$. 

  \item These adjoints clearly satisfy the Beck-Chevaley condition.
  For instance, for the existential quantifier, we have for all $I,I',J$, 
  for any  $[(a_{i'j})_{(i',j)\in I'\times J}]\in\T(I'\times J)$ and any $s:I \to I'$,
  $$ \begin {array}{cl}
   (\T(s)\circ(\exists J)_{I'}) ([(a_{i'j})_{(i',j)\in I'\times J}]) 
   &= \T(s)(\big[(\bigex_{j\in J} a_{i'j})_{i'\in I'} \big]) \\
   &= \big[(\bigex_{j\in J} a_{s(i)j})_{i\in I} \big] \\
   &= ((\exists J)_{I}) ([(a_{s(i)j})_{ij\in I\times J}])\\
   &= ((\exists J)_{I} \circ \T(s\times \id_J) ([(a_{ij})_{i,j\in I\times J}])
\end {array}$$
  \item Finally, we define $\Prop \defeq \A$ and verify that $\trth\defeq [\id_\A] \in \T(\Prop)$ is a generic predicate.
  Let then $I$ be a set, and $a=[(a_i)_{i\in I}]\in\T(I)$. We let $\chi_a: i \mapsto a_i$ be the characteristic function of $a$ 
  (it is in $I\to \Prop$), which obviously satisfies that for all $i\in I$:
      $$\T(\chi_a)(\trth) = [(\chi_a(i))_{i\in I})] = [(a_i)_{i\in I}]$$
      \end{itemize}
      \end{proof}

\newpage
\section{Disjunctive algebras}
\label{a:da}
\subsection{The $\Lpar$ calculus}
\label{s:Lpar}

The $\Lpar$-calculus is the restriction of Munch-Maccagnoni's system L~\cite{Munch09}, 
to the negative fragment corresponding to the connectives $\parr$,  $\neg^-$
(which we simply write $\neg$ since there is no ambiguity here) and $\forall$.
To simplify things (and ease the connection with the \lmmt-calculus~\cite{CurHer00}),
we slightly change the notations of the original paper.
As Krivine's $\lambda_c$-calculus, this language describes
\emph{commands} of abstract machines $c$ that are made of a \emph{term} $t$ taken within its \emph{evaluation context} $e$. 
The syntax is given by\footnote{The reader may observe 
 that in this setting, values are defined as contexts, so that we may have called them \emph{covalues} rather than values.
We stick to this denomination to stay coherent with 
Munch-Maccagnoni's paper~\cite{Munch09}.}:
$$
\begin{array}{c|c}
 \begin{array}{@{\hspace{-0.15cm}}l@{\quad}c@{~}c@{~}l}
\text{\bf Terms}  & t & ::= & x \mid\mu(\alpha_1,\alpha_2).c \mid \mu [x].c \mid \mu\alpha.c \\
\text{\bf Values} & V & ::= & \alpha \mid (V_1,V_2) \mid [t]\\ 
\end{array}&
\begin{array}{l@{\quad}c@{~}c@{~}l}  
\text{\bf Contexts} & e & ::= & \alpha \mid (e_1,e_2) \mid [t] \mid \mu x.c \\
\text{\bf Commands} & c & ::= & \cut{t}{e}  
\end{array}
\end{array}
  $$
We write $\T_0$, $\V_0$, $\E_0$, $\C_0$ for the sets of closed terms, values, contexts and commands.
 the corresponding set of closed values.
We shall say a few words about it:
 \begin{itemize}
 \item $(e_1,e_2)$ are pairs of contexts, which we will relate to usual stacks;
 \item $\mu( \alpha_1, \alpha_2).c$, which binds the co-variables $\alpha_1,\alpha_2$, is the dual destructor for pairs;
 \item $[t]$ is a constructor for the negation, which allows us to embed a term into a context;
 \item $\mu[x].c$, which binds the variable $x$, is the dual destructor;
 \item $\mu \alpha.c$ binds a covariable and 
 allows to capture a context: as such, it implements classical control.
\end{itemize}
\begin{remark}[Notations]
We shall explain that in (full) L, the same syntax allows us to define terms $t$ and contexts $e$ (thanks to the duality between them). 
In particular, no distinction is made between $t$ and $e$, which are both written $t$, and commands are indifferently 
of the shape $\cut{t^+}{t^-}$ or $\cut{t^-}{t^+}$. 
For this reason, in~\cite{Munch09} is considered a syntax where a notation $\bar x$ is used to distinguish between the positive variable $x$ 
(that can appear in the left-member $\lcut{x}$ of a command) and the positive co-variable $\bar x$
(resp. in the right member $\rcut{x}$ of a command).
In particular, the $\mu\alpha$ binder of the {\lmmt}-calculus would have been written $\mu\bar x$ 
and the $\tmu x$ binder would have been denoted by $\mu\alpha$ (see \cite[Appendix A.2]{Munch09}).
We thus switched the $x$ and $\alpha$ of L (and removed the bar), in order to stay coherent with the notations in the rest of this manuscript.
\end{remark}

The reduction rules correspond to what could be expected from the syntax of the calculus:
destructors reduce in front of the corresponding constructors, both $\mu$ binders catch values in front of them
and pairs of contexts are expanded if they are not values\footnote{The reader might recognize the rule ($\zeta$) of Wadler's sequent calculus~\cite{Wadler03}.}.

 $$
 \begin{array}{c|c}
  \begin{array}{r@{~~}c@{~~}l}
  \cut{\mu[x].c}{[t]} 			    & \to & c[t/x] \\
  \cut{t}{\mu x.c} 			    & \to & c[t/x] \\
  \cut{\mu\alpha.c}{V} 			    & \to & c[V/\alpha] \\
  \end{array}
  \quad&\quad
  \begin{array}{r@{~~}c@{~~}l}
  \cut{\mu(\alpha_1,\alpha_2).c}{(V_1,V_2)} & \to & c[V_1/\alpha_1,V_2/\alpha_2] \\
  \cut{t}{(e,e')}			    & \to & \cut{\mu\alpha.\cut{\mu\alpha'.\cut{t}{(\alpha,\alpha')}}{e'}}{e} \\
  \end{array}
 \end{array}
 $$
where in the last rule, $(e,e')\notin V$.

 Finally, we shall present the type system of $\Lpar$. 
 
 In the continuity of the presentation of implicative algebras, we are interested 
 in a second-order settings. Formulas are then defined by the following grammar:
 $$A,B := X \mid A \parr B \mid \neg A \mid \forall X.A \leqno\quad \textbf{Formulas}$$
The type system is presented in a sequent calculus fashion.
We work with two-sided sequents, where typing contexts are defined as usual as finite lists of bindings between variable and formulas:
$$\Gamma ::= \varepsilon \mid \Gamma,x:A \qquad\qquad\qquad \Delta ::= \varepsilon \mid \Delta,\alpha:A$$
Sequents are of three kinds: $\Gamma\vdash t:A \mid \Delta$ for typing terms,
 $\Gamma\mid e:A\vdash  \Delta$ for typing contexts,
 $c: \Gamma\vdash  \Delta$ for typing commands.
In the type system, left rules corresponds to constructors
while right rules type destructors. 
The type system is given in Figure~\ref{fig:Lpar:typing}.
%

\newcommand{\vsep}{0.7em}
\begin{figure}[t]
\myfig{ 
$$
  \begin{array}{c}
{ \infer[\cutrule]{ \cut{t}{e}:\Gamma \vdash\Delta}{\Gamma \vdash t :A \mid \Delta & \Gamma \mid e:A\vdash \Delta}} 
\qquad
 \infer[\axlrule ]{\Gamma\mid \alpha : A \vdash \Delta}{(\alpha: A)\in\Delta} 
 \qquad  
 \infer[\axrrule ]{\Gamma\vdash x : A \mid \Delta}{(x: A)\in\Gamma}  \\[\vsep]
 
 \infer[\mulrule ]{\Gamma\mid \mu x.c : A \vdash \Delta}{c:(\Gamma,x: A\vdash\Delta)} 
 \qquad   
 \infer[\parlrule]{\Gamma\mid (e_1,e_2): A \parr B \vdash \Delta}{\Gamma \mid e_1:A \vdash \Delta & \Gamma \mid e_2 :B \vdash \Delta} 
 \qquad
  \infer[\neglrule]{\Gamma\mid [t]:\neg A\vdash\Delta}{\Gamma \vdash t:A \mid \Delta} 
 \\[\vsep]
 
 \infer[\murrule ]{\Gamma\vdash \mu\alpha.c : A \mid \Delta}{c:(\Gamma\vdash\Delta,\alpha: A)}  
 \qquad 
 \infer[\parrrule]{\Gamma\vdash \mu(\alpha_1,\alpha_2).c:A\parr B\mid \Delta}{c:\Gamma\vdash \Delta,\alpha_1:A,\alpha_2:B} 
 \qquad 
 \infer[\negrrule]{\Gamma\vdash \mu[x].c:\neg A\mid\Delta}{c:\Gamma,x:A \vdash\Delta} 
 \\ [\vsep]
 
%
 \infer[\falrule ]{\Gamma\mid e:\forall X. A\vdash\Delta}{\Gamma \mid e:A[B/X] \vdash \Delta}  \qquad 
 \infer[\farrule ]{\Gamma\vdash t:\forall X. A\mid \Delta}{\Gamma\vdash t: A\mid\Delta & X\notin{FV(\Gamma,\Delta)}}
  \end{array}
 $$
} 
 \caption{Typing rules for the $\L_{\parr,\neg}$-calculus}
 \label{fig:Lpar:typing}
\end{figure}

\subsubsection*{Embedding of the $\lambda$-calculus}
Following Munch-Maccagnoni's paper~\cite[Appendix E]{Munch09},
we can embed the $\lambda$-calculus into the $\Lpar$-calculus.
To this end, we are guided by the expected definition of the arrow $A \imp B ~~\defeq~~ \neg A \parr B $.
It is easy to see that with this definition, 
a stack $u\cdot e$ in $A\imp B$ (that is with $u$ a term of type $A$ and $e$ a context of type $B$)
is naturally defined as a shorthand for the pair $([u],e)$, which indeed inhabits the type $\neg A \parr B$.
Starting from there, the rest of the definitions are straightforward:
$$\hspace{-0.8em}\begin{array}{c@{\,}|@{}c}
\begin{array}{r@{\,~}c@{~}l}
  u\cdot e 		& \defeq & ([u],e)\\
 \mu([x],\beta).c	& \defeq & \mu(\alpha,\beta).\cut{\mu[x].c}{\alpha} \\
\end{array}&
\begin{array}{r@{\,~}c@{~}l}
 \lambda x.t		& \defeq & \tmu([x],\beta).\cut{t}{\beta}\\
  t\,u	 		& \defeq & \mu\alpha.\cut{t}{u\cdot\alpha}\\
\end{array}
\end{array}
$$
These definitions are sound with respect to the typing rules expected  from the \lmmt-calculus~\cite{CurHer00}.
In addition, they induce the usual rules of $\beta$-reduction for the call-by-name evaluation strategy 
in the Krivine abstract machine\footnote{Note that in the KAM, all stacks are values.}: 
$$
\cut{t\,u}{\pi} 		 \bred  \cut{t}{u\cdot \pi}\qquad
\cut{\lambda x.t}{u\cdot \pi} 	 \bred  \cut{t[u/x]}{\pi}
\eqno
(\pi\in V)
$$


\subsubsection*{Realizability models}
\label{s:CbNReal}

We briefly go through the definition of the realizability interpretation \emph{à la} Krivine
for $\Lpar$. 
As is usual, we begin with the definition of a pole:
\begin{definition}[Pole]
 A \emph{pole} is defined as any subset $\pole\subseteq\C$ 
 s.t. for all $c,c'\in\C$, if $c\bred c'$ and $c'\in\pole$ then $c\in\pole$.
\end{definition}

As it is common in Krivine's call-by-name realizability,
{falsity values} are defined primitively as sets of {contexts}. 
{Truth values} are then defined by orthogonality to the corresponding falsity values.
We say that a term $t$ is \emph{orthogonal} (with respect to the pole $\pole$) to a context $e$
and we write $t\orth e$ when $\cut{t}{e}\in\pole$.
A term $t$ (resp. a context $e$) is said to be orthogonal to a set $S\subseteq \E_0$ (resp. $S\subseteq \T_0$),
which we write $t\orth S$,
when for all $e\in S$, $t$ is orthogonal to $e$.
Due to the call-by-name\footnote{See \cite[Chapter 3]{these} for a more detailed explanation on this point.}
(which is induced here by the choice of connectives), 
a formula $A$ is primitively interpreted by its \emph{ground falsity value}, which we write $\fvval{A}$ and 
which is a set in $\P(\V_0)$.
Its \emph{truth value} $\tv{A}$ is then defined by orthogonality to $\fvval{A}$ (and is a set in $\P(\T_0)$), 
while its \emph{falsity value} $\fv{A}\in\P(\E_0)$ is again obtained by orthogonality to $\tv{A}$.
To ease the definitions
we assume that for each subset $S$ of $\P(\V_0)$, there is a constant symbol $\dot S$ 
in the syntax of formula. 
Given a fixed pole $\pole$, the interpretation is given by:
$$
\begin{array}{c|c}
\begin{array}{r@{~~}c@{~~}l}
 \fvval{\dot S}      & \defeq & S \\
 \fvval{\forall X.A} & \defeq & \bigcup_{S\in\P(\V_0)} \fvval{A\{X:=\dot S\}}\\
 \fvval{A \parr B}   & \defeq & \{(V_1,V_2):V_1\in\fvval{A} \land V_2 \in \fvval{B}\}\\
\end{array}
&
\begin{array}{r@{~~}c@{~~}ll}
 \fvval{\neg A}      & \defeq & \{[t]:t\in\tval{A}\}\\
 \tval{A}            & \defeq & \{t:\forall V\in\fvval{A}, t\orth V\}\\
 \fval{A}            & \defeq & \{e:\forall t\in\tval{A},  t\orth e\}\\
\end{array}
\end{array}
$$

 We shall now verify that the type system of $\Lpar$ is indeed adequate with this interpretation.
 We first prove the following simple lemma:
\begin{lemma}[Substitution]\label{lm:subs_var}
 Let $A$ be a formula whose only free variable is $X$.
 For any closed formula $B$, if $S=\fvval{B}$, then $\fvval{A[B/X]} = \fvval{A[\dot S/X]}$.
\end{lemma}
\begin{proof}
 Easy induction on the structure of formulas, with the observation that the statement for primitive falsity values implies
 the same statement for truth values ($\tv{A[B/X]} = \tv{A[\dot S/X]}$) and falsity values ($\fv{A[B/X]} = \fv{A[\dot S/X]}$).
 The key case is for the atomic formula $A\equiv X$, where we easily check that:
 $$\fvval{X[B/X]} = \fvval{B} = S = \fvval{\dot S} = \fvval{X[\dot S/X]}$$
\end{proof}

We define $\Gamma\cup\Delta$ as the union of both contexts where we annotate the type of hypothesis $(\kappa:A)\in\Delta$ with $\kappa:A^\negt$:
$$\begin{array}{rcl}
\Gamma\cup(\Delta,\kappa:A) & \defeq & (\Gamma \cup \Delta),\kappa:A^\negt \\
\Gamma\cup\varepsilon	  & \defeq & \Gamma  \\
\end{array}
$$


 The last step before proving adequacy consists in defining substitutions and valuations.
 We say that a \emph{valuation}, which we write $\rho$, is a function mapping each second-order variable 
 to a primitive falsity value $\rho(X)\in\P(\V_0)$. 
 A \emph{substitution}, which we write $\sigma$, is a function mapping each variable $x$ to a closed term $c$ and
 each variable $\alpha$ to a closed value $V\in\V_0$:
 $$\sigma ::= \varepsilon \mid \sigma,x\mapsto t\mid \sigma,\alpha\mapsto V^+$$
 We say that a substitution $\sigma$ realizes a context $\Gamma$ and note $\sigma \real \Gamma $ when
 for each binding $(x:A)\in\Gamma$, $\sigma(x)\in\tval{A}$. Similarly, we say that $\sigma$ realizes a context $\Delta$
 if for each binding $(\alpha:A)\in\Delta$, $\sigma(\alpha)\in\fvval{A}$.

We can now state the property of adequacy of the realizability interpretation:
\begin{proposition}[Adequacy]\label{a:p:lpar:adequacy}
Let $\Gamma,\Delta$ be typing contexts, $\rho$ be a valuation and $\sigma$ be a substitution such that 
$\sigma\real \Gamma[\rho]$ and $\sigma \Vdash \Delta[\rho]$. We have:
\begin{enumerate}
 \item If~ $V^+$ is a positive value such that $\Gamma\mid V^+:A \vdash \Delta$, then $V^+[\sigma] \in\fvval{A[\rho]}$.
 \item If~ $t$ is a term such that $\Gamma\vdash t:A \mid \Delta$, then $t[\sigma] \in\tval{A[\rho]}$.
 \item If~ $e$ is a context such that $\Gamma\mid e:A \vdash \Delta$, then $e[\sigma] \in\fval{A[\rho]}$.
 \item If~ $c$ is a command such that $c: (\Gamma\vdash \Delta)$, then $c[\sigma] \in \pole$. 
\end{enumerate}
\end{proposition}
\begin{proof}
We only give some key cases, the full proof can be found in~\cite{Munch09}.
We proceed by induction over the typing derivations. 
Let $\sigma$ be a substitution realizing $\Gamma[\rho]$ and $\Delta[\rho]$.
\prfcase{$(\vdash\neg)$}
Assume that we have:
 $$\infer[\scriptstyle(\vdash\neg   )]{\Gamma\vdash \tmu[ x ].c:\neg A}{c:\Gamma, x :A \vdash\Delta}$$
and let $[t]$ be a term in $\fvval{A[\rho]}$, that is to say that $t\in\tval{A[\rho]}$.
We know by induction hypothesis that for any valuation $\sigma'\real(\Gamma, x :A)[\rho]$, $c[\sigma']\in\pole$
and we want to show that $\mu[ x ].c[\sigma] \orth [t]$.
We have that:
$$\mu[ x ].c \pole [t]\quad \bred\quad c[\sigma][t/ x ] = c[\sigma,x\mapsto t]$$
hence it is enough by saturation to show that $c[\sigma][u/ x ]\in\pole$. 
Since $t\in\tval{A[\rho]}$, $\sigma[ x \mapsto t]\real(\Gamma, x :A)[\rho]$ 
and we can conclude by induction hypothesis.
The cases for $(\mu\,\vdash)$, $(\vdash\,\mu)$ and $(\vdash\,\parr)$ proceed similarly.
\prfcases{$(\neg\vdash)$}
Trivial by induction hypotheses.


\prfcase{$(\parr\,\vdash)$}
Assume that we have:
$$\infer[\scriptstyle(\parr\,\vdash)]{\Gamma\mid (e_1,e_2): A \parr B \vdash \Delta}{\Gamma \mid e_1:A \vdash \Delta 
		  & \Gamma \mid u :B \vdash \Delta}$$
Let then $t$ be a term in $\tval{(A \parr B )[\rho]}$, to show that $\cut{t}{(e_1,e_2)}\in\pole$, we proceed by anti-reduction:
$$ \cut{t}{(e,e')}			  \bred \cut{\mu\alpha.\cut{\mu\alpha'.\cut{t}{(\alpha,\alpha')}}{e'}}{e}$$
It now easy to show, using the induction hypotheses for $e$ and $e'$ that this command is in the pole:
it suffices to show that the term $\mu\alpha.\cut{\mu\alpha'.\cut{t}{(\alpha,\alpha')}}{e'}\in\tval{A}$, which amounts to showing that
for any value $V_1\in\fvval{A}$:
$$\cut{\mu\alpha.\cut{\mu\alpha'.\cut{t}{(\alpha,\alpha')}}{V}} \bred \cut{\mu\alpha'.\cut{t}{(V,\alpha')}}{e'} \in\pole$$
Again this holds by showing that for any $V'\in\tval{B}$, 
$$\cut{\mu\alpha'.\cut{t}{(V,\alpha')}}{V'} \bred \cut{t}{(V,V')}\in\pole$$

\prfcase{$(\vdash\forall)$}
Trivial.
%
\prfcase{$(\forall\,\vdash)$}
Assume that we have:
$$\infer[\falrule ]{\Gamma\mid e:\forall X. A\vdash\Delta}{\Gamma \mid e:A[B/X] \vdash \Delta}$$
By induction hypothesis, we obtain that $e[\sigma]\in\fv{A[B/X][\rho]}$; so that if we denote $\fvval{B[\rho]}\in\P(\V_0)$
\newline by $S$, we have:
$$e[\sigma]\in\fv{A[\dot S/X]} \subseteq \union{S\in\P(\V_0)} \fvval{A[\dot S/X][\rho]}^{\pole\pole}
\subseteq (\union{S\in\P(\V_0)} \fvval{A[\dot S/X][\rho]})^{\pole\pole} = \fv{\forall X.A[\rho]}$$
where we make implicit use of \Cref{lm:subs_var}.
 \end{proof}

\subsection{Disjunctive structures}

We should now define the notion of \emph{disjunctive structure}.
Regarding the expected commutations, as we choose negative connectives and in particular a universal quantifier, 
we should define commutations with respect to arbitrary meets. 
The following properties of the realizability interpretation for $\Lpar$  provides us with a safeguard
for the definition to come:
\begin{proposition}[Commutations]\label{a:p:commutations}
In any $\Lpar$ realizability model (that is to say for any pole $\pole$), the following equalities hold:
\begin{enumerate}
 \item If $X\notin \FV(B)$, then $\fvval{\forall X. (A \parr B)}= \fvval{(\forall X. A) \parr B}$.
 \item If $X\notin \FV(A)$, then $\fvval{\forall X. (A \parr B)}= \fvval{A \parr (\forall X.B)}$.
 \item $\fvval{\neg{(\forall X. A)}}= \bigcap_{S\in\P(\V_0)}\fvval{\neg{A\{X:=\dot S\}}}$
 \end{enumerate}
\end{proposition}
\begin{proof}
\begin{enumerate}
  \item Assume the $X\notin \FV(B)$, then we have:
 \begin{align*}
  \fvval{\forall X. (A \parr B)} &=
  \union{S\in\P(\V_0)} \fvval{A\{X:=\dot S\} \parr B} \\
  &= \union{S\in\P(\V_0)}\{(V_1,V_2):V_1\in\fvval{A\{X:=\dot S\}} \land V_2 \in \fvval{B}\} \\
  &=\{(V_1,V_2):V_1\in\union{S\in\P(\V_0)}\fvval{A\{X:=\dot S\}} \land V_2 \in \fvval{B}\}\\
  &=\{(V_1,V_2):V_1\in\fvval{\forall X.A} \land V_2 \in \fv{B}\}
  ~~=~~\fvval{(\forall X. A) \parr B}
  \end{align*}
  \item Identical.
\item The proof is again a simple unfolding of the definitions: 
 \begin{align*}
  \fvval{\neg{(\forall X. A})}
  & = \{[t]:t\in\tval{\forall X.A}\}
  ~~~ =~~~ \{[t]:t\in\bigcap_{S\in\P(\V_0)}\tval{A\{X:=\dot S\}}\}\\
  & = \bigcap_{S\in\P(\V_0)}\{[t]:t\in\tval{A\{X:=\dot S\}]}\}
  = \bigcap_{S\in\P(\V_0)}\fvval{\neg{A\{X:=\dot S\}}}
  \end{align*}
 \end{enumerate}
\end{proof}

\begin{proposition} If $(\A,\leq,\parr,\neg)$ is a disjunctive structure, then the following hold for all $a\in\A$:\vspace{-1em}
\begin{multicols}{3}
 \begin{itemize}
 \item [\purl{par_top_l}{1.}] $\top \parr a = \top$
 \item [\purl{par_top_r}{2.}] $a \parr \top = \top$
 \item [\purl{neg_top}  {3.}] $\neg \top = \bot$     
 \end{itemize}
\end{multicols} 
\end{proposition}
\begin{proof}
Using the axioms of disjunctive structures, we prove:
 \begin{enumerate}
 \item for all $a\in\A$,~~ $\top \parr a = (\minf \emptyset) \parr a = \minf_{x,a\in\A} \{x\parr a: x\in\emptyset\} = \minf\emptyset = \top$ 
 \item for all $a\in\A$,~~ $a \parr \top = a \parr (\minf \emptyset) = \minf_{x,a\in\A} \{a\parr x: x\in\emptyset\} = \minf\emptyset = \top$ 
 \item $\neg \top = \neg (\minf \emptyset) = \msup_{x\in\A} \{\neg x: x\in\emptyset \} = \msup\emptyset = \bot$
 \end{enumerate}
\end{proof}

\paragraph*{Disjunctive structures from $\Lpar$ realizability models}
If we abstract the structure of the realizability interpretation of $\Lpar$,
it is a structure of the form $(\T_0,\E_0,\V_0,(\cdot,\cdot),[\cdot],\pole)$,
where $(\cdot,\cdot)$ is a binary map from $\E_0^2$ to $\E_0$ (whose restriction to $\V_0$ has values in $\V_0$),
$[\cdot]$ is an operation from $\T_0$ to $\V_0$, and $\pole\subseteq \T_0\times \E_0$ is a relation.
From this sextuple, we can define:
$$\begin{array}{l@{\qquad}l}
 \text{\textbullet}~\A \defeq \P(\V_0)           & \text{\textbullet}~a \parr b \defeq \{(V_1,V_2): V_1\in a \land V_2 \in b\}\qquad\qquad \\
 \text{\textbullet}~a \leq b \defeq a \supseteq b  & \text{\textbullet}~\neg a \defeq [a^\orth]=\{[t]: t\in a^\orth\}
\end{array} 
 $$

\begin{proposition}\label{a:ds_real}
 The quadruple $(\A,\leq,\parr,\neg)$ is a disjunctive structure.
\end{proposition}
\begin{proof}
We show that the axioms of Definition~\ref{def:dis_struct} are satisfied.
 \begin{enumerate}
 \item (Contravariance) Let $a,a'\in\A$, such that 
 $a\leq a'$ ie $a'\subseteq a$. Then $a^\orth \subseteq {a'}^\orth$ and thus 
 $$\neg a =\{[t]: t\in a^\orth\} \subseteq \{[t]: t\in {a'}^\orth\} = \neg a'$$
 \emph{i.e.} $\neg a' \leq \neg a$.
 
 \item (Covariance) Let $a,a',b,b'\in\A$ such that $a'\subseteq a$ and $b'\subseteq b$.
 Then we have 
 $$a \parr b = \{(V_1,V_2):V_1\in a \land V_2 \in b\}\subseteq \{(V_1,V_2):V_1\in a' \land V_2 \in b'\} = a' \parr b'$$
 \emph{i.e.} $a\parr b \leq a'\parr b'$.

 \item (Distributivity) Let $a\in\A$ and $B\subseteq \A$, we have:
    $$\minf_{b\in B} (a \parr b) = \minf_{b\in B} \{(V_1,V_2):V_1\in a \land e_2 \in b\}
    = \{(V_1,V_2):V_1\in a \land V_2 \in \minf_{b\in B} b\}= a \parr (\minf_{b\in B}  b)$$
 \item (Commutation) Let $B\subseteq \A$, we have (recall that $\msup_{b\in B} b = \bigcap_{b\in B} b$):
       $$\msup_{b\in B} (\neg b) = \msup_{b\in B} \{[t]:t \in b^\pole\}=  \{[t]:t \in \msup_{b\in B} b^\pole\}= \{[t]:t \in (\minf_{b\in B} b)^\pole\}=\neg (\minf_{b\in B} b)$$
\end{enumerate}
\end{proof}

\subsection{Interpreting $\Lpar$}
\label{a:int_lpar}
Following the interpretation of the $\lambda$-calculus in implicative structures, 
we shall now see how $\Lpar$ commands can be recovered from disjunctive structures.
From now on, we assume given a disjunctive structure $(\A,\leq,\parr,\neg)$.

\subsubsection{Commands}\label{s:disj_cmd} 
We define the \emph{commands} of the disjunctive structure $\A$ as the pair $(a,b)$ (which we continue to write $\cut{a}{b}$) with $a,b\in\A$, and
we define the pole $\pole$ as the ordering relation $\leq$. We write $\C_\A=\A\times\A$ for the set of commands in $\A$ and $(a,b)\in\pole$ for $a\leq b$.
Besides, we define an ordering on commands which extends the intuition that the order reflect the ``definedness'' of objects: 
given two commands $c,c'$ in $\C_\A$, we say that $c$ is lower than $c'$ and we write $c \cord c'$
if $c\in \pole$ implies that $c'\in \pole$.
It is straightforward to check that:
\begin{proposition}[][cord_preOrder]
 The relation $\cord$ is a preorder.
\end{proposition}

Besides, the relation $\cord$ verifies the following property of variance with respect to the order $\leq$:
\begin{proposition}[Commands ordering][cord_mon]
For all $t,t',\pi,\pi'\in\A$, if $t\leq t'$ and $\pi'\leq \pi$, then $\cut{t}{\pi} \cord \cut{t'}{\pi'}$.
\end{proposition}
\begin{proof}
 Trivial by transitivity of $\leq$.
\end{proof}

Finally, it is worth noting that meets are covariant with respect to $\cord$ and $\leq$, while joins are contravariant:
\begin{lemma}[][cord_meet] \label{p:cord_meet} 
If  $c$ and $c'$ are two functions associating to each $a\in \A$ the commands $c(a)$ and $c'(a)$ such that $c (a) \cord c' (a)$, 
then we have:
$$\minf_{a\in\A}\{a: c (a)\in\pole\} \leq \minf_{a\in\A}\{a: c'(a)\in\pole\}
\qquad\qquad\quad
\msup_{a\in\A}\{a: c' (a)\in\pole\} \leq \msup_{a\in\A}\{a: c(a)\in\pole\}
$$
\end{lemma}
\begin{proof}
Assume $c, c'$ are such that for all $a\in \A$, $c a \cord c' a$.
Then it is clear that by definition we have the inclusion $\{a\in\A: c(a)\in\pole\} \subseteq \{a\in\A: c' (a)\in\pole\}$, whence the expected results.
\end{proof}

\subsubsection{Contexts}
We are now ready to define the interpretation of $\Lpar$ contexts in the disjunctive structure $\A$.
The interpretation for the contexts corresponding to the connectives is very natural:
\begin{definition}[Pairing][pairing] 
For all $a,b\in\A$, we let $(a,b) \defeq a\parr b $.
\end{definition}

\begin{definition}[Boxing][box]
For all $a\in\A$, we let $[a] \defeq \neg a $.
\end{definition}

Note that with these definitions, the encodings of pairs and boxes directly inherit of the properties of the internal law $\parr$ and $\neg$ in disjunctive structures.
As for the binder $\mu x.c$, which we write $\tmu^+c$, it should be defined in such a  way that if $c$ is a function mapping each $a\in\A$
to a command $c(a)\in\C_\A$, then $\mu^+.c$ should be ``compatible'' with any $a$ such that $c(a)$ is well-formed (\emph{i.e.} $c(a)\in\pole$).
As it belongs to the side of opponents, the ``compatibility'' means that it should be greater than any such $a$, and we thus define it as a join.
 \begin{definition}[{$\mu^+$}][mup]For all $c:\A\to\C_\A$, we define: 
$$\mu^+.c := \msup_{a\in\A}\{a:c(a) \in \pole\}$$
\end{definition}

These definitions enjoy the following properties with respect to the $\beta$-reduction and the $\eta$-expansion:
\begin{proposition}[Properties of $\mu^+$]
\label{prop:mu_p}
For all functions $c,c':\A\to\C_\A$, the following hold:
\begin{itemize}
 \item[\lpurl{mup_mon}{1.}] 	 If for all $a\in\A$,   $c(a)\cord c'(a)$, then $\mu^+ .c'      \leq \mu^+ .c $ 	\hfill(Variance)
 \item[\lpurl{mup_beta}{2.}] 	 For all $t\in\A$, then $\cut{t}{\mu^+.c}\cord c(t) $				\hfill($\beta$-reduction)
 \item[\lpurl{mup_eta}{3.}] 	 For all $e\in\A$, then $t = \mu^+. (a \mapsto \cut{a}{e}) $		       	\hfill($\eta$-expansion)
\end{itemize}
\end{proposition}
\begin{proof}
 \begin{enumerate}
  \item Direct consequence of Proposition \ref{p:cord_meet}.
  \item[2,3.] Trivial by definition of $\mu^+$.
 \end{enumerate}
\end{proof}

\begin{remark}[Subject reduction]
The $\beta$-reduction $c\bred c'$ is reflected by the ordering relation $c\cord c'$, which 
reads \emph{``if $c$ is well-formed, then so is $c'$''}. In other words, this corresponds to the usual property
of subject reduction. In the sequel, we will see that $\beta$-reduction rules of $\Lpar$ will always been
reflected in this way through the embedding in disjunctive structures.
\end{remark}

\subsubsection{Terms}
Dually to the definitions of (positive) contexts $\mu^+$ as a join, we define the embedding of 
(negative) terms, which are all binders, by arbitrary meets:
 \begin{definition}[{$\mu^-$}][mun]
For all $c:\A\to\C_\A$, we define: 
 $$\mu^-. c := \minf_{a\in\A}  \{a: c(a)  \in\pole\} $$
  \end{definition}
 \begin{definition}[{$\mu^{()}c$}][mu_pair]
For all $c:\A^2\to\C_\A$, we define: 
 $$\mu^{()}.c 	:=  \minf_{a,b\in\A}\{a\parr b : c(a,b)\in\pole\} $$
 \end{definition}
 \begin{definition}[{$\mu^{[]}$}][mu_neg]
For all $c:\A\to\C_\A$, we define: 
 $$ \mu^{[]}.c	 :=  \minf_{a\in\A}  \{\neg a : c(a)  \in\pole\} $$
\end{definition}

These definitions also satisfy some variance properties with respect to the preorder $\cord$ and the order relation $\leq$,
namely, negative binders for variable ranging over positive contexts are covariant, while negative binders intended to catch negative terms
are contravariant.
\begin{proposition}[Variance]
For any functions $c,c'$ with the corresponding arities, the following hold:
\begin{itemize}
 \item[\lpurl{mun_mon}{1.}] 	If  $c(a)\cord c'(a)$ for all $a\in\A$,    then $\mu^- .c   \leq \mu^- .c' $
 \item[\lpurl{mu_pair_mon}{2.}] If  $c(a,b)\cord c'(a,b)$ for all $a,b\in\A$, then $\mu^{()}.c \leq \mu^{()}.c'$
 \item[\lpurl{mu_neg_mon}{3.}]  If  $c(a)\cord c'(a)$ for all $a\in\A$, then $\mu^{[]}.c' \leq \mu^{[]}.c$
\end{itemize}
\end{proposition}
\begin{proof}
 Direct consequences of Proposition \ref{p:cord_meet}.
\end{proof}

The $\eta$-expansion is also reflected as usual by the ordering relation $\leq$:
\begin{proposition}[$\eta$-expansion]
\label{prop:lpar_eta_n}
For all $t\in\A$, the following holds:
\begin{itemize}
 \item[\lpurl{mun_eta}{1.}]      $t = \mu^-. (a \mapsto \cut{t}{a}) $
 \item[\lpurl{mu_pair_eta}{2.}]  $t \leq \mu^{()}.(a,b \mapsto \cut{t}{(a,b)})$
 \item[\lpurl{mu_neg_eta}{3.}]   $t \leq \mu^{[]}.(a \mapsto \cut{t}{[a]})$
\end{itemize}
\end{proposition}
\begin{proof} Trivial from the definitions.  \end{proof}

The $\beta$-reduction is reflected by the preorder $\cord$:
\begin{proposition}[$\beta$-reduction]
\label{prop:lpar_bred_n}
For all $e,e_1,e_2,t\in\A$, the following holds:
\begin{itemize}
 \item[\lpurl{mun_beta}{1.}]      $\cut{\mu^-.c}{e}        \cord c(e)  $
 \item[\lpurl{mu_pair_beta}{2.}]  $\cut{\mu^{()}.c}{(e_1,e_2)} \cord c(e_1,e_2)$
 \item[\lpurl{mu_neg_beta}{3.}]   $\cut{\mu^{[]}.c}{[t]}   \cord c(t)  $
\end{itemize}
\end{proposition}
 \begin{proof} Trivial from the definitions.  \end{proof}

Finally, we call a \emph{$\Lpar$ term with parameters in $\A$} (resp. context, command)
any $\Lpar$ term (possibly) enriched with constants taken in the set $\A$. 
Commands with parameters are equipped with the same rules of reduction as in $\Lpar$, 
considering parameters as inert constants.
To every closed $\Lpar$ term $t$ (resp. context $e$,command $c$) we associate an element $t^\A$ (resp. $e^\A$, $c^\A$) of $\A$, defined by induction on the structure of $t$ as follows:
$$
\begin{array}{c@{\qquad}|c}\begin{array}{ccl}
\multicolumn{3}{l}{\textbf{Contexts}:}\\[0.3em]
a^\A 		& \defeq & a 					\\
(e_1,e_2)^\A 	& \defeq & (e_1^\A,e_2^\A)                          \\
{[}t{]^\A}	& \defeq & [t^\A]                               \\
(\mu x.c)^\A    & \defeq & \mu^- (a\mapsto (c[x:=a])^\A)  \qquad   \\
\end{array}
&
\begin{array}{ccl}
\multicolumn{3}{l}{\textbf{Terms}:}\\[0.3em]
a^\A                & \defeq & a 					\\
(\mu\alpha.c)^\A    & \defeq & \mu^- (a\mapsto (c[\alpha:=a])^\A)     \\
(\mu(\alpha_1,\alpha_2).c)^\A    & \defeq & \mu^{()}(a,b\mapsto (c[\alpha_1:=a,\alpha_2:=b])^\A) \\
(\mu[x].c)^\A & \defeq & \mu^{[]}(a \mapsto (c[x:=a])^\A)   \\
\end{array}
\end{array}
$$
$$\begin{array}{ccc}
 \cut{t}{e}^\A & \defeq & \cut{t^\A}{e^\A}) \\
 \end{array}\leqno\quad\textbf{Commands:}$$

In particular, this definition has the nice property of making the pole $\pole$ (\emph{i.e.} the order relation $\leq$)
closed under anti-reduction, as reflected by the following property of $\cord$:
 \begin{proposition}[Subject reduction]
 For any closed commands $c_1,c_2$ of ~$\Lpar$ , if $c_1\bred c_2$ then $c_1^\A \cord c_2^\A$,
 \emph{i.e.} if $c_1^\A$ belongs to $\pole$ then so does $c_2^\A$.
\end{proposition}
\begin{proof}
 Direct consequence of Propositions~\ref{prop:mu_p} and \ref{prop:lpar_bred_n}.
\end{proof}

\subsection{Adequacy}
We shall now prove that the interpretation of $\Lpar$ is adequate with respect to its type system.
Again, we extend the syntax of formulas to define second-order formulas with parameters by:
$$ A,B ::= a \mid X \mid \neg A \mid A\parr B \mid \forall X.A  \eqno (a\in\A)$$
This allows us to embed closed formulas with parameters into the disjunctive structure $\A$. 
The embedding is trivially defined by:
$$\begin{array}{ccl}
a^\A 		& \defeq & a 				      \\
(\neg A)^\A 	& \defeq & \neg A^\A                           \\
(A\parr B)^\A 	& \defeq & A^\A \parr B^\A                         \\
(\forall X.A)^\A& \defeq & \minf_{a\in\A} (A\{X:=a\})^\A     \\
\end{array}
\eqno\begin{array}{r}(\text{if }a\in\A)\\\\\\\end{array}
$$

As for the adequacy of the interpretation for the second-order $\lambda_c$-calculus, we
define substitutions, which we write $\sigma$, as functions mapping variables (of terms, contexts and types) to element of $\A$:
$$\sigma ::= \varepsilon \mid \sigma[x\mapsto a]\mid \sigma[\alpha\mapsto a]\mid\sigma[X\mapsto a] \eqno(a\in\A,~ x,X~ \text{variables})$$
In the spirit of the proof of adequacy in classical realizability, 
we say that a substitution $\sigma$ {realizes} a typing context $\Gamma$, which write $\sigma \Vdash \Gamma$, if for all bindings $(x:A)\in\Gamma$ 
we have $\sigma(x) \leq (A[\sigma])^\A$. Dually, we say that $\sigma$ realizes $\Delta$ if for all bindings $(\alpha:A)\in\Delta$ ,
we have $\sigma(\alpha)\geq (A[\sigma])^\A$.
We can now prove

\begin{theorem}[Adequacy]
\label{p:lpar_adequacy_bis}
The typing rules of $\Lpar$ (\Cref{fig:Lpar:typing}) are adequate with respect to the interpretation of terms, contexts, commands and formulas.
 Indeed, for all contexts $\Gamma,\Delta$, for all formulas with parameters $A$
 then  for all substitutions $\sigma$ such that $\sigma\Vdash \Gamma$ and $\sigma\Vdash \Delta$, we have:
\begin{enumerate}
\item for any term $t$,    if~ $\Gamma\vdash t:A\mid\Delta$,  then ~  $(t[\sigma])^\A \leq A[\sigma]^\A$;
\item for any context $e$, if~ $\Gamma\mid e:A\vdash \Delta$, then ~ $(e[\sigma])^\A  \geq A[\sigma]^\A$;
\item for any command $c$, if~ $c:(\Gamma\vdash \Delta)$,     then ~  $(c[\sigma])^\A \in \pole$.
\end{enumerate}
\end{theorem}

\begin{proof}
By induction over the typing derivations.
\newcommand{\mysubst}[1]{(#1[\sigma])^\A}
\prfcase{\cutrule}
Assume that we have:
$$\infer[\cutrule]{ \cut{t}{e}:\Gamma \vdash\Delta}{\Gamma \vdash t :A \mid \Delta & \Gamma \mid e:A\vdash \Delta}$$
By induction hypotheses, we have  $\mysubst{t}\leq A[\sigma]^\A$ and $\mysubst{e}\geq A[\sigma]^\A$. By transitivity of the relation $\leq$, 
we deduce that $\mysubst{t}\leq \mysubst{e}$, so that $\mysubst{\cut{t}{e}}\in\pole$.

\prfcase{$(\vdash ax)$}
Straightforward, since if $(x:A)\in\Gamma$, then $\mysubst{x}\leq\mysubst{A}$. The case $(ax\vdash)$ is identical.

\prfcase{($\vdash\mu$)}
Assume that we have:
$$\infer[\murrule ]{\Gamma\vdash \mu\alpha.c : A \mid \Delta}{c:\Gamma\vdash\Delta,\alpha: A}  $$
By induction hypothesis, we have that $(c[\sigma,\alpha\mapsto \mysubst{A}])^\A\in\pole$. 
Then, by definition we have: 
$$((\mu\alpha.c)[\sigma])^\A=(\mu\alpha.(c[\sigma]))^\A = \minf_{b\in\A}\{b:(c[\sigma,\alpha\mapsto b])^\A\in\pole \} \leq \mysubst{A}$$

\prfcase{($\mu\,\vdash$)}
Similarly, assume that we have:
 $$\infer[\mulrule ]{\Gamma\mid \mu x.c : A \vdash \Delta}{c:\Gamma,x: A\vdash\Delta}$$
By induction hypothesis, we have that $(c[\sigma,x\mapsto \mysubst{A}])^\A\in\pole$. 
Therefore, we have:
$$((\mu x.c)[\sigma])^\A=(\mu x.(c[\sigma]))^\A = \msup_{b\in\A}\{b:(c[\sigma,x\mapsto b])^\A\in\pole\} \geq \mysubst{A}\,.$$

\prfcase{($\parr\,\vdash$)}
Assume that we have:
$$ \infer[\parlrule]{\Gamma\mid (e_1,e_2): A_1 \parr A_2 \vdash \Delta}{\Gamma \mid e_1:A_1 \vdash \Delta & \Gamma \mid e_2 :A_2 \vdash \Delta}$$
By induction hypotheses, we have that $\mysubst{e_1}\geq \mysubst{A_1}$ and $\mysubst{e_2}\geq \mysubst{A_2}$. 
Therefore, by monotonicity of the $\parr$ operator, we have: 
$$\mysubst{(e_1,e_2)}=(e_1[\sigma],e_2[\sigma])^\A=\mysubst{e_1}\parr\mysubst{e_2}\geq \mysubst{A_1}\parr \mysubst{A_2}\,.$$

\prfcase{($\vdash\,\parr$)}
Assume that we have:
 $$\infer[\parrrule]{\Gamma\vdash \mu(\alpha_1,\alpha_2).c:A_1\parr A_2\mid \Delta}{c:\Gamma\vdash \Delta,\alpha_1:A_1,\alpha_2:A_2}$$
By induction hypothesis, we get that $(c[\sigma,\alpha_1\mapsto\mysubst{A_1},\alpha_2\mapsto\mysubst{A_2}])^\A\in\pole$. 
Then by definition we have 
$$((\mu(\alpha_1,\alpha_2).c)[\sigma])^\A
= \minf_{a,b\in\A}\{a\parr b:(c[\sigma,\alpha_1\mapsto a,\alpha_2\mapsto b])^\A\in\pole\} \leq \mysubst{A_1}\parr\mysubst{A_2}\,.$$

\prfcase{($\neg\,\vdash$)}
Assume that we have:
 $$\infer[\neglrule]{\Gamma\mid [t]:\neg A\vdash\Delta}{\Gamma \vdash t:A \mid \Delta}$$
By induction hypothesis, we have that $\mysubst{t}\leq\mysubst{A}$.
Then by definition of $[\,]^\A$ and covariance of the $\neg$ operator, we have:
$$([t[\sigma]])^\A = \neg \mysubst{t}\geq \neg \mysubst{A}.$$

\prfcase{($\vdash\,\neg$)}
Assume that we have:
 $$\infer[\negrrule]{\Gamma\vdash \mu[x].c:\neg A\mid\Delta}{c:\Gamma,x:A \vdash\Delta}$$
By induction hypothesis, we have that $(c[\sigma,x\mapsto \mysubst{A}])^\A\in\pole$. 
Therefore, we have:
$$((\mu [x].c)[\sigma])^\A=(\mu [x].(c[\sigma]))^\A = \minf_{b\in\A}\{\neg b:(c[\sigma,x\mapsto b])^\A\in\pole\} \leq \neg \mysubst{A}.$$

\prfcase{\falrule}
Assume that we have:
$$\infer[\falrule ]{\Gamma\mid e:\forall X. A\vdash\Delta}{\Gamma \vdash e:A\{X:=B\} \mid \Delta}$$
By induction hypothesis, we have that $\mysubst{e}\geq ((A\{X:=B\})[\sigma])^\A = (A[\sigma,X\mapsto \mysubst{B}])^\A $.
Therefore, we have that $\mysubst{e}\geq (A[\sigma,X\mapsto \mysubst{B}])^\A \geq \minf_{b\in\A}\{A\{X:=b\}[\sigma]^\A\} $.

\prfcase{\farrule}
Similarly, assume that we have:
$$\infer[\farrule ]{\Gamma\vdash t:\forall X. A}{\Gamma\vdash t: A\mid\Delta & X\notin{FV(\Gamma,\Delta)}}$$
By induction hypothesis, we have that $\mysubst{t}\leq (A[\sigma,X\mapsto b])^\A$ for any $b\in A$.
Therefore, we have that $\mysubst{t}\leq\minf_{b\in\A}(A\{X:=b\}[\sigma]^\A) $.
\end{proof}

\subsection{The induced implicative structure }
\setCoqFilename{ImpAlg.ParAlgebras}
As expected, any disjunctive structures directly induces an implicative structure:
\begin{proposition}[][PS_IS]
If $(\A,\myleq,\parr,\neg)$ is a disjunctive structure, then $(\A,\myleq,\arp)$ is an implicative structure.
\end{proposition}
\begin{proof}
 We need to show that the definition of the arrow fulfills the expected axioms:
 \begin{enumerate}
  \item (Variance) Let $a,b,a',b'\in \A$ be such that $a'\myleq a$ and $b\myleq b'$, then we have:
  $$a\arp b = \neg a \parr b \myleq \neg a' \parr b' = a'\arp b'$$
  since $\neg a \leq \neg a'$ by contra-variance of the negation and $b\leq b'$.
  \item (Distributivity) Let $a\in \A$ and $B\subseteq \A$, then we have:
  $$\minf_{b\in B}(a\arp b) = \minf_{b\in B}(\neg a \parr b)= \neg a\parr (\minf_{b\in B} b)= a \arp (\minf_{b\in B} b)$$
  by distributivity of the infimum over the disjunction.
 \end{enumerate}

\end{proof}

\begin{lemma}[][mu_abs_char] The shorthand $\mu([ x ],\alpha).c$ is interpreted in $\A$ by:
$$(\mu([x],\alpha).c)^\A =  \minf_{a,b\in \A} \{(\neg a )\parr b: c[x:=a,\alpha:=b]\in\leq\}$$
\end{lemma}
\begin{proof}
\begin{align*}
\mu([ x ],\alpha).c)^\A &= (\mu( x _0,\alpha).\cut{\mu[ x ].c}{ x _0})^\A \\
& = \minf_{a',b\in \A} \{a'\parr b: (\cut{\mu[ x ].c[\alpha:=b]}{a'})^\A\in\leq\} \\
& = \minf_{a',b\in \A} \{a'\parr b: (\minf_{a\in \A} \{\neg a: c^\A[ x :=a,\alpha:=b]\in\leq\}\leq a' \} \\
& = \minf_{a,b\in \A} \{(\neg a)\parr b: c^\A[ x :=a,\alpha:=b] \in \leq \}
\end{align*}
\end{proof}

\begin{proposition}[$\lambda$-calculus]\label{a:p:sanity}
 Let $\A^\parr=(\A,\myleq,\parr,\neg)$ be a disjunctive structure, 
 and $\A^\imp=(\A,\myleq,\arp)$ the implicative structure it canonically defines, we write $\iota$ for the corresponding inclusion.
 Let $t$ be a closed $\lambda$-term (with parameter in $\A$), and $\tr{t}$ his embedding in $\Lpar$.
 Then we have
 $$\iota(t^{\A^{\imp}}) =  \tr{t}^{\A^\parr}$$
 where $t^{\A^{\imp}}$ (resp. $t^{\A^{\parr}}$) is the interpretation of $t$ within $\A^\imp$ (resp. $\A^\parr$).
\end{proposition}

In other words, this proposition expresses the fact that the following diagram commutes:
\begin{center}
\begin{tikzpicture}[->,>=stealth',shorten >=1pt,auto,node distance=2cm,
  thick,main node/.style={}]

  \node[main node] (1) {$\lambda$-calculus};
  \node[main node] (2) [right of=1] {$\Lpar$};
  \node[main node] (3) [below of=1] {$(\A^\to,\myleq,\imp)$};
  \node[main node] (4) [below of=2] {$(\A^\parr,\myleq,\parr,\neg)$};

  \path[every node/.style={font=\sffamily\scriptsize}]
    (1) edge [left] node [above] {$\tr{\,}$} (2)
        edge [right] node [left] {$[\;]^{\A^\imp}$} (3)
    (2) edge node [right]        {$[\;]^{\A^\parr}$} (4)
    (3) edge [right] node [above] {$\iota$} (4);  
\end{tikzpicture}
\end{center}

\begin{proof}
 By induction over the structure of terms.
 \prfcase{$a$ for some $a\in\A^\parr$}
 This case is trivial as both terms are equal to $a$.
 \prfcase{$\lambda x.u$}
 We have $\tr{\lambda x.u} = \mu ([x],\alpha).\cut{\tr{t}}{\alpha}$ and
 \begin{align*}
  (\mu ([x],\alpha).\cut{\tr{t}}{\alpha})^{\A^\parr} 
 & = \minf_{a,b\in\A}\{\neg a \parr b: (\tr{t[x:=a]}^{\A^\parr},b) \in \pole\} \\
 & = \minf_{a,b\in\A}\{\neg a \parr b: \tr{t[x:=a]}^{\A^\parr} \myleq b \} \\
 & = \minf_{a\in\A}(\neg a \parr \tr{t[x:=a]}^{\A^\parr})
 \end{align*}
 On the other hand,  
 $$\iota([\lambda x.t]^{\A^\imp}) 
 = \iota(\minf_{a\in\A}(a\arp (t[x:=a])^{\A^\imp}))
 = \minf_{a\in\A}(\neg a\parr \iota(t[x:=a]^{\A^\imp}))$$
 Both terms are equal since
 $\tr{t[x:=a]}^{\A^\parr} = \iota(t[x:=a])^{\A^\imp})$
 by induction hypothesis.

 \prfcase{$u\,v$}~\newline
 On the one hand, we have $\tr{u\,v} = \mu (\alpha).\cut{\tr{u}}{([\tr{v}],\alpha)}$ and 
 \begin{align*}
  (\mu (\alpha).\cut{\tr{u}}{([\tr{v}],\alpha)})^{\A^\parr} 
 & = \minf_{a\in\A}\{a: (\tr{u}^{\A^\parr},(\neg \tr{v}^{\A^\parr}\parr a)) \in \pole\} \\
 & = \minf_{a\in\A}\{a: \tr{u}^{\A^\parr}\myleq (\neg \tr{v}^{\A^\parr}\parr a) \} 
 \end{align*}
 On the other hand,  
 $$\iota([u\,v]^{\A^\imp}) 
 = \iota(\minf_{a\in\A}\{a:(u^{\A^\imp})\myleq (v^{\A^\imp}) \arp a\})
 = \minf_{a\in\A}\{a:\iota(u^{\A^\imp})\myleq \neg (\iota(v^{\A^\imp}) \parr a\}))$$
 Both terms are equal since 
 $\tr{u}^{\A^\parr} = \iota(u^{\A^\imp})$ and 
 $\tr{v}^{\A^\parr} = \iota(v^{\A^\imp})$
 by induction hypotheses.
\end{proof}

\subsection{Disjunctive algebras} \renewcommand{\arp}{\to}
\subsubsection*{Separation in disjunctive structures}
We recall the definition of separators for disjunctive structures:
\begin{definition}{11}[Separator][ParAlgebra]
 We call \emph{separator} for the disjunctive structure $\A$ any subset $\sep\subseteq\A$ that fulfills the following conditions for all $a,b\in \A$:
 \begin{enumerate}
  \item If $a \in \sep$ and $a\leq b$ then $b\in\sep$.		\hfill(upward closure)
  \item $\bs1,\bs2,\bs3,\bs4$ and $\bs5$ are in $\sep$.		        \hfill(combinators)
  \item If $a\arp b \in \sep$ and $a\in\sep$ then $b\in\sep$.	\hfill(modus ponens)
 \end{enumerate}
 A separator $\sep$ is said to be \emph{consistent} if $\bot\notin\sep$.
\end{definition}

\begin{remark}[Generalized Modus Ponens]
 \label{lm:mod_pon_inf}
The modus ponens, that is the unique rule of deduction we have, is actually compatible with meets.
Consider a set $I$ and two families $(a_i)_{i\in I},(b_i)_{i\in I} \in \A^I$,
we have:
$$\infer{\vdash_I b}{a \vdash_I b & \vdash_I a}$$
where we write $a\vdash_{I} b$ for $(\minf_{i\in I}a_i\arp b_i) \in \sep$ and 
$\vdash_{I} a$ for $(\minf_{i\in I}a_i) \in \sep$.
The proof is straightforward using that the separator is closed upwards and by application, and that:{} 
$$\hspace{-0.1cm}\begin{array}{l@{~~}l}
   &(\minf_{i\in I}a_i\arp b_i) (\minf_{i\in I}a_i) \leq (\minf_{i\in I}b_i)\\
   \Leftarrow & (\minf_{i\in I}a_i\arp b_i)  \leq (\minf_{i\in I}a_i) \to (\minf_{i\in I}b_i)\\
  \end{array}\eqno\begin{array}{r}\\\text{\small (by \hyperref[p:adj]{adj.})}\end{array}$$
which is clearly true. 
\end{remark}

\begin{example}[Realizability model]\label{ex:real_lpar}
Recall from \Cref{ex:lpar_ds} that any model of classical realizability based on the $\Lpar$-calculus induces a disjunctive structure.
As in the implicative case, the set of formulas realized by a closed term\footnote{Proof-like terms in $\Lpar$ 
simply correspond to closed terms.}:
$$\sep_\pole\defeq \{a\in\P(\V_0): a^\pole \cap \T_0 \neq \emptyset\}$$
defines a valid separator. 
The conditions  (1) and (3) are clearly verified (for the same reasons as in the implicative case), but we should verify that
the formulas corresponding to the combinators are indeed realized.

Let us then consider the following closed terms:
{
$$
\begin{array}{l}
 PS_1 \defeq \mu([x],\alpha).\cut{ x}{(\alpha,\alpha)} \\
 PS_2 \defeq \mu([x],\alpha).\cut{\mu(\alpha_1,\alpha_2).\cut{ x}{\alpha_1}}{\alpha} \\
 PS_3 \defeq \mu([x],\alpha).\cut{\mu(\alpha_1,\alpha_2).\cut{ x}{(\alpha_2,\alpha_1)}}{\alpha} \\
 PS_4 \defeq \mu([x],\alpha).\cut{\mu([ y],\beta).\cut{\mu(\gamma,\delta).					\cut{y}{(\gamma,\mu z.\cut{ x}{([ z],\delta)}}}{\beta}}{\alpha} \\
 PS_5 \defeq \mu([x],\alpha).\cut{\mu(\beta,\alpha_3).\cut{\mu(\alpha_1 ,\alpha _2).\cut{ x}{(\alpha_1,(\alpha_2,\alpha_3))}}{\beta}}{\alpha}
\end{array}
$$
}
\begin{proposition}\label{p:real_paralg}
The previous terms have the following types in $\Lpar$:
\begin{enumerate}
 \item $\vdash PS_1:\forall A. (A \parr A )\to A\mid$
 \item $\vdash PS_2:\forall A B. A \to A\parr B\mid$
 \item $\vdash PS_3:\forall A B. A\parr B \to B\parr A\mid$
 \item $\vdash PS_4:\forall A B C. (A \to B) \to (C\parr A \to C\parr B)\mid$
 \item $\vdash PS_5:\forall A B C. (A\parr (B \parr C)) \to ((A\parr B) \parr C)\mid$
\end{enumerate}
 \label{prop:S_types}
\end{proposition}
\begin{proof}
 Straightforward typing derivations in $\Lpar$.
\end{proof}
We deduce that $\S_\pole$ is a valid separator:
\begin{proposition}
 The quintuple $(\P(\V_0),\leq,\parr,\neg,\sep_\pole)$ as defined above is a disjunctive algebra.
\end{proposition}
\begin{proof}
Conditions (1) and (3) are trivial.
Condition (2) follows from the previous proposition and the adequacy lemma for the realizability interpretation of  $\Lpar$ (\Cref{a:p:lpar:adequacy}).
\end{proof}
\end{example}
\subsection{Internal logic}
From the combinators, we directly get that:
\begin{proposition}[Combinators] For all $a,b,c \in \A$, the following holds: \noem
\begin{multicols}{2}
\begin{enumerate}
   \coqitem[S1_p] $ (a \parr a) \vdash_\sep a            $ 
  \coqitem[S2_p] $a \vdash_\sep (a \parr b)             $  
  \coqitem[S3_p]  $(a \parr b) \vdash_\sep (b \parr a)  $ 
  \coqitem[S4_p]  $(a \arp b) \vdash_\sep (c \parr a) \arp (c \parr b)$ 
  \coqitem[S5_p]  $ a \parr  (b\parr c) \vdash_\sep (a \parr b)  \parr c$ \\
  
\end{enumerate}
\end{multicols}
\end{proposition}

 \begin{proposition}[Preorder]
 \label{a:po}
 For any $a,b,c\in\A$, we have:\nomidem
 \begin{enumerate}
  \coqitem[PI] $ a\vdash_\sep a$\hfill(Reflexivity)
  \coqitem[C6_p] if $a\vdash_{\sep} b$ and $b\vdash_{\sep} c$ then $a\vdash_{\sep} c$ \hfill(Transitivity)
\end{enumerate}
\end{proposition}
\begin{proof}
 We first that (2) holds by applying twice the closure by modus ponens, then we use it with the relation
 $a\vdash_\sep a \parr a$ and $a\parr a \vdash_\sep$ that can be deduced from the combinators $\ps1,\ps2$ to get 1.
\end{proof}
 
\subsubsection*{Negation}
We can relate the primitive negation to the one induced by the underlying implicative structure:
\begin{proposition}[Implicative negation]
 For all $a\in \A$, the following holds:\nomidem
 \begin{multicols}{2}
  \begin{itemize}
   \item[\purl{neg_imp_bot}{1.}~] $\neg a\vdash_\sep a \arp \bot$ 
   \item[\purl{imp_bot_neg}{2.}~] $ a\arp\bot \vdash_\sep \neg a $
  \end{itemize}
 \end{multicols}\label{p:imp_neg} 
 \noem
\end{proposition}
\begin{proof}
The first item follows directly from $\ps2$ belongs to the separator,
since $a \arp \bot = \neg a \parr \bot$ and that 
$\neg a \vdash_\sep \neg a \parr \bot$.

For the second item, we use the transitivity with the following hypotheses:
 $$(a \arp \bot) \vdash_\sep a \arp \neg a \qquad\qquad (a \arp \neg a) \vdash_\sep \neg a$$
The statement on the left hand-side is proved by subtyping from the identity ($\minf_{a\in\A} (a \arp a)$, 
which is in $\sep$ as the generalized version of $a\vdash_\sep a$ above).
On the right hand-side, we use twice the modus ponens to prove that
$$(a \arp a) \vdash_\sep (\neg a \arp \neg a) \arp (a \arp \neg a) \arp \neg a$$
The two extra hypotheses are trivially subtypes of the identity again.
This statement follows from this more general property (recall that $a\to a = \neg a \parr a$):
$$\minf_{a,b\in\A} ((a\parr b) \arp a + b)\in\S $$
 that we shall prove thereafter (see \Cref{a:p:par_or}).
\end{proof}

Additionally, we can show that the principle of double negation elimination is valid with respect to any separator:
\begin{proposition}[Double negation]\label{a:p:dne}
 For all $a\in \A$, the following holds:
 \begin{multicols}{2}
  \begin{itemize}
   \item[\purl{dni_entails}{1.}~] $a\vdash_\sep \neg\neg a$ 
   \item[\purl{dne_entails}{2.}~] $\neg \neg a \vdash_\sep a $
  \end{itemize}
 \end{multicols}
\end{proposition}
\begin{proof}
The first item is easy since for all $a\in\A$, we have $a\arp \neg\neg a = (\neg a)\parr \neg\neg a \seq \neg\neg a \parr \neg a = \neg  a \arp \neg a$.
As for the second item, we use Lemma \ref{a:po} and Proposition \ref{p:imp_neg} to it reduce to the statement:
$$\minf_{a\in\A} ((\neg a) \arp \bot) \arp a~~\in\sep$$
We use again  Lemma \ref{a:po} to prove it, by showing that:
$$\minf_{a\in\A} ((\neg a) \arp \bot) \arp  (\neg a) \arp a~~\in\sep
\qquad\qquad
\minf_{a\in\A} ((\neg a) \arp a) \arp  (\neg a) \arp a~~\in\sep
$$
where the statement on the left hand-side from by subtyping from the identity. For the one on the right hand-side, 
we use the same trick as in the last proof in order to reduce it to:
$$\minf_{a\in\A}(a \arp \neg a) \arp ( a \arp a) \arp (\neg a \arp  a) \arp a)~\in\sep$$
\end{proof}

\subsubsection*{Sum type}
As in implicative structures, we can define the sum type by:
$$a+b \defeq \minf_{c\in\A} ((a\arp c) \arp (b\arp c) \arp c) \eqno (\forall a,b\in\A)$$
We can prove that the disjunction and this sum type are equivalent from the point of view of the separator:
\begin{proposition}[Implicative sum type] \label{a:p:par_or}
 For all $a,b\in \A$, the following holds:
 \begin{multicols}{2}
  \begin{itemize}
   \item[\purl{par_or}{1.}~] $a\parr b \vdash_\sep a + b$ 
   \item[\purl{or_par}{2.}~] $a +  b \vdash_\sep a \parr b$
  \end{itemize}
 \end{multicols}\label{p:imp_sum}
\end{proposition}
\begin{proof}
 We prove in both cases a slightly more general statement, namely that the meet over all $a,b$ or the corresponding implication belongs to the separator.
 For the first item, we have:
 $$\minf_{a,b\in\A} (a\parr b) \arp a + b = \minf_{a,b,c\in\A} (a\parr b) \arp (a\arp c) \arp (b\arp c) \arp c$$
 Swapping the order of the arguments, we prove that $\minf_{a,b,c\in\A} (b\arp c) \arp (a\parr b) \arp (a\arp c) \arp  c\in\sep$.
 For this, we use Lemma \ref{a:po} and the fact that:
 $$\minf_{a,b,c\in\A} (b\arp c) \arp (a\parr b) \arp (a\parr c)~\in\sep \qquad\qquad\qquad \minf_{a,c\in\A} (a\parr c) \arp (a\arp c) \arp  c~ \in\sep$$
 The left hand-side statement is proved using $\ps4$, while on the right hand-side we prove it from the fact that:
 $$\minf_{a,c\in\A}  (a\arp c) \arp (a\parr c) \arp c\parr c ~ \in\sep$$
 which is a subtype of $\ps4$, by using Lemma \ref{a:po} again with $\ps1$ and by manipulation on the order of the argument.
 
 The second item is easier to prove, using Lemma \ref{a:po} again and the fact that:
 $$ \minf_{a,b\in\A} a + b \arp  (a\arp (a\parr b)) \arp (b\arp (a\parr b)) \arp (a\parr b) \in \sep$$
 which is a subtype of $\comb{I}^\A$ (which belongs to $\sep$).  
 The other part, which is to prove that:
 $$ \minf_{a,b\in\A} ( (a\arp (a\parr b)) \arp (b\arp (a\parr b)) \arp (a\parr b))\arp (a\parr b)~~ \in \sep$$
 follows from Lemma \ref{lm:mod_pon_inf} and the fact that $\minf_{a,b\in\A}  (a\arp (a\parr b))$ and $\minf_{a,b\in\A}  (b\arp (a\parr b))$ are both in the separator.
 
\end{proof}
\subsection{Induced implicative algebras}
\label{a:s:da_ia}
\begin{proposition}[Combinator $\comb{k}^\A$][psep_K]
 We have $\comb{k}^\A\in\sep$.
\end{proposition}
\begin{proof}
 This directly follows by upwards closure from the fact that $\minf_{a,b\in\A} (a\arp (b\parr a)) \in\sep$.
\end{proof}

\begin{proposition}[Combinator $\comb{s}^\A$][psep_S]
 For any disjunctive algebra $(\A,\leq,\parr,\neg,\sep)$, we have $\comb{s}^\A\in\sep$.
 \end{proposition}
\begin{proof}
See \Cref{a:da}.
 We make several applications of Lemmas \ref{a:po} and \ref{lm:mod_pon_inf} consecutively. 
 We prove that:
 $$\minf_{a,b,c\in\A} ((a\to b \to c)\to (a\to b) \to a \to c)                            ~\in\sep   $$
 is implied by Lemma \ref{a:po} and the following facts:
 \begin{enumerate}
  \item $
\minf_{a,b,c\in\A} ((a\to b \to c)\to (b\to a \to c))                             ~\in\sep$
\item $\minf_{a,b,c\in\A} ((b\to a \to c)\to (a\to b) \to a \to c)                             ~\in\sep$
\end{enumerate}
The first statement is an ad-hoc lemma, while the second is proved by generalized transitivity (Lemma \ref{lm:mod_pon_inf}),
using a subtype of $\ps4$ as hypothesis, from:
$$
\minf_{a,b,c\in\A} ((a\to b)\to (a\to a \to c))\to (a\to b) \to a \to c    \in\sep
$$
The latter is proved, using again generalized transitivity with
 a subtype of $\ps4$ as premise, from:
$$
\minf_{a,b,c\in\A} (a\to a \to c) \to (a \to c) \in\sep
$$
This is proved using again Lemmas \ref{a:po} and \ref{lm:mod_pon_inf} with $\ps5$ and a variant of $\ps4$.
\end{proof}

\begin{proposition}[Combinator $\comb{cc}^\A$][psep_cc]
 We have $\comb{cc}^\A\in\sep$.
\end{proposition}
\begin{proof}
 We make several applications of Lemmas \ref{a:po} and \ref{lm:mod_pon_inf} consecutively. 
 We prove that:
 $$\minf_{a,b\in\A} ((a\to b) \to a)\to a                             ~\in\sep   $$
 is implied by generalized modus ponens (Lemma \ref{a:po}) and:
$$
\begin{array}{c}\minf_{a,b\in\A} ((a\to b) \to a)\to (\neg a \arp a \arp b) \arp \neg a \arp a                             ~\in\sep   \\[0.5em]

\minf_{a,b\in\A} ((\neg a \arp a \arp b) \arp \neg a \arp a)\arp a                           ~\in\sep   
\end{array}\leqno\begin{array}{l}\\[0.5em]\text{and}\qquad\end{array}
$$
The statement above is  a subtype of $\ps4$,
while the other is proved, by Lemma \ref{a:po}, from:
$$
\begin{array}{c}
\minf_{a,b\in\A} ((\neg a \arp a \arp b)\arp \neg a \arp a) \arp \neg a \arp a                            ~\in\sep   \\[0.5em]
\minf_{a\in\A} ((\neg a ) \arp a ) \arp a                           ~\in\sep   
\end{array}\leqno\begin{array}{l}\\[0.5em]\text{and}\qquad\end{array}
$$
The statement below is proved as in \Cref{p:dne}, while the statement above
is proved by a variant of the modus ponens and:
$$
\minf_{a,b\in\A} (\neg a \arp a \arp b)         ~\in\sep 
$$
We conclude by proving this statement using the connections between $\neg a$ and $a \to \bot$, reducing the latter to:
$$\minf_{a,b\in\A} (a \arp\bot) \arp a \arp b         ~\in\sep $$
which is a subtype of the identity.
\end{proof}

As a consequence, we get the expected theorem:
\begin{theorem}[][PA_IA]\label{a:pa_ia}
 Any disjunctive algebra is a classical implicative algebra.
\end{theorem}
\begin{proof}
 The conditions of upward closure and closure under modus ponens coincide for implicative and disjunctive separators, 
 and the previous propositions show that $\comb{k},\comb{s}$ and $\comb{cc}$ belong to the separator of any disjunctive algebra.
\end{proof}
\begin{corollary}
 If $t$ is a closed $\lambda$-term and $(\A,\leq,\parr,\neg,\sep)$ a disjunctive algebra, then $t^\A\in\sep$.
\end{corollary}

\newpage

\section{Conjunctive algebras}
\label{a:ta}
The \Ltens calculus corresponds exactly to the restriction of $\L$ to the positive fragment 
induced by the connectives $\tensor,\neg$ and the existential quantifier $\exists$~\cite{Munch09}.
Its syntax is given by:
$$
\begin{array}{c|c}
 \begin{array}{@{\hspace{-0.15cm}}l@{\quad}c@{~}c@{~}l}
  \text{\bf Terms}  & t & ::= & x \mid (t,t) \mid [e] \mid  \mu \alpha . c  \\
\text{\bf Values} & V & ::= & x \mid (V,V) \mid [e]\\ 
\end{array}~&~
\begin{array}{l@{\quad}c@{~}c@{~}l}  
\text{\bf Contexts} & e & ::= & \alpha \mid \mu (x,y).c \mid \mu [\alpha].c \mid \mu x.c\\
\text{\bf Commands} & c & ::= & \cut{t}{e}  
\end{array}
\end{array}
  $$
  
We denote by $\V_0$, $\T_0$, $\E_0$, $\C_0$ for the sets of closed values, terms, contexts and commands.
The syntax is really close to the one of $\Lpar$:
it has the same constructors but on terms, while destructors are now evaluation contexts.
We recall the meanings of the different constructions:
\begin{itemize}
 \item $(t,t)$ are pairs of positive terms;
 \item $\mu( x_1, x_2).c$, which binds the variables $x_1,x_2$, is the dual destructor;
 \item $[e]$ is a constructor for the negation, which allows us to embed a negative context intro a positive term;
 \item $\mu[x].c$, which binds the variable $x$, is the dual destructor;
 \item $\mu \alpha.c$ and  $\mu x.c$ are unchanged.
\end{itemize}
The reduction rules correspond again to the intuition one could have from the syntax of the calculus:
 all destructors and  binders reduce in front of the corresponding values, while pairs of terms
are expanded if needed:
$$
\begin{array}{c|c}
 \begin{array}{rcl}
  \cut{\mu\alpha.c}{e} 		& \to & c[e/\alpha] \\
  \cut{[e]}{\mu[ \alpha].c}	& \to & c[e/\alpha ] \\
  \cut{V}{\mu x.c} 		& \to & c[V/x] \\
  \end{array}
  &
  \begin{array}{rcl}
  \cut{(V,V')}{\mu( x, x').c} 	& \to & c[V/ x,V'/ x'] \\
  \cut{(t,u)}{e} 		& \to & \cut{t}{\mu x.\cut{u}{\mu y.\cut{(x,y)}{e}}} \\
  \end{array}
  \end{array}
$$
where $(t,u)\notin V$ in the last $\beta$-reduction rule.

Finally, we shall present the type system of $\Ltens$.
Once more, we are interested in the second-order formulas defined from the positive connectives:
$$A,B := X \mid A \tensor B \mid \neg A \mid \exists X.A \leqno\quad \textbf{Formulas}$$

We still work with threetwo-sided sequents, where typing contexts are defined as finite lists of bindings between variable and formulas:
$$\Gamma ::= \varepsilon \mid \Gamma,x:A \qquad\qquad\qquad \Delta ::= \varepsilon \mid \Delta,\alpha:A$$
Sequents are again of three kinds, as in the \lmmt-calculus and $\Lpar$:
\begin{itemize}
 \item $\Gamma\vdash t:A \mid \Delta$ for typing terms,
 \item $\Gamma\mid e:A\vdash  \Delta$ for typing contexts,
 \item $c: \Gamma\vdash  \Delta$ for typing commands.
\end{itemize}

The type system, which uses the same three kinds of sequents for terms, contexts, commands than in {\Lpar} type system,
is given in Figure~\ref{fig:Ltens:type}.
\begin{figure}[t]
\myfig{\small
 $$
  \begin{array}{c}
{\infer[\autorule{\textsc{Cut} }]{ \cut{t}{e}:\Gamma \vdash\Delta}{\Gamma \vdash t :A \mid \Delta & \Gamma \mid e:A\vdash \Delta}}  
\qquad
\infer[\!\!\autorule{ax\vdash     }]{\Gamma\mid \alpha : A \vdash \Delta}{(\alpha: A)\in\Delta} 
\qquad 
 \infer[\!\!\autorule{\vdash ax    }]{\Gamma\vdash x : A \mid \Delta}{( x: A)\in\Gamma}  \\[\vsep]
 
 \infer[\!\!\autorule{\mu\,\vdash  }]{\Gamma\mid \mu x.c : A \vdash \Delta}{c:\Gamma\vdash\Delta, x: A}  
 \qquad 
 \infer[\!\!\autorule{\tensor\,\vdash}]{\Gamma\mid \mu( x, x').c: A \tensor B \vdash \Delta}{c:(\Gamma, x:A, x':B\vdash\Delta)} 
 \qquad
  \infer[\!\!\autorule{\neg\,\vdash }]{\Gamma\mid \mu[\alpha].c:\neg A}{c:\Gamma\vdash\Delta,\alpha:A}  
 \\[\vsep]
 
 \infer[\!\!\autorule{\vdash\mu    }]{\Gamma\vdash \mu\alpha.c : A \mid \Delta}{c:\Gamma\vdash\Delta,\alpha:A} 
 \qquad 
 \infer[\!\!\autorule{\vdash\tensor}]{\Gamma\vdash (t,u):A\tensor B\mid \Delta}{\Gamma\vdash t:A \mid \Delta & \Gamma\vdash u:B \mid \Delta}  
 \qquad 
 \infer[\!\!\autorule{\vdash\neg   }]{\Gamma\vdash [e]:\neg A\vdash\Delta}{\Gamma\mid e:A  \vdash \Delta}
 \\[\vsep]
 
 \infer[\!\!\exists\vdash ]{\Gamma\mid e:\exists X. A\vdash\Delta}{\Gamma \mid e:A \vdash \Delta& X\notin{FV(\Gamma,\Delta)}} 
 \qquad
 \quad\infer[\!\!\vdash\exists ]{\Gamma\vdash V:\exists X. A\mid\Delta}{\Gamma\vdash V: A[B/X]\mid\Delta }
 \end{array}
 $$
 }
 \caption{Typing rules for the $\Ltens$-calculus}
 \label{fig:Ltens:type}
\end{figure}

\setCoqFilename{ImpAlg.TensorAlgebras}
\subsection{Embedding of the $\lambda$-calculus}
Guided by the expected definition of the arrow $A \imp B~\defeq~ \neg( A \tensor \neg B)$,
we can follow Munch-Maccagnoni's paper~\cite[Appendix E]{Munch09},
to embed the $\lambda$-calculus into $\Ltens$.

With such a definition, a stack $u\cdot e$ in $A\imp B$ (that is with $u$ a term of type $A$ and $e$ a context of type $B$)
is naturally embedded as a term $(u,[e])$, which is turn into the context
$\mu[\alpha].\cut{(u,[e])}{\alpha}$ which indeed inhabits the ``arrow'' type $\neg (A \tensor \neg B)$.
The rest of the definitions are then direct:
$$\begin{array}{r@{~}c@{~}l}
 \mu( x,[\alpha]).c & \defeq & \mu( x, x').\cut{ x'}{\mu[\alpha].c} \\
 \lambda  x.t	& \defeq & [\mu( x,[\alpha]).\cut{t}{\alpha}]\\
\end{array}
$$
These shorthands allow for the expected typing rules:
\begin{proposition}\label{prop:lmmt_typing_tens}
The following typing rules are admissible:
$$
\infer{\Gamma\vdash \lambda x.t:A\imp B}{\Gamma,x:A\vdash t:B}
\qquad\quad
\infer{\Gamma\mid u\cdot e: A\imp B\vdash\Delta}{\Gamma\vdash u:A\mid\Delta & \Gamma \mid e:B\vdash\Delta} 
\qquad\quad
\infer{\Gamma\vdash t\,u:B\mid\Delta}{\Gamma\vdash t:A\imp B\mid\Delta & \Gamma \vdash u :A\mid\Delta}
$$
\end{proposition}

\begin{proof}
\renewcommand{\prfcase}[1] {\paragraph*{\normalfont\textbullet~ \textbf{Case}  #1:}}
Each case is directly derivable from $\Ltens$ type system. 
\newcommand{\defrule}{\autorule{\textit{def}}}
We abuse the notations to denote by {\defrule} a rule which simply consists in unfolding the shorthands defining the $\lambda$-terms.
\prfcase{$\mu(x,[\alpha]).c$}
$$\infer[\defrule]{\Gamma\mid\mu(x,[\alpha]).c: A \tensor \neg B\vdash \Delta}{
 	\infer[\autorule{\tensor\,\vdash}]{\Gamma\mid \mu(x,x').\cut{x'}{\mu[\alpha].c}:A \tensor\neg  B\vdash\Delta}{
 	  \infer[\cutrule]{\cut{x'}{\mu[\alpha].c}:(\Gamma,x:A,x':\neg B\vdash\Delta)}{
        \infer[\axrrule]{\Gamma,x:A,x':\neg B\vdash x':\neg B\mid \Delta}{}
        &
        \infer[\mulrule]{\Gamma\mid \mu[\alpha].c:\neg B\vdash\Delta}
 	      {c:(\Gamma,x:A\vdash \Delta,\alpha:B)}
 	  }
 	}
    }	
 $$
 \prfcase{$\lambda x.t$}
 $$\infer[\defrule]{\Gamma\vdash \lambda x.t:A\imp B\mid \Delta}{
    \infer[\negrrule]{\Gamma\vdash  [\mu( x,[\beta]).\cut{t}{\beta}]:\neg(A\tensor \neg B)\mid \Delta}{
	\infer[]{\Gamma\mid \mu(x,[\beta]).\cut{t}{\beta}:A \tensor\neg B\vdash \Delta}{
	  \infer[\cutrule]{\cut{t}{\beta}:(\Gamma,x:A\vdash\beta:B,\Delta)}{
	    \Gamma,x:A\vdash t:B \mid \Delta
	    & 
	    \infer[\axlrule]{\Gamma\mid \beta:B \vdash \Delta,\beta:B}{}
	  }
	}
     }
    }	
 $$
\prfcase{$u\cdot e$} 
 $$
 \infer[\defrule]{\Gamma\mid u\cdot e: A\imp B\vdash\Delta}{
  \infer[\neglrule]{\Gamma \mid \mu[\alpha].\cut{(u,[e])}{\alpha}:\neg(A\tensor \neg B)\vdash \Delta}{
    \infer[\cutrule]{\cut{(u,[e])}{\alpha}:(\Gamma \vdash \Delta,\alpha:A\tensor \neg B)}{
      \infer[\autorule{\vdash \,\tensor}]{\Gamma\vdash (u,[e]):   A \tensor \neg B\mid\Delta}{
	\Gamma \vdash u:A\vdash\Delta
	&
	\infer[\negrrule]{\Gamma\vdash [e]:\neg B\mid\Delta}{\Gamma\mid e:B\vdash\Delta}
	}
      &
      \infer[\axlrule]{\Gamma\mid\alpha:(A\tensor \neg B)\vdash \Delta,\alpha:(A\tensor \neg B)}{}
     } 
   }    
  }
 $$
\prfcase{$t\,u$} 
 $$
 \infer[\defrule]{\Gamma\vdash t\,u:B\mid\Delta}{
  \infer[\murrule]{\Gamma\vdash \mu\alpha.\cut{t}{u\cdot\alpha}:B\mid\Delta}{
   \infer[\cutrule]{\cut{t}{u\cdot\alpha}:(\Gamma\vdash\Delta,\alpha:B)}{
    \Gamma\vdash t:A\imp B\mid \Delta
    &
    \infer{\Gamma\mid u\cdot\alpha: A\imp B\vdash \Delta,\alpha:B}{
      \Gamma\vdash u:A\mid\Delta
      & 
      \infer{\Gamma \mid \alpha:B\vdash\Delta,\alpha:B}{}
    }
   } 
  }
 }
$$
\end{proof}

Besides, the usual rules of $\beta$-reduction for the call-by-value evaluation strategy
are simulated through the reduction of $\Ltens$:
\begin{proposition}[$\beta$-reduction]
We have the following reduction rules:
$$
\begin{array}{rcl}
  \cut{t\,u}{e} 		& \bred & \cut{t}{u\cdot e}\\
  \cut{\lambda x.t}{u\cdot e} 	& \bred & \cut{u}{\mu x.\cut{t}{e}}\\
  \cut{V}{\mu x.c} 		& \bred & c[V/x]\\
\end{array}
$$
\end{proposition}
 
\begin{proof}
The third rule is included in $\Ltens$ reduction system, the first follows from:
 $$\cut{t u}{e} = \cut{\mu \alpha.\cut{t}{u\cdot \alpha}}{e} \bred \cut{t}{u\cdot e}$$
For the second rule, we first check that we have:
 $$\cut{(V,[e])}{\mu( x,[\alpha]).c} = \cut{(V,[e])}{\mu( x,x').\cut{x'}{\mu [\alpha].c}}
  \bred \cut{[e]}{\mu [\alpha].c[V/X]}
  \bred c[V/x][e/\alpha]$$
from which we deduce:  
 \begin{align*}
\cut{\lambda x.t}{u\cdot e} & = \cut{ [\mu( x,[\alpha]).\cut{t}{\alpha}]}{\mu[\alpha].\cut{(u,[e])}{\alpha}} \\
			    & \bred \cut{(u,[e])}{\mu( x,[\alpha]).\cut{t}{\alpha}}\\
			    & \bred \cut{u}{\mu y.\cut{(y,[e])}{\mu( x,[\alpha]).\cut{t}{\alpha}}} \\
			    & \bred \cut{u}{\mu x.\cut{t}{e}} 			    
 \end{align*}
\end{proof}

\subsection{A realizability model based on the $\Ltens$-calculus}
We briefly recall the definitions necessary to the realizability interpretation \emph{à la} Krivine
of $\Ltens$. Most of the properties being the same as for $\Lpar$,
we spare the reader from a useless copy-paste and go straight to the point.
A \emph{pole} is again defined as any subset of $\C_0$ closed by anti-reduction.
As usual in call-by-value realizability models~\cite{Lepigre16},
formulas are primitively interpreted as sets of values, which we call \emph{ground truth values},
while \emph{falsity values} and \emph{truth values} are then defined by orthogonality.
Therefore, an existential formula $\exists X.A$ is interpreted by the union over all the possible
instantiations for the primitive truth value of the variable $X$ by a set $S\in\P(\V_0)$.
We still assume that for each subset $S$ of $\P(\V_0)$, there is a constant symbol $\dot S$ in the syntax. 
The interpretation is given by:
$$
\begin{array}{c|c}
\begin{array}{r@{~~}c@{~~}ll}
 \tvval{\dot S}    & \defeq & S& \hspace{-1.6cm}(\forall S \in\P(\V_0))\\
 \tvval{A\tensor B}& \defeq & \{(t,u):t\in\tvval{A} \land u \in \tvval{B}\}\\
 \tvval{\neg A}    & \defeq & \{[e]:e\in\fval{A}\}\\
\end{array}&
\begin{array}{r@{~~}c@{~~}ll}
 \tvval{\exists X. A}& \defeq & \bigcup_{S\in\P(\V_0)}\tvval{A\{X:=\dot S\}}\\
 \fval{A}          & \defeq & \{e:\forall V\in\tvval{A}, V\pole e\}\\
 \tval{A}          & \defeq & \{t:\forall e\in\fval{A},  t\pole e\}\\
\end{array}
\end{array}
$$
 We define again a \emph{valuation} $\rho$ as a function mapping each second-order variable 
 to a primitive falsity value $\rho(X)\in\P(\V_0)$. 
 In this framework, we say that a \emph{substitution} $\sigma$ is a function mapping each variable $x$ to a closed value $V\in\V_0$ and
 each variable $\alpha$ to a closed context $e\in\E_0$.
 We write $\sigma \real \Gamma $ and we say that a substitution $\sigma$ realizes a context $\Gamma$, when
 for each binding $(x:A)\in\Gamma$, we have $\sigma(x)\in\tvval{A}$. Similarly, we say that $\sigma$ realizes a context $\Delta$
 if for each binding $(\alpha:A)\in\Delta$, we have $\sigma(\alpha)\in\fval{A}$.
\begin{proposition}[Adequacy~\cite{Munch09}]
Let $\Gamma,\Delta$ be typing contexts, $\rho$ be a valuation and $\sigma$ be a substitution which verifies that $\sigma\real\Gamma[\rho]$ and $\sigma\real\Delta[\rho]$. We have:
\noem
\begin{multicols}{2}
\begin{enumerate}
 \item If~ $ \vdash V:A \mid $, then $V \in\tvval{A}$.
 \item If~ $ \mid e:A \vdash $, then $e \in\fval{ A}$.
 \item If~ $ \vdash t:A \mid $, then $t \in\tval{ A}$.
 \item If~$c: (\vdash )$,       then $c \in \pole$.  
\end{enumerate}
\end{multicols}
\end{proposition}
\begin{myproof}
By induction over typing derivations, see ~\cite{Munch09}.
\end{myproof}

\subsection{Conjunctive structures}
\begin{proposition}[Commutations]\label{a:ts:comm}
In any $\Ltens$ realizability model, if $X\notin \FV(B)$ the following equalities hold:
\begin{enumerate}
 \item $\tvval{\exists X. (A \tensor B)}= \tvval{(\exists X. A) \tensor B}$.
 \item $\tvval{\exists X. (B \tensor A)}= \tvval{B \tensor (\exists X.A)}$.
 \item $\tvval{\neg{(\exists X. A)}}= \bigcap_{S\in\P(\V_0)}\tvval{\neg{A\{X:=\dot S\}}}$
 \end{enumerate}
\end{proposition}
\begin{proof}
\begin{enumerate}
  \item Assume the $X\notin \FV(B)$, then we have:
 \begin{align*}
  \tvval{\exists X. (A \tensor B)} &=
  \union{S\in\P(\V_0)} \tvval{A\{X:=\dot S\} \tensor B} \\
  &= \union{S\in\P(\V_0)}\{(V_1,V_2):V_1\in\tvval{A\{X:=\dot S\}} \land V_2 \in \tvval{B}\} \\
  &=\{(e_1,e_2):e_1\in\union{S\in\P(\V_0)}\tvval{A\{X:=\dot S\}} \land e_2 \in \tvval{B}\}\\
  &=\{(e_1,e_2):e_1\in\tvval{\exists X.A} \land e_2 \in \fv{B}\}
  ~~=~~\tvval{(\exists X. A) \tensor B}
  \end{align*}
  \item Identical.
\item The proof is again a simple unfolding of the definitions: 
 \begin{align*}
  \tvval{\neg{(\exists X. A})}
  & = \{[t]:t\in\tval{\exists X.A}\}
  ~~~ =~~~ \{[t]:t\in\bigcap_{S\in\P(\V_0)}\tval{A\{X:=\dot S\}}\}\\
  & = \bigcap_{S\in\P(\V_0)}\{[t]:t\in\tval{A\{X:=\dot S\}]}\}
  = \bigcap_{S\in\P(\V_0)}\tvval{\neg{A\{X:=\dot S\}}}
  \end{align*}
 \end{enumerate}
\end{proof}

\begin{repexample}{22}[Realizability models]
As for the disjunctive case, we can abstract the structure of the realizability interpretation of $\Ltens$
into a structure of the form $(\T_0,\E_0,\V_0,(\cdot,\cdot),[\cdot],\pole)$, where
$(\cdot,\cdot)$ is a map from $\T_0^2$ to $\T_0$ (whose restriction to $\V_0$ has values in $\V_0$),
$[\cdot]$ is an operation from $\E_0$ to $\V_0$,
and $\pole\subseteq \T_0\times \E_0$ is a relation.
From this sextuple we can define:
$$\begin{array}{l@{\qquad\qquad\quad}l}
 \text{\textbullet}~~\A \defeq \P(\V_0)           & \text{\textbullet}~~a \tensor b \defeq (a,b)=\{(V_1,V_2): V_1\in a \land V_2 \in b\}\qquad\qquad \\
 \text{\textbullet}~~a \leq b \defeq a \subseteq b  & \text{\textbullet}~~\neg a \defeq [a^\orth]=\{[e]: e\in a^\orth\}
\end{array} 
\eqno 
 (\forall a,b\in\A)
 $$

\begin{proposition}\label{prop:ltens_ts}
 The quadruple $(\A,\leq,\tensor,\neg)$ is a conjunctive structure.
\end{proposition}
\begin{proof}
We show that the axioms of Definition~\ref{def:conj_struct} are satisfied.
 \begin{enumerate}
 \item Anti-monotonicity. Let $a,a'\in\A$, such that 
 $a\leq a'$ ie $a\subseteq a'$. Then ${a'}^\orth \subseteq a^\orth$ and thus 
 $$\neg a' =\{[t]: t\in {a'}^\orth\} \subseteq \{[t]: t\in a^\orth\} = \neg a$$
 \emph{i.e.} $\neg a' \leq \neg a$.
 
 \item Covariance of the conjunction. Let $a,a',b,b'\in\A$ such that $a'\subseteq a$ and $b'\subseteq b$.
 Then we have 
 $$a \tensor b = \{(t,u):t\in a \land u \in b\}\subseteq \{(t,u):t\in a' \land u \in b'\} = a' \tensor b'$$
 \emph{i.e.} $a\tensor b \leq a'\tensor b'$
  \item Distributivity. Let $a\in\A$ and $B\subseteq \A$, we have:
    $$\msup_{b\in B} (a \tensor b) = \msup_{b\in B} \{(v,u):t\in a \land u \in b\}= \{(t,u):t\in a \land u \in \msup_{b\in B} b\}= a \tensor (\msup_{b\in B}  b)$$
 \item Commutation. Let $B\subseteq \A$, we have (recall that $\minf_{b\in B} b = \bigcap_{b\in B} b$):
       $$\minf_{b\in B} \{ \neg b\} = \minf_{b\in B} \{[t]:t \in b^\pole\}=  \{[t]:t \in \minf_{b\in B} b^\pole\}= \{[t]:t \in (\msup_{b\in B} b)^\pole\}=\neg (\msup_{b\in B} b)$$ 
\end{enumerate}

\end{proof}
\end{repexample}

\subsection{Interpreting $\Ltens$ terms}\label{a:ta:embed}
We shall now see how to embed $\Ltens$ commands, contexts and terms into any conjunctive structure.
For the rest of the section, we assume given a conjunctive structure $(\A,\leq,\tensor,\neg)$.

\setCoqFilename{ImpAlg.LTensor}
\subsubsection{Commands}
Following the same intuition as for the embedding of $\Lpar$ into disjunctive structures,
we define the \emph{commands} $\cut{a}{b}$ of the conjunctive structure $\A$ as the pairs $(a,b)$, and
we define the pole $\pole$ as the ordering relation $\leq$. 
We write $\C_\A=\A\times\A$ for the set of commands in $\A$ and $(a,b)\in\pole$ for $a\leq b$.

We consider the same relation $\cord$ over $\C_\A$, which was defined by:
$$c \cord c' ~~\defeq~~\text{if } c\in \pole  \text{  then  } c'\in \pole \eqno(\forall c,c'\in\C_\A)$$
Since the definition of commands only relies on the underlying lattice
of $\A$, the relation $\cord$ has the same properties as in disjunctive structures and in particular it defines a preorder (see \Cref{s:disj_cmd}).

\subsubsection{Terms}
The definitions of terms are very similar to the corresponding definitions for the dual contexts in disjunctive structures.
\begin{definition}[Pairing][pairing]
 For all $a,b\in\A$, we let $(a,b) \defeq a\tensor b $.
\end{definition}
\begin{definition}[Boxing][box]
For all $a\in\A$, we let $[a] \defeq \neg a $.
\end{definition}
\begin{definition}[$\mu^+$][mup]
$$\mu^+.c \defeq \minf_{a\in\A}\{a:c(a) \in \pole\}$$  
\end{definition}
We have the following properties for $\mu^+$:, whose proofs are trivial:
\begin{proposition}[Properties of $\mu^+$]
\label{p:tmu_p}
For any functions $c,c':\A\to\C_\A$, the following hold:
\begin{itemize}
 \item[\lturl{mup_mon}{1.}] 	 If for all $a\in\A$,   $c(a)\cord c'(a)$, then $\mu^+ .c'      \leq \mu^+ .c $ 	\hfill(Variance)
 \item[\lturl{mup_eta}{2.}] 	 For all $t\in\A$, then $t = \mu^+. (a \mapsto \cut{t}{a}) $		       	\hfill($\eta$-expansion)
 \item[\lturl{mup_beta}{3.}] 	 For all $e\in\A$, then $ \cut{\mu^+.c}{e} \cord c(e)$				\hfill($\beta$-reduction)
\end{itemize}
\end{proposition}
\begin{proof}
 \begin{enumerate}
  \item Direct consequence of Proposition \ref{p:cord_meet}.
  \item[2,3.] Trivial by definition of $\mu^+$.
 \end{enumerate}
\end{proof}

\subsubsection{Contexts}
Dually to the definitions of the (positive) contexts $\mu^+$ as a meet, we define the embedding of 
(negative) terms, which are all binders, by arbitrary joins:
\begin{definition}[{$\mu^-$}][mun]
For all $c:\A\to\C_\A$, we define: 
 $$\mu^-. c \defeq \msup_{a\in\A}  \{a: c(a)  \in\pole\} $$
 \end{definition}
\begin{definition}[{$\mu^{()}$}][mu_pair]
For all $c:\A^2\to\C_\A$, we define: 
 $$\mu^{()}.c 	\defeq  \msup_{a,b\in\A}\{a\tensor b : c(a,b)\in\pole\} $$
\end{definition}
\begin{definition}[{$\mu^{[]}$}][mu_neg]
For all $c:\A\to\C_\A$, we define: 
 $$ \mu^{[]}.c	 \defeq  \msup_{a\in\A}  \{\neg a : c(a)  \in\pole\} $$
\end{definition}

Again, these definitions satisfy variance properties with respect to the preorder $\cord$ and the order relation $\leq$.
Observe that the $\mu^{()}$ and $\mu^-$ binders, which are negative binders catching positive terms, are contravariant with respect to these relations
while the $\mu^{[]}$ binder, which catches a negative context, is covariant.
\begin{proposition}[Variance]
For any functions $c,c'$ with the corresponding arities, the following hold:
\begin{itemize}
 \item[\lturl{mun_mon}{1.}] 	If  $c(a)\cord c'(a)$ for all $a\in\A$,    then $\mu^- .c'   \leq \mu^- .c $
 \item[\lturl{mu_pair_mon}{2.}] If  $c(a,b)\cord c'(a,b)$ for all $a,b\in\A$, then $\mu^{()}.c' \leq \mu^{()}.c$
 \item[\lturl{mu_neg_mon}{3.}]  If  $c(a)\cord c'(a)$ for all $a\in\A$, then $\mu^{[]}.c \leq \mu^{[]}.c'$
\end{itemize}
\end{proposition}
\begin{proof}
 Direct consequences of Proposition \ref{p:cord_meet}.
\end{proof}

The $\eta$-expansion is also reflected  by the ordering relation $\leq$:
\begin{proposition}[$\eta$-expansion]
\label{prop:lpar_eta_n}
For all $t\in\A$, the following holds:
\begin{itemize}
 \item[\lturl{mun_eta}{1.}]      $\mu^-. (a \mapsto \cut{t}{a})         =  t$
 \item[\lturl{mu_pair_eta}{2.}]  $\mu^{()}.(a,b \mapsto \cut{t}{(a,b)})\leq t$
 \item[\lturl{mu_neg_eta}{3.}]   $\mu^{[]}.(a \mapsto \cut{t}{[a]})    \leq t$
\end{itemize}
\end{proposition}
\begin{proof} Trivial from the definitions.  \end{proof}

The $\beta$-reduction is again reflected by the preorder $\cord$ as the property of subject reduction:
\begin{proposition}[$\beta$-reduction]
\label{prop:lpar_bred_n}
For all $e,e_1,e_2,t\in\A$, the following holds:
\begin{itemize}
 \item[\lturl{mun_beta}{1.}]      $\cut{\mu^-.c}{e}        \cord c(e)  $
 \item[\lturl{mu_pair_beta}{2.}]  $\cut{\mu^{()}.c}{(e_1,e_2)} \cord c(e_1,e_2)$
 \item[\lturl{mu_neg_beta}{3.}]   $\cut{\mu^{[]}.c}{[t]}   \cord c(t)  $
\end{itemize}
\end{proposition}
 \begin{proof} Trivial from the definitions.  \end{proof}

\subsection{Adequacy}
We shall now prove that the interpretation of $\Ltens$ is adequate with respect to its type system.
Again, we extend the syntax of formulas to define second-order formulas with parameters by:
$$ A,B ::= a \mid X \mid \neg A \mid A\tensor B \mid \exists X.A  \eqno (a\in\A)$$
This allows us to define an embedding of closed formulas with parameters into the conjunctive structure $\A$;
$$\begin{array}{ccl}
a^\A 		& \defeq & a 				      \\
(\neg A)^\A 	& \defeq & \neg A^\A                           \\
(A\tensor B)^\A 	& \defeq & A^\A \tensor B^\A                         \\
(\exists X.A)^\A& \defeq & \msup_{a\in\A} (A\{X:=a\})^\A     \\
\end{array}
\eqno\begin{array}{r}(\text{if }a\in\A)\\\\\\\end{array}
$$

As in the previous chapter,
we define substitutions, which we write $\sigma$, as functions mapping variables (of terms, contexts and types) to element of $\A$:
$$\sigma ::= \varepsilon \mid \sigma[x\mapsto a]\mid \sigma[\alpha\mapsto a]\mid\sigma[X\mapsto a] \eqno(a\in\A,~ x,X~ \text{variables})$$
We say that a substitution $\sigma$ {realizes} a typing context $\Gamma$, which write $\sigma \Vdash \Gamma$, if for all bindings $(x:A)\in\Gamma$ 
we have $\sigma(x) \leq (A[\sigma])^\A$.
Dually, we say that $\sigma$ realizes $\Delta$ if for all bindings $(\alpha:A)\in\Delta$ , we have $\sigma(\alpha)\geq (A[\sigma])^\A$.

\begin{theorem}[Adequacy]
\label{p:ltens_adequacy}
The typing rules of $\Ltens$ (\Cref{fig:Ltens:type}) are adequate with respect to the interpretation of terms (contexts,commands) and formulas:
 for all contexts $\Gamma,\Delta$, for all formulas with parameters $A$
 and  for all substitutions $\sigma$ such that $\sigma\Vdash \Gamma$ and $\sigma\Vdash \Delta$, we have:
\begin{enumerate}
\item For any term $t$,    if~ $\Gamma\vdash t:A\mid\Delta$,  then ~  $(t[\sigma])^\A \leq A[\sigma]^\A$;
\item For any context $e$, if~ $\Gamma\mid e:A\vdash \Delta$, then ~ $(e[\sigma])^\A  \geq A[\sigma]^\A$;
\item For any command $c$, if~ $c:(\Gamma\vdash \Delta)$,     then ~  $(c[\sigma])^\A \in \pole$.
\end{enumerate}
\end{theorem}
\begin{proof}
By induction on the typing derivations. 
Since most of the cases are similar to the corresponding cases for the adequacy 
of the embedding of $\Lpar$ into disjunctive structures, we only give some key cases.
\newcommand{\mysubst}[1]{(#1[\sigma])^\A}

\prfcase{($\vdash\tensor$)}
Assume that we have:
$$ \infer[\autorule{\vdash\,\tensor}]{\Gamma\vdash (t_1,t_2): A_1 \tensor A_2 \mid \Delta}
{\Gamma \vdash t_1:A_1 \mid \Delta & \Gamma \vdash t_2 :A_2 \mid \Delta}$$
By induction hypotheses, we have that $\mysubst{t_1}\leq \mysubst{A_1}$ and $\mysubst{t_2}\leq \mysubst{A_2}$. 
Therefore, by monotonicity of the $\tensor$ operator, we have: 
$$\mysubst{(t_1,t_2)}=(t_1[\sigma],t_2[\sigma])^\A=\mysubst{t_1}\tensor\mysubst{t_2}\leq \mysubst{A_1}\parr \mysubst{A_2}\,.$$

\prfcase{($\tensor\,\vdash$)}
Assume that we have:
 $$\infer[\autorule{\tensor\,\vdash}]{\Gamma\mid \mu(x_1,x_2).c:A_1\tensor A_2\vdash \Delta}{c:\Gamma,x_1:A_1,x_2:A_2\vdash \Delta}$$
By induction hypothesis, we get that $(c[\sigma,x_1\mapsto\mysubst{A_1},x_2\mapsto\mysubst{A_2}])^\A\in\pole$. 
Then by definition we have 
$$((\mu(x_1,x_2).c)[\sigma])^\A= 
\msup_{a,b\in\A}\{a\parr b:(c[\sigma,x_1\mapsto a,x_2\mapsto b])^\A\in\pole\} \geq \mysubst{A_1}\tensor \mysubst{A_2}.$$

%

\prfcase{\autorule{\exists\vdash}}
Assume that we have:
$$\infer[\autorule{\exists\,\vdash}]{\Gamma\mid e:\exists X. A\vdash\Delta}{\Gamma\mid e: A\vdash\Delta & X\notin{FV(\Gamma,\Delta)}}$$
By induction hypothesis, we have that for all $a\in\A$, $\mysubst{e}\geq ((A)[\sigma,x\mapsto a])^\A$.
Therefore, we have that $\mysubst{e}\geq \msup_{a\in \A}(A\{X:=a\}[\sigma])^\A$.

\prfcase{\autorule{\vdash\exists}}
Similarly, assume that we have:
$$\infer[\autorule{\vdash\,\exists} ]{\Gamma\vdash t:\exists X. A\mid\Delta}{\Gamma\vdash t: A\{X:=B\}\mid\Delta}$$
By induction hypothesis, we have that $\mysubst{t}\leq (A[\sigma,X\mapsto \mysubst{B}])^\A$.
Therefore, we have that $\mysubst{t}\leq\sup_{b\in\A}\{A\{X:=b\}[\sigma]^\A\} $.
\end{proof}

\setCoqFilename{ImpAlg.TensorAlgebras}
\subsection{Conjunctive algebras}
\label{a:ca:ca}
If we analyse the tensorial calculus underlying {\Ltens} type system
and try to inline all the typing rules involving commands and contexts 
within the one for terms, we are is left with
the following four rules:
$$\infer{\Gamma,A\vdash A}{}
\qquad
\infer{\Gamma \vdash A\tensor B}{\Gamma\vdash A & \Gamma \vdash B}
\qquad
\infer{\Gamma\vdash\neg(A\tensor B)}{\Gamma,A,B\vdash C & \Gamma,A,B\vdash \neg C}
\qquad
\infer{\Gamma\vdash \neg A}{\Gamma,A\vdash C & \Gamma,A\vdash\neg C}
$$

This emphasizes that this positive fragment is a \emph{calculus of contradiction}:
both deduction rules have a negated formula as a conclusion. 
This justifies considering the following deduction rule in the separator:
$$
\infer{¬b ∈ʆ}{¬(a⊗b) ∈ʆ & a∈ʆ }
$$
rather than the modus ponens. The latter can
be retrieved when assuming that the separator also satisfies that
if $¬¬ a∈ʆ$ then $a∈ʆ$. Computationally, this corresponds to the intuition
that when composing values, one essentially gets a computation ($¬¬a$) rather
than a value. Extracting the value from a computation requires some extra 
computational power that is provided by classical control.

\newcommand{\nentails}{\vdash^\neg}
\subsubsection{Internal logic}
\begin{proposition}[Preorder][]
 \label{a:po}
 For any $a,b,c\in\A$, we have:
 \begin{enumerate}
  \coqitem[id_t] $ a\vdash_\sep a$\hfill(Reflexivity)
  \coqitem[C6_t] if $a\vdash_{\sep} b$ and $b\vdash_{\sep} c$ then $a\vdash_{\sep} c$ \hfill(Transitivity)
\end{enumerate}
\end{proposition}
\begin{proof}
 We deduce (2) from its variant defined in terms of $a \nentails b \defeq \neg(a\tensor b) \in \sep$:
  if $a\nentails b$ and $\neg b\nentails c$ then $a\nentails c$.
 As for disjunctive algebras, this is proven by applying twice 
 the deduction rule (3) of separators to prove instead that:
 $$¬ (¬ (¬ b \tensor c) \tensor ¬ (a \tensor b) \tensor a \tensor c) ∈ʆ$$
 This follows directly from the fact that $\ts{4}∈ʆ$.
\end{proof}

\begin{proposition}[Implicative negation]
 For all $a\in \A$, the  following holds:\noem
 \begin{multicols}{4}
  \begin{enumerate}
   \coqitem[dni_t] $a\vdash_\sep \neg\neg a$ 
   \coqitem[dne_t] $\neg \neg a \vdash_\sep a $
   \coqitem[tneg_imp_bot] $\neg a\vdash_\sep a \art \bot$ 
   \coqitem[imp_bot_tneg] $ a\art\bot \vdash_\sep \neg a $
  \end{enumerate}
 \end{multicols}\noem
\end{proposition}
\begin{proof}
 Easy manipulations of the algebras, see Coq proofs. We sketch the two last to give an idea:
 \begin{enumerate}
 \setcounter{enumi}{2}
  \item Observe that $\ts{4}$ implies that 
  if $a \vdash_\sep b$ and $( \neg(c \tensor b)) \in\sep$ then $ (\neg  (c \tensor a))\in\sep$.
  Then the claim follows from the fact that 
  $\neg\neg (a \tensor \neg \bot) \vdash_\sep  a \tensor \neg \bot$ (by 2) and that 
  $(\neg  a) \tensor a \tensor \neg \bot \in\sep $ (using $\ts 2$).
  \item We first prove that $\neg a \vdash_\sep \neg b$ implies $b \vdash_\sep a$
  and  that $\neg \bot = \top$. 
  Then $a \vdash_\sep a \tensor \top$ follows by upward closure from $\ts{1}$.\qedhere
 \end{enumerate}

\end{proof}

For technical reasons,
we define: 
$$a\Diamond b \defeq \msup \{c: a \leq \neg (b \tensor c) \}$$
We first show that:
\begin{lemma}[Adjunction][tadj]
 For any $a,b,c\in\A$, $c \leq a \Diamond b$   iff   $ a \leq \neg (b \tensor c)$.
\end{lemma}
\begin{proof}
 \textbf{($\bm{\Rightarrow}$)} 
 Assume $c \leq a \Diamond b$.
 We use the transitivity to prove that:
 $$ a \leq \neg (b \tensor (a\Diamond b)) \qquad \text 
 \qquad \neg (b \tensor (a\Diamond b)) \leq \neg (b \tensor c)$$
 The right hand side follows directly from the assumption, while the
 left hand side is a consequences of distribution laws:
 $$
 \neg (b \tensor (a\Diamond b))
 = \neg (b \tensor \msup \{c: a \leq \neg (b \tensor c) \})
 = \minf \{\neg (b \tensor c) : a \leq \neg (b \tensor c)\}\geq a$$
 
 \textbf{($\bm{\Leftarrow}$)} 
 Trivial by definition of $a\Diamond b$.

\end{proof}

\begin{proposition}[][app_closed]
 If $a\in\sep$ and $b\in\sep$ then $ab\in\sep$.
\end{proposition}
\begin{proof}
 First, observe that we have:
 $$ab = \minf\{\neg\neg c: a \leq b\art c\} 
      = \neg \msup \{\neg c: a \leq b\art c\}$$
 Then one can easily show that:
 $$\neg (a\Diamond b) \leq ab $$
 and therefore it suffices to show that
 $\neg (a\Diamond b)\in\sep$.
 To that end, we use the deduction rule of the separator and show that 
 $$\neg (b \tensor (a\Diamond b))\in\sep$$
 This is now easy using the previous lemma and that $a\in\sep$, since we have:
 $$a \leq \neg (b \tensor (a\Diamond b)) 
 \text{~~iff~~} a\Diamond b \leq a\Diamond b\eqno\qedhere$$
 
\end{proof}

The beta reduction rule now involves a double-negation 
on the reduced term:
\begin{proposition}[][beta_reduction]
 $(\lambda f)a \leq \neg\neg f(a)$
\end{proposition}
\begin{myproof}
 We first show that $t\leq a\art b$ implies $ta \leq \neg\neg b$,
 and we use that $\lambda f \leq a \art f(a)$ to conclude.
  \end{myproof}

We show that Hilbert's combinators $\comb{k}$ and $\comb{s}$
belong to any conjunctive separator:
\begin{proposition}[$\comb{k}$ and $\comb{s}$][tsep_K]
We have:\\[0.5em]
\begin{minipage}{0.43\textwidth}
 \begin{enumerate}
 \coqitem[tsep_K] $(\lambda x y.x)^\A \in\sep$
 \end{enumerate}
\end{minipage}
\begin{minipage}{0.55\textwidth}
 \begin{enumerate}\setcounter{enumi}{1}
 \coqitem[tsep_S] $(\lambda x y z.x\,z\,(y\,z))^\A \in\sep$
\end{enumerate}
\end{minipage}
\end{proposition}
\begin{proof}
We only sketch these proofs are quite involved and require several auxiliary results (see the Coq files). 
\begin{enumerate}
 \item We first show $(\lambda x y.x)^\A=  \minf_{a,b}(¬ (a ⊗ b ⊗ ¬ b))$.
 Then we conclude by transitivity by showing that $\minf_{a,b}(¬ (a ⊗ ¬b ⊗  b))∈ʆ$
 and $\minf_{a,b}¬ ((¬b ⊗  b)⊗(b⊗¬ b))∈ʆ$. The latter follows from $\ts{3}$, while
 the former follows from $\ts{2}$ modulo the facts that the tensor is
 associative and commutative with respect to the separator.
 \item We first show that the interpretation of the type of $\comb{s}$ belongs to $\sep$:
 $$\minf_{a,b,c} ]((a \art b  \art c ) \art (a \art b ) \art a \art c ) ∈ʆ\eqno \coqlink[tsep_S_type]{(S1)}$$
 We then mimic the proof that $\comb{s}$ and its type are the same in implicative structure, 
 replacing the ordering relation by the entailment. We begin by showing (using distribution laws) 
 that $\comb{s}∈ʆ$ can be deduced from:
 $$\minf \{¬ (((a  ⊗ b ) ⊗ c ) ⊗ ¬ d): a c \leq b c \art d\} ∈ʆ\eqno \coqlink[tsep_S_true]{(S2)}$$
 Then, after showing that:
 $$\minf_{a,b} (a \art (b \art a b ))∈ʆ\eqno \coqlink[app_entails_inf]{(S3)}$$
 we prove that the $(S3)$ can be deduced from:
 $$ \minf_{b,c,d} ((c \art b c \art d) \art (c \art b c) \art c \art d) ∈ʆ \eqno\coqlink[preS_true]{(S4)}$$
 which follows from $(S1)$.

\end{enumerate}

\end{proof}

%

In the case where the separator is classical\footnote{Actually,
since we always have that if $\neg\neg\neg\neg a\in\sep$ then $\neg\neg a\in\sep$,
the same proof shows that in the intuitionistic 
case we have at $\neg\neg t^\A\in\sep$.}, we can prove
that it contains the interpretation 
of all closed $\lambda$-terms:
\begin{theorem}[$\lambda$-calculus][] \label{thm:lambda}
 If $\sep$ is classical and $t$ is a closed $\lambda$-term,
 then $t^\A\in\sep$.
\end{theorem}
\begin{proof}
 By combinatorial completeness, we have the existence of a
 combinatorial term $t_0$ (\emph{i.e.} a composition of $\comb{k}$ and $\comb{s}$)
 such that $t_0 \to^* t$. Since $\comb{k}\in\sep$, $\comb{s}\in\sep$ and $\sep$ is closed under application,
 $t_0^\A\in\sep$.
  For each  step $t_n \to  t_{n+1}$, 
  we have $t_n^\A \leq \neg\neg t_{n+1}^\A$, and thus $t_n^\A\in\sep$ implies $t_{n+1}^\A\in\sep$.
  We conclude by induction on the length of the reduction $t_0 \to^* t$.
\end{proof}

\subsubsection*{Induced Heyting algebra}
As in the implicative case, 
the entailment relation induces a structure of (pre)-Heyting algebra,
whose conjunction and disjunction are naturally given
$ a+b\defeq \neg(\neg a \tensor \neg b)$ and $a\times b \defeq a\tensor b$.

\begin{proposition}[Heyting Algebra][]  For any $a,b,c\in\A$ \label{a:ta_heyting}
 For any $a,b,c\in\A$, we have:\noem
 \begin{multicols}{3}
  \begin{enumerate}
    \coqitem[Heyting_and_l] $a\times b \vdash_\sep a$
    \coqitem[Heyting_and_r] $a\times b \vdash_\sep b$
    \coqitem[Heyting_or_l] $a \vdash_\sep a+b$ 
    \coqitem[Heyting_or_r] $b \vdash_\sep a+b$
    \coqitem[Heyting_adj] $a \vdash_{\sep}  b \art c ~~ \text{iff } ~~ a\times b \vdash_{\sep} c$
  \end{enumerate}

 \end{multicols}
\end{proposition}

\begin{proof}
 Easy manipulation of conjunctive algebras, see the Coq proofs.
\end{proof}

\subsubsection*{Conjunctive tripos}
We will need the following lemma:
\begin{lemma}[Adjunction][demiadj_app]
 If $a\leq b \art c$ then $ab\leq \neg\neg c$.
\end{lemma}

In order to obtain a conjunctive tripos, we define:
$$\hugeex_{i\in I} a_i ~~\defeq ~~ \msup_{i\in I} a_i 
\qquad\qquad
\qquad\qquad
\hugefa_{i\in I} a_i ~~\defeq ~~ \neg (\msup_{i\in I} \neg a_i)
$$

\begin{theorem}[Conjunctive tripos][]\label{a:ta:tripos}
 Let $(\A,\leq,\to,\sep)$ be a classical
conjunctive algebra. The following functor 
 (where $f:J\to I$):
 $$
 \T :  I\mapsto \A^I/\usep
\qquad
\T(f) :\left\{\begin{array}{lcl}
		 \A^I/\usep &\to& \A^J/\sep[J]\\ [0.5em]
		 \left[(a_i)_{i\in I}\right]& \mapsto &[(a_{f(j)})_{j\in J}]
		\end{array}
		\right.
	$$
 defines a tripos.  
\end{theorem}

\begin{proof}
 The proof mimics the proof in the case of implicative algebras,
  see \Cref{a:ta:tripos}.
 
We verify that $\T$ satisfies all the necessary conditions to be a tripos.
    \begin{itemize}
    \item The functoriality of $\T$ is clear.
    \item For each $I\in\Set$, the image of the corresponding diagonal morphism $\T(\delta_I)$ associates to any element $[(a_{ij})_{(i,j)\in I\times I}]\in\T(I\times I)$ 
    the element  $[(a_{ii})_{i\in I}]\in\T(I)$.
    We define
    $$(=_{I})~:~ i,j \mapsto \begin{cases}\minf_{a\in\A}(a\to a)  & \text{if } i=j\\ \bot \to \top & \text{if } i\neq j\\ \end{cases}$$
    and we need to prove that for all $[a]\in\T(I\times I)$:
    $$     [\top]_{I} \uleq \T(\delta_I)(a) \qquad\Leftrightarrow  \qquad[=_{I}] \oldleq_{S[I\times I]} [a]$$
    Let then $[(a_{ij})_{i,j\in I}]$ be an element of $\T(I\times I)$.

    From left to right, assume that $[\top]_{I} \uleq \T(\delta_I)(a)$, that is to say that 
    there exists $s\in\sep$ such that for any $i\in I$, $s\leq \top \to a_{ii}$.
    We would like to reproduce the proof in the implicative case, which uses $\lambda z. z(s(\lambda x.x))$.
    Here, due to the double-negation induced by the application (see \Cref{s:conj_lambda}), 
    we can only show \turl{tripos_eq1}{that}:
    $$\lambda z. z(s(\lambda x.x)) \leq i =_I j \to \neg\neg\neg\neg a_{ij}\eqno (\forall i,j)$$.
    Indeed, if $i\neq j$, we have that: 
    $$\begin{array}{ll}
                &\lambda z. z(s(\lambda x.x)) \leq (\top \to \bot) \to \neg\neg\neg\neg a_{ij}\\
     \Leftarrow \qquad & (\top\to\bot) \to (\top\to\bot)(s(\lambda x.x)) \leq (\top \to \bot) \to \neg\neg\neg\neg a_{ij}\\
     \Leftarrow \qquad & (\top\to\bot)(s(\lambda x.x)) \leq \neg\neg\neg\neg a_{ij}\\
     \Leftarrow \qquad & \top\to\bot\leq (s(\lambda x.x)) \to \neg\neg a_{ij}
     \end{array}\eqno\begin{array}{r}\\ (\lambda-\text{def})\\(\text{variance})\\(\text{adjunction})\\
    \end{array}
    $$
    the last one being true by subtyping.
    If $i= j$, we have that:
    $$\begin{array}{ll}
                   &\lambda z. z(s(\lambda x.x)) \leq (\minf a\to a) \to \neg\neg\neg\neg a_{ii} \\
     \Leftarrow \qquad & (\minf a\to a) \to (\minf a\to a)(s(\lambda x.x)) \leq (\minf a\to a) \to \neg\neg\neg\neg a_{ij}\\
     \Leftarrow \qquad & (\minf a\to a)(s(\lambda x.x)) \leq \neg\neg\neg\neg a_{ij}   \\
     \Leftarrow \qquad & (\minf a\to a)\leq (s(\lambda x.x)) \to \neg\neg a_{ij}         \\
     \Leftarrow \qquad & (s(\lambda x.x))\to (s(\lambda x.x))\leq (s(\lambda x.x)) \to \neg\neg a_{ij} \\
     \Leftarrow \qquad & s(\lambda x.x)\leq  \neg\neg a_{ij}                                            \\
     \Leftarrow \qquad & s\leq \lambda x.x\to a_{ij}                              \\
     \end{array}
     \eqno
     \begin{array}{r}\\ (\lambda-\text{def})\\(\text{variance})\\(\text{adjunction})\\(\text{variance})\\(\text{adjunction})\\
    \end{array}
    $$
    the last one being true by assumption. We conclude using the fact that any $\lambda$-terms with parameters in $\sep$ belongs to $\sep$
    using a slight variant of \Cref{thm:lambda}.

    \turl{tripos_eq2}{From right to left}, if there exists $s\in\sep$ such that for any $i,j\in I$, $s\leq  i =_I j \to a_{ij}$,
    then in particular for all $i$ $\minf_a( a \art a) \vdash a_{ii}$.
    We use the transitivy of $\vdash$ to show that $\top \vdash a_{ii}$ follows from $\minf_a( a \art a) \vdash a_{ii}$ and
    $\top \vdash \minf_a( a \art a)$. Writing $Id$ for $\minf_a( a \art a)$, the latter is obtained by using the deduction rule:
    $$\infer{\vdash \neg (\top \tensor Id}{\vdash \neg (Id\tensor(\top \tensor \neg Id)) & \vdash Id}$$
    Using $\ts5$ to get $\neg (Id\tensor(\top \tensor \neg Id))$ from $\neg(\neg(Id\tensor \top) \tensor \neg Id)$ which follows from $\ts2$.
     
     \item For each projection $\pi^1_{I\times J}:I\times J \to I$ in $\C$, the monotone function 
 $\T(\pi^1_{I,J}):\T(I) \to \T(I\times J)$ has both a left adjoint $(\exists J)_I$ and a right adjoint $(\forall J)_I$ which are defined by:
 
 $$(\forall J)_I \big(\left[(a_{ij})_{i,j\in I\times J}\right]\big) \defeq \big[(\bigfa_{j\in J} a_{ij})_{i\in I} \big]\qquad\qquad     
  (\exists J)_I \big(\left[(a_{ij})_{i,j\in I\times J}\right]\big) \defeq \big[(\bigex_{j\in J} a_{ij})_{i\in I} \big]$$

  We only give the case of $\forall$, the case for $\exists$ is easier (it corresponds to \turl{tripos_ex1}{this} and \turl{tripos_ex2}{this} Coq lemmas).
  We need to show that for any $[(b_{ij})_{(i,j)\in I\times J}]\in\T(I\times J)$ and for any  $[(a_{i})_{i\in I}]$, we have:
      $$ [(a_{i})_{(i,j)\in I\times J}] \oldleq_{\sep[I\times J]} [(b_{ij})_{(i,j)\in I}]
      \quad \Leftrightarrow \quad
        [(a_i)_{i\in I}] \oldleq_{\sep[I]} \big[(\bigfa_{j\in J} b_{ij})_{i\in I} \big]=\big[(\neg \msup_{j\in J} \neg b_{ij})_{i\in I} \big]$$
  Let us fix some $[a]$ and $[b]$ as above.
  
  From \turl{tripos_fa1}{left to right}, assume that for all $i\in I$, $j\in J$, 
  $a_{ij} \vdash b_i$, we want to prove that $\forall i\in I$, we have  $a_i \vdash \neg \msup_{j\in J} \neg b_{ij}$.
  We first show that for any $a,b,c$, the following rule is \turl{S4d_n}{valid} (it mainly amount to $\ts4$):
  $$\infer{\neg(c\tensor a)\in\sep}{a\vdash b & \neg(c\tensor b)\in\sep}$$
  and prove instead that $\neg\neg \msup_{j\in J} \neg b_{ij} \vdash \msup_{j\in J} \neg b_{ij}$ 
  and $\neg(a_i\tensor \msup_{j\in J} \neg b_{ij})\in\sep$.
  The former amount to $\neg\neg a \vdash a$ while we can use commutation rule on the latter to rewrite it as:
  $\minf_{j\in J}\neg(a_i\tensor \neg b_{ij})\in\sep$ which follows from the assumption.

  From \turl{tripos_fa2}{right to left}, the processus is almost the same and relies on the fact 
  we have for all $i\in I$, $j\in J$, $(a_i \art b_ij) \in\sep$ if and only if
  for all $i\in I$, $\minf_j (a_i \art b_ij) \in\sep$ if and only if 
  for all $i\in I$, $(a_i \tensor \msup_j \neg b_ij) \in\sep$. We then use the same lemma with the reverse law $a\vdash \neg\neg a$.

  \item These adjoints clearly satisfy the Beck-Chevaley condition as in the implicative cases.
  \item Finally, we define $\Prop \defeq \A$ and verify that $\trth\defeq [\id_\A] \in \T(\Prop)$ is a generic predicate,
  as in the implicative case.
      \end{itemize}

\end{proof}

\newpage
\setCoqFilename{ImpAlg.Duality}
\section{The duality of computation}
\label{a:du}
\begin{proposition}
Let $(\A,\leq,\parr,\neg)$ be a disjunctive structure.
Let us define:
$$\begin{array}{l@{\qquad\quad}l@{\qquad\quad}l}
 \text{\textbullet}~\A^\tensor \defeq \A^\parr   & \text{\textbullet}~\minf^{\!\!\tensor}\defeq \msup^\parr  & \text{\textbullet}~a \tensor b \defeq a \parr b\qquad\qquad \\
 \text{\textbullet}~a \revord b \defeq b \leq a  & \text{\textbullet}~\msup^\tensor\defeq \minf^{\!\!\parr}  & \text{\textbullet}~\neg a ~~\defeq~ \neg a
\end{array} 
 $$
then $(\A^\tensor,\revord,\tensor,\neg)$ is a conjunctive structure.
\end{proposition}
\begin{proof}[\turl{PS_TS}{Proof}]
We check that for all $a, a', b, b'\in\A$ and for all subsets $A\subseteq \A$, we have:
 \begin{itemize}
\item[\turl{rev_tneg_mon}{1.}~]   If $a \revord  a'$ then $\neg a' \revord  \neg a$ \hfill(Variance) 
\item[\turl{rev_tensor_mon}{2.}~] If $a \revord  a'$ and $b \revord b'$ then $a \tensor b \revord  a' \tensor b'$.\hfill (Variance) 
\item[\turl{rev_tensor_join_l}{3.}~]$(\minf^{\!\!\tensor}_{a\in A} a) \tensor b = \minf^{\!\!\tensor}_{a\in A} (a \tensor b)$ and 
$b \tensor (\minf^{\!\!\tensor}_{a\in A} a) = \minf^{\!\!\tensor}_{a\in A} (b \tensor a)$\hfill (Distributivity) 
\item[\turl{rev_tneg_join}{4.}~]$\neg (\msup^\tensor_{a\in A} a) = \minf^{\!\!\tensor}_{a\in A} (\neg a) $\hfill (Commutation)
\end{itemize}
All the proof are trivial from the corresponding properties of disjunctive structures. 
\end{proof}

\begin{proposition}
Let $(\A,\leq,\tensor,\neg)$ be a conjunctive structure.
Let us define:
$$\begin{array}{l@{\qquad\quad}l@{\qquad\quad}l}
 \text{\textbullet}~~\A^\parr \defeq \A^\tensor   
 & \text{\textbullet}~~\minf^{\!\!\parr}\defeq \msup^\tensor  
 & \text{\textbullet}~~a \parr b \defeq a \tensor b\qquad\qquad \\
 \text{\textbullet}~~a \revord b \defeq b \leq a  
 & \text{\textbullet}~~ \msup^\parr\defeq \minf^{\!\!\tensor} 
 & \text{\textbullet}~~\neg a ~~\defeq~ \neg a
\end{array} 
$$
then $(\A^\tensor,\revord,\tensor,\neg)$ is a conjunctive structure.
\end{proposition}
\begin{proof}[\turl{TS_PS}{Proof}]
We check that for all $a, a', b, b'\in\A$ and for all subsets $A\subseteq \A$, we have:
 \begin{itemize}
\item[\turl{rev_pneg_mon}{1.}~] If $a \revord  a'$ then $\neg a' \revord  \neg a$. \hfill(Variance) 
\item[\turl{rev_parr_mon}{2.}~] If $a \revord  a'$ and $b \revord b'$ then $a \parr b \revord  a' \parr b'$.\hfill (Variance) 
\item[\turl{rev_parr_join_l}{3.}~]$ (\minf^{\!\!\parr}_{a\in A} a) \parr b = \minf^{\!\!\parr}_{a\in A} (a \parr b)$ ~~and~~
$a \parr (\minf^{\!\!\parr}_{b\in B} b)  = \minf^{\!\!\parr}_{b\in B} (a \parr b)            $\hfill (Distributivity) 
\item[\turl{rev_pneg_meet}{4.}~]$\neg (\minf^{\!\!\parr} A) = \msup^\parr_{a\in A} (\neg a)$\hfill (Commutation)
\end{itemize}
All the proof are trivial from the corresponding properties of conjunctive structures. 
\end{proof}

\begin{theorem}
\label{a:p:ta_pa}
 Let $(\A^\tensor,\tsep)$ be a conjunctive algebra, the set:
 $$\psep ~\defeq ~\neg^{-1}(\tsep) = \{a\in\A : \neg a\in \tsep\}$$
 is a valid separator for the dual disjunctive structure $\A^\parr$. 
 \end{theorem}

 \begin{theorem}
\label{a:p:pa_kta}
 Let $(\A^\parr,\psep)$ be a disjunctive algebra.
 The set:
 $$\tsep ~\defeq ~\neg^{-1}(\psep) = \{a\in\A : \neg a\in \psep\}$$
 is a classical separator for the dual conjunctive structure $\A^\tens$. 
\end{theorem}
\begin{proof}
\trurl{PA_KTA}{Both} \trurl{TA_PA}{proofs} rely on the fact that:
$$ a \vdash_{\tsep} b \Leftrightarrow \neg a \vdash_{\psep} \neg b
~~~~\text{and}~~~~
a \vdash_{\psep} b \Leftrightarrow \neg a \vdash_{\tsep} \neg b$$
In particular, to prove that the modus ponens is valid when passing from 
$\A^\tensor$ to $\A^\parr$, we need to show that if
$a,a\to b\in\neg^{-1}(\tsep)$, then $b\in\neg^{-1}(\tsep)$ \emph{i.e.}
$\neg b\in\tsep$.
By hypothesis, we thus have 
that $\neg a \to \neg b \in\tsep$, from which we deduce that $\neg(b\tensor \neg a)\in\tsep$
(by internal contraposition). Using the deduction axiom (since $\neg a \in \tsep$), we finally get $\neg b \in\tsep$.
\end{proof}

}
\end{document}